\documentclass[aos,preprint]{imsart}

\RequirePackage[OT1]{fontenc}
\RequirePackage{amsthm,amsmath}
\RequirePackage[numbers]{natbib}
\RequirePackage[colorlinks,citecolor=blue,urlcolor=blue]{hyperref}

\arxiv{arXiv:0000.0000}
\usepackage{array}
\usepackage{bm}
\usepackage{caption}
\usepackage{enumerate}
\usepackage{epsfig}
\usepackage{fontenc}
\usepackage{graphicx,times,verbatim,amsmath,amssymb,amsthm}
\usepackage{mathrsfs}
\usepackage{multirow}
\usepackage{natbib}
\usepackage{subcaption}

\usepackage[]{algorithm2e}
\RestyleAlgo{boxruled}

\usepackage{tikz}
\usetikzlibrary{arrows,matrix,positioning,calc}

\DeclareGraphicsExtensions{.eps}

\newcommand{\vertiii}[1]{{\left\vert\kern-0.25ex\left\vert\kern-0.25ex\left\vert #1
    \right\vert\kern-0.25ex\right\vert\kern-0.25ex\right\vert}}

\newcommand{\thickhline}{%
    \noalign {\ifnum 0=`}\fi \hrule height 1pt
    \futurelet \reserved@a \@xhline
}

\newtheorem{theorem}{Theorem}[section]
\newtheorem{lemma}[theorem]{Lemma}
\newtheorem{corollary}[theorem]{Corollary}
\newtheorem{proposition}[theorem]{Proposition}

\newtheorem{remark}[theorem]{Remark}

\newtheorem{example}[theorem]{Example}

\newtheorem*{remark*}{Remark}

\startlocaldefs
\numberwithin{equation}{section}
\theoremstyle{plain}

\endlocaldefs

\begin{document}

\begin{frontmatter}

\title{A Statistically and Numerically Efficient Independence Test based on Random Projections and Distance Covariance}
\runtitle{Random Projected Distance Covariance}


\begin{aug}
\author{\fnms{Cheng} \snm{Huang}\thanksref{t2,m1}\ead[label=e1]{c.huang@gatech.edu}},
\and
\author{\fnms{Xiaoming} \snm{Huo}\thanksref{t2,m1}\ead[label=e2]{huo@gatech.edu}}

\thankstext{t1}{Some comment}
\thankstext{t2}{First supporter of the project}
\runauthor{C.Huang \& X.Huo}

\affiliation{Georgia Institute of Technology\thanksmark{m1}}

\address{Georgia Institute of Technology \\
School of Industrial and Systems Engineering\\
765 Ferst Drive, Atlanta, Ga 30332-0205 \\
\printead{e1}\\
\phantom{E-mail:\ }\printead*{e2}}
\end{aug}

\begin{abstract}
Test of independence plays a fundamental role in many statistical techniques.
Among the nonparametric approaches, the distance-based methods (such as the distance correlation based hypotheses testing for independence) have numerous advantages, comparing with many other alternatives.
A known limitation of the distance-based method is that its computational complexity can be high.
In general, when the sample size is $n$, the order of computational complexity of a distance-based method, which typically requires computing of all pairwise distances, can be $O(n^2)$.
Recent advances have discovered that in the {\it univariate} cases, a fast method with $O(n \log n)$ computational complexity and $O(n)$ memory requirement exists.
In this paper, we introduces a test of independence method based on random projection and distance correlation, which achieves nearly the same power as the state-of-the-art distance-based approach, works in the {\it multivariate} cases, and enjoys the $O(n K \log n)$ computational complexity and $O(\max\{n,K\})$ memory requirement, where $K$ is the number of random projections.
Note that saving is achieved when $K < n/\log n$.
We name our method a Randomly Projected Distance Covariance (RPDC).
The statistical theoretical analysis takes advantage of some techniques on random projection which are rooted in contemporary machine learning.
Numerical experiments demonstrate the efficiency of the proposed method, in relative to several competitors.
\end{abstract}



\end{frontmatter}

\section{Introduction}
\label{sec:intro}

Test of independence is a fundamental problem in statistics, with many existing work including the maximal information coefficient (MIC) \citep{reshef2011detecting}, the copula based measures
\citep{schweizer1981nonparametric,siburg2010measure}, the kernel based criterion \citep{gretton2005measuring} and the distance correlation \citep{szekely2007measuring,szekely2009brownian}, which motivated our current work.
Note that the above works as well as ours focus on the detection of the presence of the independence, which can be formulated as statistical hypotheses testing problems.
On the other hand, interesting developments (e.g., \cite{reimherr2013quantifying}) aim at a more general framework for interpretable statistical dependence, which is not the goal of this paper.

Distance correlation proposed by \citep{szekely2007measuring} is an indispensable method in test of independence.
The direct implementation of distance correlation takes $O(n^2)$ time, where $n$ is the sample size.
The time cost of distance correlation could be substantial when sample size is just a few thousands.
When the random variables are univariate, there exist efficient numerical algorithms of time complexity $O(n \log n)$ \citep{huo2015fast}.
However, for the multivariate random variables, we have not found any efficient algorithms in existing papers after an extensive literature survey.

Independence tests of multivariate random variables could have a wide range of applications.
In many problem settings, as metioned in \citep{taskinen2005multivariate}, each experimental unit will be measured multiple times, resulting in multivariate data.
Researchers are often interested in exploring potential relationships among subsets of these measurements.
For example, some measurements may represent attributes of physical characteristics while others represent attributes of psychological characteristics.
It may be of interests to determine whether there exists a relationship between the physical and the psychological characteristics.
A test of independence between pairs of vectors, where the vectors may have different dimensions and scales, becomes crucial.
Moreover, the number of experimental units, or equivalently, sample size, could be massive, which requires the test to be computationally efficient.
This work will meet the demands for numerically efficient independence tests of multivariate random variables.

The newly proposed test of independence between two (potentially multivariate) random variable $X$ and $Y$ works as follows.
Firstly, both $X$ and $Y$ are randomly projected to one-dimensional spaces.
Then the fast computing method for distance covariances between a pair of {\it univariate} random variables is adopted to compute for an surrogate distance covariance.
The above two steps are repeated for numerous times.
The final estimate of the distance covariance is the average of all aforementioned surrogate distance covariances.

For numerical efficiency, we will show (in Theorem \ref{th:fast-multivariate}) that the newly proposed algorithm enjoys the $O(Kn \log n)$ computational complexity and $O(\max\{n,K\})$ memory requirement, where $K$ is the number of random projections and $n$ is the sample size.
On the statistical efficiency, we will show (in Theorem \ref{th:independence_test}) that the asymptotic power of the test of independence by utilizing the newly proposed statistics is as efficient as its original multivariate counterpart, which achieves the stat-of-the-art rates.

The rest of this paper is organized as follows.
In Section \ref{sec:review}, we review the definition of distance covariance, its fast algorithm in univariate cases and related distance-based independence tests.
Section \ref{sec:methods} gives the detailed algorithm for distance covariance of random vectors and corresponding independence tests.
In Section \ref{sec:theory}, we present some theoretical properties on distance covariance and the asymptotic distribution of the proposed estimator.
In Section \ref{sec:simulations}, we conduct numerical examples to compare our method against others in existing literature.
Some discussions are presented in Section \ref{sec:discuss}.
We conclude in Section \ref{sec:conclude}.
All technical proofs as well as formal presentation of algorithms are relegated to the appendix when appropriate.

Throughout this paper, we adopt the following notations.
We denote $c_p=\frac{\pi^{(p+1)/2}}{\Gamma((p+1)/2)}$ and $c_q=\frac{\pi^{(q+1)/2}}{\Gamma((q+1)/2)}$ as two constants, where $\Gamma(\cdot)$ denotes the Gamma function.
We will also need the following constants: $C_p = \frac{c_1c_{p-1}}{c_p} = \frac{\sqrt{\pi}\Gamma((p+1)/2)}{\Gamma(p/2)}$ and $C_q = \frac{c_1c_{q-1}}{c_q} = \frac{\sqrt{\pi}\Gamma((q+1)/2)}{\Gamma(q/2)}$.
For any vector $v$, let $v^t$ denote its transpose.

\section{Review of Distance Covariance: Definition, Fast Algorithm, and Related Independence Tests}
\label{sec:review}

In this section, we review some related existing works.
In Section \ref{sec:distCov}, we recall the concept of distance variances and correlations, as well as some of their properties.
In Section \ref{sec:sample-version}, we discuss the estimators of distance covariances and correlations, as well as their computation.
We present their applications in testing of independence in Section \ref{sec:test-of-independ}.

\subsection{Definition of Distance Covariances}
\label{sec:distCov}

Measuring and testing the dependency between two random variables is a fundamental problem in statistics.
The classical Pearson's correlation coefficient can be inaccurate and even misleading when nonlinear dependency exists.
Paper \cite{szekely2007measuring} proposes the novel measure--distance correlation--which is exactly zero if and only if two random variables are independent.
A limitation is that if the distance correlation is implemented based on its original definition, the corresponding computational complexity can be as high as $O(n^2)$, which is not desirable when $n$ is large.

We review the definition of the distance correlation in \cite{szekely2007measuring}.
Let us consider two random variables $X \in \mathbb{R}^p$, $Y \in \mathbb{R}^q, p\ge 1, q\ge 1$.
Let the complex-valued functions $\phi_{X,Y}(\cdot)$, $\phi_{X}(\cdot)$, and $\phi_{Y}(\cdot)$ be the characteristic functions of the joint density of $X$ and $Y$, the density of $X$, and the density of $Y$, respectively.
For any function $\phi$, we denote $|\phi|^2 = \phi\bar{\phi}$, where $\bar{\phi}$ is the conjugate of $\phi$; in words, $|\phi|$ is the magnitude of $\phi$ at a particular point.
For vectors, let us use $|\cdot|$ to denote the Euclidean norm.
In \cite{szekely2007measuring}, the definition of distance covariance between random variables $X$ and $Y$ is
\begin{equation}
\mathcal{V}^2(X,Y) = \int_{\mathbb{R}^{p+q}} \frac{|\phi_{X,Y}(t,s) - \phi_X(t) \phi_Y(s)|^2}{c_pc_q|t|^{p+1}|s|^{q+1}} dt ds,
\label{eq:def_1}
\end{equation}
where two constants $c_p$ and $c_q$ have been defined at the end of the Section \ref{sec:intro}.
The distance correlation is defined as
\begin{equation*}
\mathcal{R}^2(X,Y) = \frac{\mathcal{V}^2(X,Y)}{\sqrt{\mathcal{V}^2(X,X)}\sqrt{\mathcal{V}^2(Y,Y)}}.
\label{eq:def_2}
\end{equation*}

\noindent
The following property has been established in the aforementioned paper.
\begin{theorem}
\label{th:dist-cov}
Suppose $X \in \mathbb{R}^p, p \ge 1$ and $Y \in \mathbb{R}^q, q\ge 1$ are two random variables, the following statements are equivalent:
\begin{enumerate}
\item $X \mbox{ is independent of } Y$;

\item $\phi_{X,Y}(t,s) = \phi_X(t) \phi_Y(s)$, for any $t \in \mathbb{R}^p$ and $s \in \mathbb{R}^q$;

\item $\mathcal{V}^2(X,Y) = 0$;

\item $\mathcal{R}^2(X,Y)=0$.
\end{enumerate}
\end{theorem}

Given sample $(X_1,Y_1),\ldots,(X_n,Y_n)$, we can estimate the distance covariance by replacing the population characteristic function with the sample characteristic function: for $i=\sqrt{-1}, t\in \mathbb{R}^p, s\in \mathbb{R}^q$, we define
\begin{eqnarray*}
\hat{\phi}_X(t) &=& \frac{1}{n} \sum_{j=1}^n e^{iX_j^t t}, \\
\hat{\phi}_Y(s) &=& \frac{1}{n} \sum_{j=1}^n e^{iY_j^t s}, \mbox{ and }\\
\hat{\phi}_{X,Y}(t,s) &=& \frac{1}{n} \sum_{j=1}^n e^{iX_j^t t + iY_j^t s}.
\end{eqnarray*}
Consequently one can have the following estimator for $\mathcal{V}^2(X,Y)$:
\begin{equation}
\mathcal{V}_n^2(X,Y) = \int_{\mathbb{R}^{p+q}} \frac{|\hat{\phi}_{X,Y}(t,s) - \hat{\phi}_X(t) \hat{\phi}_Y(s)|^2}{c_pc_q|t|^{p+1}|s|^{q+1}} dt \cdot ds.
\label{eq:def_3}
\end{equation}
Note that the above formula is convenient to define a quantity, however is {\it not} convenient for computation, due to the integration on the right hand side.
In the literature, other estimates have been introduced and will be presented in the following.

\subsection{Fast Algorithm in the Univariate Cases}
\label{sec:sample-version}

The paper \cite{lyons2013distance} gives an equivalent definition for the distance covariance between random variables $X$ and $Y$:
\begin{multline}
\mathcal{V}^2(X,Y) = \mathbb{E}[ d(X,X') d(Y,Y') ] = \mathbb{E}[|X-X'||Y-Y'|] \\
- 2\mathbb{E}[|X - X'||Y-Y''|] + \mathbb{E}[|X - X'|] \mathbb{E}[|Y-Y'|],
\label{eq:def_4}
\end{multline}
where the double centered distance $d(\cdot,\cdot)$ is defined as
$$
d(X,X') = |X-X'| - \mathbb{E}_{X}[|X-X'|]  - \mathbb{E}_{X'}[|X-X'|] + \mathbb{E}[|X-X'|],
$$
where $\mathbb{E}_{X}$, $\mathbb{E}_{X'}$ and $\mathbb{E}$ are expectations over $X$, $X'$ and $(X,X')$, respectively.

Motivated by the above definition, one can give an unbiased estimator for $\mathcal{V}^2(X,Y)$.
The following notations will be utilized: for $1\le i,j \le n, $
\begin{align}
&a_{ij} = |X_i - X_j|, \hspace{0.2in} b_{ij} = |Y_i - Y_j|, \nonumber \\
&a_{i\cdot} = \sum_{l=1}^n a_{il}, \hspace{0.2in} b_{i\cdot } = \sum_{l=1}^n b_{il}, \label{notation:1} \\
&a_{\cdot \cdot} = \sum_{k,l=1}^n a_{kl}, \hspace{0.1in} \text{and} \hspace{0.1in} b_{\cdot \cdot} = \sum_{k,l=1}^n b_{kl}. \nonumber
\end{align}
It has been proven \citep{szekely2014partial,huo2015fast} that
\begin{multline}
\label{eq:def_5}
\Omega_n(X,Y) = \frac{1}{n(n-3)} \sum_{i \neq j} a_{ij} b_{ij} \\ - \frac{2}{n(n-2)(n-3)} \sum_{i=1}^n a_{i \cdot} b_{i \cdot} + \frac{a_{\cdot \cdot} b_{\cdot \cdot}}{n(n-1)(n-2)(n-3)}
\end{multline}
is an unbiased estimator of $\mathcal{V}^2(X,Y)$.
In addition, a fast algorithm has been propose \citep{huo2015fast} for the aforementioned sample distance covariance in the univariate cases with complexity order $O(n \log n)$ and storage $O(n)$.
We list the result below for reference purpose.
\begin{theorem}[Theorem 3.2 \& Corollary 4.1 in \cite{huo2015fast}]
\label{th:fast-univariate}
Suppose $X_1,\ldots,X_n$ and $Y_1,\ldots,Y_n \in \mathbb{R}$.
The unbiased estimator $\Omega_n$ defined in (\ref{eq:def_5}) can be computed by an $O(n \log n)$ algorithm.
\end{theorem}
In addition, as a byproduct, the following result is established in the same paper.
\begin{corollary}
The quantity
$$
\frac{a_{\cdot \cdot} b_{\cdot \cdot}}{n(n-1)(n-2)(n-3)} = \frac{\sum_{k,l=1}^n a_{kl} \sum_{k,l=1}^n b_{kl}}{n(n-1)(n-2)(n-3)}
$$
can be computed by an $O(n \log n)$ algorithm.
\end{corollary}
We will use the above result in our test of independence.
However, as far as we know, in the multivariate cases, there does not exist any work on fast algorithm of the order of complexity $O(n \log n)$.
This paper will fill in this gap by introducing an order $O(n K \log n)$ complexity algorithm in the multivariate cases.

\subsection{Distance Based Independence Tests}
\label{sec:test-of-independ}

In \cite{szekely2007measuring} an independence test is proposed using the distance covariance.
We summarizes it below as a theorem, which serves as a benchmark.
Our test will be aligned with the following one, except that we introduced a new test statistic, which can be more efficiently computed, and it has comparable asymptotic properties with the test statistic that is used below.
\begin{theorem}[\cite{szekely2007measuring}, Theorem 6]
For potentially multivariate random variables $X$ and $Y$, a prescribed level $\alpha_s$, and sample size $n$,
one rejects the independence if and only if
$$
\frac{n \mathcal{V}_n^2(X,Y)}{S_2} > (\Phi^{-1}(1-\alpha_s/2))^2,
$$
where $\mathcal{V}_n^2(X,Y)$ has been defined in \eqref{eq:def_3}, $\Phi(\cdot)$ denote the cumulative distribution function of the standard normal distribution and
$$
S_2 = \frac{1}{n^4} \sum_{i,j=1}^n |X_i - X_j| \sum_{i,j=1}^n |Y_i - Y_j|.
$$
Moreover, let $\alpha(X,Y,n)$ denote the achieved significance level of the above test.
If $\mathbb{E}[|X|+|Y|] < \infty$, then for all $0<\alpha_s<0.215$, one can show the following:
\begin{align*}
&\lim_{n \rightarrow \infty} \alpha(X,Y,n) \le \alpha_s, \mbox{ and }\\
&\sup_{X,Y} \left\{ \lim_{n \rightarrow \infty} \alpha(X,Y,n) : \mathcal{V}(X,Y) =0 \right\} = \alpha_s.
\end{align*}
\end{theorem}
Note that the quantity $\mathcal{V}_n^2(X,Y)$ that is used above as in \cite{szekely2007measuring} differs from the one that will be used in our proposed method.
As mentioned, we use the above as an illustration for distance-based tests of independence, as well as the theoretical (or asymptotic) properties that such a test can achieve.

\section{Numerically Efficient Method for Random Vectors}
\label{sec:methods}

This section is made of two components.
We present a random-projection-based distance covariance estimator that will be proven to be unbiased with a computational complexity that is $O(K n \log n)$ in Section \ref{sec:rp-dist-cov}.
In Section \ref{sec:testIndep}, we describe how the test of independence can be done by utilizing the above estimator.
For user's conveniences, stand-alone algorithms are furnished in the appendix.

\subsection{Random Projection Based Methods for Approximating Distance Covariance}
\label{sec:rp-dist-cov}

We consider how to use a fast algorithm for univariate random variables to compute or approximate the sample distance covariance of random vectors.
The main idea works as follows:
first, projecting the multivariate observations on some random directions;
then, using the fast algorithm to compute the distance covariance of the projections;
finally, averaging distance covariances from different projecting directions.

More specifically, our estimator can be computed as follows.
For potentially multivariate $X_1, \ldots, X_n \in \mathbb{R}^p, p \ge 1$ and $Y_1, \ldots, Y_n \in \mathbb{R}^q, q \ge 1$,
let $K$ be a predetermined number of iterations, we do:
\begin{enumerate}
\item For each $k$ ($1 \le k \le K$), randomly generate $u_k$ and $v_k$ from Uniform($\mathcal{S}^{p-1}$) and Uniform($\mathcal{S}^{q-1}$), respectively. Here $\mathcal{S}^{p-1}$ and $\mathcal{S}^{q-1}$ are the unit spheres in $\mathbb{R}^{p}$ and $\mathbb{R}^{q}$, respectively. Uniform($\mathcal{S}^{p-1}$) is a uniform measure (or distribution) on $\mathcal{S}^{p-1}$.

\item Let $u_k^t X$ and $v_k^t Y$ denote the projections of $X$ and $Y$ to the spaces that are spanned by  vector $u_k$ and $v_k$, respectively. That is we have
\begin{eqnarray*}
u_k^t X = (u_k^t X_1, \ldots, u_k^t X_n), \mbox{ and }
v_k^t Y = (v_k^t Y_1, \ldots, v_k^t Y_n).
\end{eqnarray*}
Note that samples $u_k^t X$ and $v_k^t Y$ are now univariate.

\item Utilize the fast (i.e., order $O(n \log n)$) algorithm that was mentioned in Theorem \ref{th:fast-univariate} to compute for the unbiased estimator in \eqref{eq:def_5} with respect to $u_k^t X$ and $v_k^t Y$.
    Formally, we denote
    $$
    \Omega_n^{(k)} = C_p C_q  \Omega_n(u_k^t X, v_k^t Y),
    $$
    where $C_p$ and $C_q $ have been defined at the end of Section \ref{sec:intro}.

\item The above three steps are repeated for $K$ times.
The final estimator is the average:
\begin{equation}
\label{eq:bar-ome-n}
\overline{\Omega}_n = \frac{1}{K} \sum_{k=1}^K \Omega_n^{(k)}.
\end{equation}
To emphasize the dependency of the above quantity with $K$, we sometimes use a notation $\overline{\Omega}_{n,K} \triangleq \overline{\Omega}_n$.

\end{enumerate}
See Algorithm \ref{algo:approx_algo} in the appendix for a stand-alone presentation of the above method.
In the light of Theorem \ref{th:fast-univariate}, we can handily declare the following.
\begin{theorem}
\label{th:fast-multivariate}
For potentially multivariate $X_1, \ldots, X_n \in \mathbb{R}^p$ and $Y_1, \ldots, Y_n \in \mathbb{R}^q$, the order of computational complexity of computing the aforementioned $\overline{\Omega}_n$ is $O(Kn \log n)$ with storage $O(\max\{n,K\})$, where $K$ is the number of random projections.
\end{theorem}
The proof of the above theorem is omitted, because it is straightforward from Theorem \ref{th:fast-univariate}.
The statistical properties of the proposed estimator $\overline{\Omega}_n$ will be studied in the subsequent section (specifically in Section \ref{sec:th:bar-Omega_n}).

\subsection{Test of Independence}
\label{sec:testIndep}

By a later result (cf. Theorem \ref{th:independence_test}), we can apply $\overline{\Omega}_n$ in the independence testing.
The corresponding asymptotic distribution of the test statistic $\overline{\Omega}_n$ can be approximated by a Gamma$(\alpha,\beta)$ distribution with $\alpha$ and $\beta$ given in  (\ref{eq:approx_dist}).
We can compute the significant level of the test statistic by permutation and conduct the independence test accordingly.
Recall that we have potentially multivariate $X_1, \ldots, X_n \in \mathbb{R}^p$ and $Y_1, \ldots, Y_n \in \mathbb{R}^q$.
Recall that $K$ denotes the number of Monte Carlo iterations in our previous algorithm.
Let $\alpha_s$ denote the prescribed significance level of the independence test.
Let $L$ denote the number of random permutations that we will adopt.
We would like to test the null hypothesis $\mathcal{H}_0$---$X$ and $Y$ are independent---against its alternative.
Recall $\overline{\Omega}_n$ is our proposed estimator in \eqref{eq:bar-ome-n}.
The following algorithm describes an independence test which applies permutations to generate a threshold.
\begin{enumerate}
\item For each $\ell, 1\le \ell \le L$, one generates a random permutation of $Y$: $Y^{\star,\ell} = (Y_1^\star, \ldots Y_n^\star)$;

\item Using the algorithm in Section \ref{sec:rp-dist-cov}, one can compute the estimator $\overline{\Omega}_n$ as in \eqref{eq:bar-ome-n} for $X$ and $Y^{\star,\ell}$; denote the outcome to be $V_\ell = \overline{\Omega}_n(X, Y^{\star,\ell})$.
Note under the random permutations, $X$ and $Y^{\star,\ell}$ are independent.

\item The above two steps are executed for all $\ell= 1,\ldots, L$.
One rejects $\mathcal{H}_0$ if and only if we have
$$
\frac{1+ \sum_{\ell=1}^L I(\overline{\Omega}_n>V_\ell)}{1+L} > \alpha_s.
$$
\end{enumerate}
See Algorithm \ref{algo:perm_based_test} in the appendix for a stand-alone description.

One can also use the information of an approximate asymptotic distribution to estimate a threshold in the aforementioned independence test.
The following describes such an approach.
Recall that we have random vectors $X_1, \ldots, X_n \in \mathbb{R}^p, p \ge 1$ and $Y_1, \ldots, Y_n \in \mathbb{R}^q, q \ge 1$, the number of random projections $K$, and a prescribed significance level $\alpha_s$ that has been mentioned earlier.
\begin{enumerate}
\item For each $k$ ($ 1\le k \le K$), randomly generate $u_k$ and $v_k$ from uniform($\mathcal{S}^{p-1}$) and uniform($\mathcal{S}^{q-1}$), respectively.

\item Use the fast algorithm in Theorem \ref{th:fast-univariate} to compute the following quantities:
\begin{eqnarray*}
&& \Omega_n^{(k)} = C_p C_q  \Omega_n(u_k^t X, v_k^t Y),\\
&& S_{n,1}^{(k)} = C_p^2 C_q^2 \Omega_n(u_k^t X, u_k^t X) \Omega_n(v_k^t Y, v_k^t Y),\\
&& S_{n,2}^{(k)} = C_p \frac{ a_{\cdot \cdot}^{u_k} }{n(n-1)},\quad S_{n,3}^{(k)} = C_q \frac{ b_{\cdot \cdot}^{v_k} }{n(n-1)},
\end{eqnarray*}
where $C_p$ and $C_q $ have been defined at the end of Section \ref{sec:intro} and in the last equation, the $a_{\cdot \cdot}^{u_k}$ and $b_{\cdot \cdot}^{v_k}$ are defined as follows:
\begin{eqnarray*}
&& a_{ij}^{u_k} = |u_k^t(X_i - X_j)|,\quad b_{ij}^{v_k} = |v_k^t(Y_i - Y_j)|, \\
&& a_{\cdot \cdot}^{u_k} = \sum_{k,l=1}^n a_{kl}^{u_k},\quad b_{\cdot \cdot}^{v_k} = \sum_{k,l=1}^n b_{kl}^{v_k}.
\end{eqnarray*}

\item For the aforementioned $k$, one randomly generates $u'_k$ and $v'_k$ from uniform($\mathcal{S}^{p-1}$) and uniform($\mathcal{S}^{q-1}$), respectively.
Use the fast algorithm that is mentioned in Theorem \ref{th:fast-univariate} to compute the following.
\begin{eqnarray*}
\Omega_{n,X}^{(k)} = C_p^2 \Omega_n(u_k^t X, {u'}_k^t X),\quad \Omega_{n,Y}^{(k)} = C_p^2 \Omega_n(v_k^t Y, {v'}_k^t Y).
\end{eqnarray*}
where $C_p$ and $C_q $ have been defined at the end of Section \ref{sec:intro}.

\item Repeat the previous steps for all $k=1,\ldots, K$. Then we compute the following quantities:
\begin{eqnarray}
&&\overline{\Omega}_n = \frac{1}{K} \sum_{k=1}^K \Omega_n^{(k)},\quad \bar{S}_{n,1} = \frac{1}{K} \sum_{k=1}^K S_{n,1}^{(k)},\quad \bar{S}_{n,2} = \frac{1}{K} \sum_{k=1}^K S_{n,2}^{(k)},\nonumber \\
&& \bar{S}_{n,3} = \frac{1}{K} \sum_{k=1}^K S_{n,3}^{(k)},\quad \overline{\Omega}_{n,X} = \frac{1}{K} \sum_{k=1}^K \Omega_{n,X}^{(k)},\quad \overline{\Omega}_{n,Y} = \frac{1}{K} \sum_{k=1}^K \Omega_{n,Y}^{(k)}, \nonumber \\
&&\hspace{0.2\textwidth} \alpha = \frac{1}{2}\frac{\bar{S}_{n,2}^2 \bar{S}_{n,3}^2}{\frac{K-1}{K} \overline{\Omega}_{n,X} \overline{\Omega}_{n,Y} + \frac{1}{K}\bar{S}_{n,1}}, \label{eq:alpha} \\
&&\hspace{0.2\textwidth} \beta = \frac{1}{2}\frac{\bar{S}_{n,2} \bar{S}_{n,3}}{\frac{K-1}{K} \overline{\Omega}_{n,X} \overline{\Omega}_{n,Y} + \frac{1}{K}\bar{S}_{n,1}}. \label{eq:beta}
\end{eqnarray}

\item Reject $\mathcal{H}_0$ if $n \overline{\Omega}_n + \bar{S}_{n,2} \bar{S}_{n,3} > \mbox{Gamma}(\alpha,\beta; 1-\alpha_s)$; otherwise, accept it.
Here $\mbox{Gamma}(\alpha,\beta; 1-\alpha_s)$ is the $1-\alpha_s$ quantile of the distribution Gamma$(\alpha,\beta)$.

\end{enumerate}
The above procedure is motivated by the observation that the asymptotic distribution of the test statistic $n \overline{\Omega}_n$ can be approximated by a Gamma distribution, whose parameters can be estimated by \eqref{eq:alpha} and \eqref{eq:beta}.
A stand-alone description of the above procedure can be found in Algorithm \ref{algo:dist_based_test} in the appendix.

\section{Theoretical Properties}
\label{sec:theory}

In this section, we establish the theoretical foundation of the proposed method.
In Section \ref{sec:rp}, we study some properties of the random projections and the subsequent average estimator.
These properties will be needed in studying the properties of the proposed estimator.
We study the properties of the proposed distance covariance estimator ($\Omega_n$) in Section \ref{sec:Omega_n}, taking advantage of the fact that $\Omega_n$ is a U-statistic.
It turns out that the properties of eigenvalues of a particular operator plays an important role.
We present the relevant results in Section \ref{sec:kernel-lambda}.
The main properties of the proposed estimator ($\overline{\Omega}_n$) is presented in Section \ref{sec:th:bar-Omega_n}.

\subsection{Using Random Projections in Distance-Based Methods}
\label{sec:rp}

In this section, we will study some properties of distance covariances of randomly projected random vectors.
We begin with a necessary and sufficient condition of independence.
\begin{lemma}
\label{lemma:1}
Suppose $u$ and $v$ are points on the hyper-spheres: $u \in \mathcal{S}^{p-1} = \{ u \in \mathbb{R}^p: |u|=1\}$ and $v \in \mathcal{S}^{q-1}$.
We have
$$
\mbox{random vectors } X\in\mathbb{R}^p \mbox{ and } Y\in\mathbb{R}^q \mbox{ are independent}
$$
if and only if
$$
\mathcal{V}^2(u^tX,v^tY) = 0, \text{ for any } u \in \mathcal{S}^{p-1}, v \in \mathcal{S}^{q-1}.
$$
\end{lemma}
The proof is relatively straightforward.
We relegate a formal proof to the appendix.
This lemmas indicates that the independence is somewhat preserved under projections.
The main contribution of the above result is to motivate us to think of using random projection, to reduce the multivariate random vectors into univariate random variables.
As mentioned earlier, there exist fast algorithms of distance-based methods for univariate random variables.

The following result allows us to regard the distance covariance of random vectors of any dimension as an integral of distance covariance of univariate random variables, which are the projections of the aforementioned random vectors.
The formulas in the following lemma provides foundation for our proposed method: the distance covariances in the multivariate cases can be written as integrations of distance covariances in the univariate cases.
our proposed method essentially adopts the principle of Monte Carlo to approximate such integrals.
We again relegate the proof to the appendix.
\begin{lemma}
Suppose $u$ and $v$ are points on unit hyper-spheres: $u \in \mathcal{S}^{p-1} = \{u \in \mathbb{R}^p: |u|=1\}$ and $v \in \mathcal{S}^{q-1}$.
Let $\mu$ and $\nu$ denote the uniform probability measure on $\mathcal{S}^{p-1}$ and $\mathcal{S}^{q-1}$, respectively.
Then, we have for random vectors $X\in\mathbb{R}^p$ and $Y\in\mathbb{R}^q$,
$$
\mathcal{V}^2(X,Y) = C_p C_q \int_{\mathcal{S}^{p-1} \times \mathcal{S}^{q-1}} \mathcal{V}^2(u^tX,v^tY) d\mu(u) d\nu(v),
$$
where $C_p$ and $C_q$ are two constants that are defined at the end of Section \ref{sec:intro}.
Moreover, a similar result holds for the sample distance covariance:
$$
\mathcal{V}^2_n(X,Y) = C_p C_q \int_{\mathcal{S}^{p-1} \times \mathcal{S}^{q-1}} \mathcal{V}_n^2(u^tX,v^tY) d\mu(u) d\nu(v).
$$
\label{lemma:2}
\end{lemma}
Besides the integral equations in the above lemma, we can also establish the following result for the unbiased estimator.
Such a result provides direct foundation of our proposed method.
Recall that $\Omega_n$, which is in \eqref{eq:def_5}, is an unbiased estimator of the distance covariance $\mathcal{V}^2(X,Y)$.
A proof is provided in the appendix.
\begin{lemma}
\label{lem:4.3}
Suppose $u$ and $v$ are points on the hyper-spheres: $u \in \mathcal{S}^{p-1} = \{u \in \mathbb{R}^p: |u|=1\}$ and $v \in \mathcal{S}^{q-1}$. Let $\mu$ and $\nu$ denote the measure corresponding to the uniform densities on the surfaces $\mathcal{S}^{p-1}$ and $\mathcal{S}^{q-1}$, respectively. Then, we have
$$
\Omega_n(X,Y) = C_p C_q \int_{\mathcal{S}^{p-1} \times \mathcal{S}^{q-1}} \Omega_n(u^tX,v^tY) d\mu(u) d\nu(v),
$$
where $C_p$ and $C_q$ are constants that were mentioned at the end of Section \ref{sec:intro}.
\end{lemma}
From the above lemma, recalling the design of our proposed estimator $\overline{\Omega}_n$ as in \eqref{eq:bar-ome-n}, it is straightforward to see that the proposed estimator $\overline{\Omega}_n$ is an unbiased estimator of $\Omega_n(X,Y)$.
For completeness, we state the following without a proof.
\begin{corollary}
\label{coro:bar-omega-n}
The proposed estimator $\overline{\Omega}_n$ in \eqref{eq:bar-ome-n} is an unbiased estimator of the estimator $\Omega_n(X,Y)$ that was defined in \eqref{eq:def_5}.
\end{corollary}

Note that the estimator $\overline{\Omega}_n$ in \eqref{eq:bar-ome-n} evidently depends on the number of random projections $K$.
Recall that to emphasize such a dependency, we sometimes use a notation
$\overline{\Omega}_{n,K} \triangleq \overline{\Omega}_n$.
The following concentration inequality shows the speed that $\overline{\Omega}_{n,K}$ can converge to  $\Omega_n$ as $K \rightarrow \infty$.
\begin{lemma}
\label{lem:concentration}
Suppose $\mathbf{E}[|X|^2] < \infty$ and $\mathbf{E}[|Y|^2] < \infty$.
For any $\epsilon>0$, we have
$$
\mathbf{P}\left(| \overline{\Omega}_{n,K} - \Omega_n | > \epsilon \right) \le 2\exp\left\{ -\frac{CK\epsilon^2}{\mbox{Tr}[\Sigma_X] \mbox{Tr}[\Sigma_Y]} \right\},
$$
where $\Sigma_X$ and $\Sigma_Y$ are the covariance matrices of $X$ and $Y$, respectively,
$\mbox{Tr}[\Sigma_X]$ and $\mbox{Tr}[\Sigma_Y]$ are their matrix traces,
and $C = \frac{2}{25C_p^2C_q^2}$ is a constant.
\end{lemma}
The proof is a relatively standard application of the Hoeffding's inequality \citep{hoeffding1963probability}, which has been relegated to the appendix.
The above lemma essentially indicates that the quantity $| \overline{\Omega}_{n,K} - \Omega_n |$ converges to zero at a rate no worse than $O(1/\sqrt{K})$.

\subsection{Asymptotic Properties of the Sample Distance Covariance $\Omega_n$}
\label{sec:Omega_n}

The asymptotic behavior of a range of sample distance covariance, such as $\Omega_n$ in \eqref{eq:def_5} of this paper, has been studied in many places, seeing \cite{lyons2013distance,huo2015fast,szekely2009brownian,sejdinovic2013equivalence}.
We found that it is still worthwhile to present them here, as we will use them to establish the statistical properties of our proposed estimator.
The asymptotic distributions of $\Omega_n$ will be studied under two situations: (1) a general case and (2) when $X$ and $Y$ are assumed to be independent.
We will see that the asymptotic distributions are different in these two situations.

It has been showed in \cite[Theorem 3.2]{huo2015fast} that $\Omega_n$ is a U-statistic.
In the following, we state the result without a formal proof.
We will need the following function, denoted by $h_4$, which takes four pairs of input variables:
\begin{multline}
\label{eq:def-h4}
h_4((X_1,Y_1),(X_2,Y_2),(X_3,Y_3),(X_4,Y_4)) \\
= \frac{1}{4} \sum_{1\le i,j \le 4, i \neq j} |X_i-X_j| |Y_i - Y_j|
- \frac{1}{4} \sum_{i=1}^4 \left( \sum_{j=1, j \neq i}^4 |X_i-X_j| \sum_{j=1, j \neq i}^4 |Y_i-Y_j| \right)  \\
\hspace{1in} + \frac{1}{24} \sum_{1\le i,j \le 4, i \neq j} |X_i-X_j| \sum_{1\le i,j \le 4, i \neq j} |Y_i-Y_j|.
\end{multline}
Note that the definition of $h_4$ coincides with $\Omega_n$ when the number of observations $n=4$.
\begin{lemma}[U-statistics]
Let $\Psi_4$ denote all distinct 4-subset of $\{1,\ldots,n\}$ and let us define $X_\psi = \{X_i | i \in \psi \}$ and $Y_\psi = \{Y_i | i \in \psi \}$, then $\Omega_n$ is a U-statistic and can be expressed as
$$
\Omega_n = \binom{n}{4}^{-1}\sum_{\psi \in \Psi_4} h_4\left(X_\psi, Y_\psi\right).
$$
\end{lemma}

From the literature of the U-statistics, we know that the following quantities play critical roles.
We state them here:
\begin{align*}
&h_1((X_1,Y_1)) = \mathbb{E}_{2,3,4}[ h_4((X_1,Y_1),(X_2,Y_2),(X_3,Y_3),(X_4,Y_4)) ], \\
&h_2((X_1,Y_1),(X_2,Y_2)) = \mathbb{E}_{3,4}[ h_4((X_1,Y_1),(X_2,Y_2),(X_3,Y_3),(X_4,Y_4)) ], \\
&h_3((X_1,Y_1),(X_2,Y_2),(X_3,Y_3)) = \mathbb{E}_{4}[ h_4((X_1,Y_1),(X_2,Y_2),(X_3,Y_3),(X_4,Y_4)) ],
\end{align*}
where $\mathbb{E}_{2,3,4}$ stands for taking expectation over $(X_2,Y_2),(X_3,Y_3)$ and $(X_4,Y_4)$;
$\mathbb{E}_{3,4}$ stands for taking expectation over $(X_3,Y_3)$ and $(X_4,Y_4)$;
and $\mathbb{E}_{4}$ stands for taking expectation over $(X_4,Y_4)$; respectively.

One immediate application of the above notations is the following result, which quantifies the variance of $\Omega_n$.
Since the formula is a known result, seeing \cite[Chapter 5.2.1, Lemma A]{serfling1980approximation},
we state it without a proof.
\begin{lemma}[Variance of the U-statistic]
\label{lem:var-U}
The variance of $\Omega_n$ could be written as
\begin{multline*}
\mbox{Var}(\Omega_n) = \binom{n}{4}^{-1} \sum_{l=1}^4 \binom{4}{l} \binom{n-4}{4-l} \mbox{Var}(h_l) \\
= \frac{16}{n} \mbox{Var}(h_1) + \frac{240}{n^2} \mbox{Var}(h_1) + \frac{72}{n^2} \mbox{Var}(h_2)+ O\left(\frac{1}{n^3}\right),
\end{multline*}
where $O(\cdot)$ is the standard big O notation in mathematics.
\end{lemma}

From the above lemma, we can see that $\mbox{Var}(h_1)$ and $\mbox{Var}(h_2)$ play indispensable roles in determining the variance of $\Omega_n$.
The following lemma shows that under some conditions, we can ensure that $\mbox{Var}(h_1)$ and $\mbox{Var}(h_2)$ are bounded.
A proof has been relegated to the appendix.
\begin{lemma}
\label{lem:bound-vars}
If we have $\mathbb{E}[|X|^2] < \infty$, $\mathbb{E}[|Y|^2] < \infty$ and $\mathbb{E}[|X|^2|Y|^2] < \infty$, then we have $\mbox{Var}(h_4) < \infty$.
Consequently, we also have $\mbox{Var}(h_1) < \infty$ and $\mbox{Var}(h_2) < \infty$.
\end{lemma}

Even though as indicated in Lemma \ref{lem:var-U}, the quantities $h_1(X_1,Y_1)$ and  \\$h_2((X_1,Y_1),(X_2,Y_2))$ play important roles in determine the variance of $\Omega_n$, in a generic case, they do not have a simple formula.
The following lemma gives the generic formulas for $h_1(X_1,Y_1)$ and $h_2((X_1,Y_1),(X_2,Y_2))$.
Its calculation can be found in the appendix.
\begin{lemma}[Generic $h_1$ and $h_2$]
\label{lem:generic-h1-h2}
In the general case, assuming $(X_1,Y_1)$, $(X,Y)$, $(X',Y')$, and $(X'',Y'')$ are independent and identically distributed, we have
\begin{align}
 h_1((X_1,Y_1))
=& \quad \frac{1}{2} \mathbb{E}[|X_1 - X'| |Y_1 - Y'|] - \frac{1}{2} \mathbb{E}[|X_1 - X'| |Y_1 - Y''|] \nonumber \\
& + \frac{1}{2} \mathbb{E}[|X_1 - X'| |Y - Y''|] -\frac{1}{2} \mathbb{E}[|X_1 - X'| |Y' - Y''|] \nonumber \\
& + \frac{1}{2} \mathbb{E}[|X - X''| |Y_1 - Y'|] - \frac{1}{2} \mathbb{E}[|X' - X''| |Y_1 - Y'|] \nonumber \\
& + \frac{1}{2} \mathbb{E}[|X - X'| |Y - Y'|]    - \frac{1}{2} \mathbb{E}[|X - X'| |Y - Y''|]. \nonumber
\end{align}
We have a similar formula for $h_2((X_1,Y_1),(X_2,Y_2))$ in \eqref{def:h2_1}.
Due to its length, we do not display it here.
\end{lemma}

If one assumes that $X$ and $Y$ are independent, we can have simpler formula for $h_1$, $h_2$, as well as their corresponding variances.
We list the results below, with detailed calculation relegated to the appendix.
One can see that under independence, the corresponding formulas are much simpler.
\begin{lemma}
\label{lem:h1-h2-independent}
When $X$ and $Y$ are independent, we have the following.
For $(X,Y)$ and $(X',Y')$ that are independent and identically distributed as $(X_1,Y_1)$ and $(X_2,Y_2)$, we have
\begin{align}
h_1((X_1,Y_1)) &= 0,  \label{def: h1_2}\\
h_2((X_1,Y_1),(X_2,Y_2)) &= \frac{1}{6} \left(|X_1 - X_2| - \mathbb{E}[|X_1 - X|] - \mathbb{E}[|X_2 - X|] + \mathbb{E}[|X - X'|]\right) \label{def:h2_2} \\
&  \quad (|Y_1 - Y_2| - \mathbb{E}[|Y_1 - Y|] - \mathbb{E}[|Y_2 - Y|] + \mathbb{E}[|Y - Y'|]), \nonumber \\
\mbox{Var}(h_2) &= \frac{1}{36} \mathcal{V}^2(X,X) \mathcal{V}^2(Y,Y), \label{result:B}
\end{align}
where $\mathbb{E}$ stands for the expectation operators with respect to $X$, $X$ and $X'$, $Y$, or $Y$ and $Y'$, whenever appropriate, respectively.
\end{lemma}

If we have $0< \mbox{Var}(h_1) < \infty$, it is known that the asymptotic distribution of $\Omega_n$ is normal, as stated in the following.
Note that based on Lemma \ref{lem:h1-h2-independent}, $X$ and $Y$ cannot be independent; otherwise one should have $h_1=0$ almost surely.
The following theorem is based on a known result on the convergence of U-statistics, seeing \cite[Chapter 5.5.1 Theorem A]{serfling1980approximation}.
We state it without a proof.
\begin{theorem}
\label{th:omega-n-asymp}
Suppose $n\ge7$, $0< \mbox{Var}(h_1) < \infty$ and $\mbox{Var}(h_4) < \infty$, then we have
$$
\Omega_n \xrightarrow{P} \mathcal{V}^2(X,Y)
$$
moreover, we have
$$
\sqrt{n} ( \Omega_n- \mathcal{V}^2(X,Y)) \xrightarrow{D} N(0,16\mbox{Var}(h_1)), \text{ as } n \rightarrow \infty.
$$
\end{theorem}

When $X$ and $Y$ are independent, the asymptotic distribution of $\sqrt{n} \Omega_n$ is no longer normal.
In this case, from Lemma \ref{lem:h1-h2-independent}, we have
$$
h_1((X_1,Y_1)) = 0 \text{ almost surely}, \text{ and } \mbox{Var}[h_1((X_1,Y_1))] = 0.
$$
The following theorem, which applies a result in \cite[Chapter 5.5.2]{serfling1980approximation}, indicates that  $n \Omega_n $ converges to a weighted sum of (possibly infinitely many) independent $\chi_1^2$ random variables.
\begin{theorem}
\label{th:Omega-n-indep}
If $X$ and $Y$ are independent, the asymptotic distribution of $\Omega_n$ is
$$
n \Omega_n \xrightarrow{D} \sum_{i=1}^\infty \lambda_i (Z_i^2-1) = \sum_{i=1}^\infty \lambda_i Z_i^2 - \sum_{i=1}^\infty \lambda_i,
$$
where $Z_i^2 \sim \chi_1^2$ i.i.d, $\lambda_i$'s are the eigenvalues of operator $G$ that is defined as
$$
G g(x_1,y_1) = \mathbb{E}_{x_2,y_2} [ 6h_2((x_1,y_1),(x_2,y_2)) g(x_2,y_2) ],
$$
where function $h_2((\cdot,\cdot),(\cdot,\cdot))$ was defined in \eqref{def:h2_2}.
\end{theorem}
\begin{proof}
The asymptotic distribution of $\Omega_n$ is from the result in \cite[Chapter 5.5.2]{serfling1980approximation}.
\end{proof}
See Subsection \ref{sec:kernel-lambda} for more details on methods for computing the value of $\lambda_i$'s.
In particular, we will show that we have
$\sum_{i=1}^\infty \lambda_i = \mathbb{E}[ |X-X'| ] \mathbb{E}[ |Y-Y'| ]$ (Corollary \ref{coro: eigen_sum})
and $\sum_{i=1}^\infty \lambda_i^2 = \mathcal{V}^2(X,X) \mathcal{V}^2(Y,Y)$ (which is essentially from \eqref{result:B} and Lemma \ref{lem:var-U}).

\subsection{Properties of Eigenvalues $\lambda_i$'s}
\label{sec:kernel-lambda}

From Theorem \ref{th:Omega-n-indep}, we see that the eigenvalues $\lambda_i$'s play important role in determining the asymptotic distribution of $\Omega_n$.
We study its properties here.
Throughout this subsection, we assume that $X$ and $Y$ are independent.
Let us recall that the asymptotic distribution of sample distance covariance $\Omega_n$,
$$
n \Omega_n \xrightarrow{D} \sum_{i=1}^\infty \lambda_i (Z_i^2-1) = \sum_{i=1}^\infty \lambda_i Z_i^2 - \sum_{i=1}^\infty \lambda_i,
$$
where $\lambda_i$'s are the eigenvalues of the operator $G$ that is defined as
$$
G g(x_1,y_1) = \mathbb{E}_{x_2,y_2} [ 6h_2((x_1,y_1),(x_2,y_2)) g(x_2,y_2) ],
$$
where function $h_2((\cdot,\cdot),(\cdot,\cdot))$ was defined in \eqref{def:h2_2}.
By definition, eigenvalues $\lambda_1, \lambda_2, \ldots$ corresponding to distinct solutions of the following equation
\begin{equation}
G g(x_1,y_1) = \lambda g(x_1,y_1). \label{eq:eigen}
\end{equation}
We now study the properties of $\lambda_i$'s.
Utilizing the Lemma 12 and equation (4.4) in \cite{sejdinovic2013equivalence}, we can verify the following result.
We give details of verifications in the appendix.
\begin{lemma}
\label{lem:hXhY}
Both of the following two functions are positive definite kernels:
$$
h_X(X_1, X_2) = - |X_1 - X_2| + \mathbb{E}[ |X_1 - X| ] + \mathbb{E}[ |X_2 - X| ] - \mathbb{E}[ |X - X'| ]
$$
and
$$
h_Y(Y_1, Y_2) = - |Y_1 - Y_2| + \mathbb{E}[ |Y_1 - Y| ] + \mathbb{E}[ |Y_2 - Y| ] - \mathbb{E}[ |Y - Y'| ].
$$
\end{lemma}

The above result gives us a foundation to apply the equivalence result that has been articulated thoroughly in \cite{sejdinovic2013equivalence}.
Equipped with the above lemma, we have the following result, which characterizes a property of $\lambda_i$'s.
The detailed proof can be found in the appendix.
\begin{lemma}
\label{lem:tensor-prod}
Suppose $\{\lambda_1, \lambda_2, \ldots\}$ are the set of eigenvalues of kernel \\$6h_2((x_1,y_1),(x_2,y_2))$, $\{\lambda_1^X, \lambda_2^X, \ldots\}$ and $\{\lambda_1^Y, \lambda_2^Y, \ldots\}$ are the sets of eigenvalues of the positive definite kernels $h_X$ and $h_Y$, respectively.
We have the following:
$$
\{\lambda_1, \lambda_2, \ldots\} = \{\lambda^X_1, \lambda^X_2, \ldots \}
\otimes
\{\lambda^Y_1, \lambda^Y_2, \ldots \};
$$
that is, each $\lambda_i$ satisfying \eqref{eq:eigen} can be written as, for some $j,j'$,
$$
\lambda_i = \lambda_{j}^X \cdot \lambda_{j'}^Y
$$
where $\lambda_{j}^X$ and $\lambda_{j'}^Y$ are the eigenvalues corresponding to kernel functions $h_X(X_1, X_2)$ and $h_Y(Y_1, Y_2)$, respectively.
\end{lemma}

Above lemma implies that eigenvalues of $h_2$ could be obtained immediately after knowing the eigenvalues of $h_X$ and $h_Y$. But, in practice, there usually does not exist analytic solution for even the eigenvalues of $h_X$ or $h_Y$. Instead, given the observations $(X_1, \ldots, X_n)$ and $(Y_1, \ldots, Y_n)$, we can compute the eigenvalues of matrices $\widetilde{K}_X = (h_X(X_i, X_j))_{n \times n}$ and $\widetilde{K}_Y = (h_Y(Y_i, Y_j))_{n \times n}$ and use those empirical eigenvalues to approximate $\lambda_1^X, \lambda_2^X, \ldots$ and $\lambda_1^Y, \lambda_2^Y, \ldots$, and then consequently  $\lambda_1, \lambda_2, \ldots$

We end this subsection with the following corollary on the summations of eigenvalues, which is necessary for the proof of Theorem \ref{th:Omega-n-indep}.
The proof can be found in the appendix.
\begin{corollary}
\label{coro: eigen_sum}
The aforementioned eigenvalues
$\lambda_1^X, \lambda_2^X, \ldots$ and $\lambda_1^Y, \lambda_2^Y, \ldots$ satisfy
$$
\sum_{i=1}^\infty \lambda_i^X = \mathbb{E}[|X-X'|], \text{ and } \sum_{i=1}^\infty \lambda_i^Y = \mathbb{E}[|Y-Y'|].
$$
As a result, we have
$$
\sum_{i=1}^\infty \lambda_i = \mathbb{E}[|X-X'|] \mathbb{E}[|Y-Y'|],
$$
and
$$
\sum_{i=1}^\infty \lambda_i^2 = \mathcal{V}^2(X,X) \mathcal{V}^2(Y,Y).
$$
\end{corollary}

\subsection{Asymptotic Properties of Averaged Projected Sample Distance Covariance $\overline{\Omega}_n$}
\label{sec:th:bar-Omega_n}

We have reviewed the properties of the statistics $\Omega_n$ in a previous section (Section \ref{sec:Omega_n}).
The disadvantage of directly applying $\Omega_n$ (which is defined in \eqref{eq:def_5}) is that for multivariate $X$ and $Y$, the implementation may require at least $O(n^2)$ operations.
Recall that for univariate $X$ and $Y$, an $O(n \log n)$ algorithm exists, cf. Theorem \ref{th:fast-univariate}.
The proposed estimator ($\overline{\Omega}_n$ in \eqref{eq:bar-ome-n}) is the averaged distance covariances, after randomly projecting $X$ and $Y$ to one-dimensional spaces, respectively.
In this section, we will study the asymptotic behavior of $\overline{\Omega}_n$.
It turns out that the analysis will be similar to the works in Section \ref{sec:Omega_n}.
The asymptotic distribution of $\overline{\Omega}_n$ will differ in two cases: (1) the dependent case and (2) the case when $X$ and $Y$ are independent.

As a preparation of presenting the main result, we recall and introduce some notations.
Recall the definition of $\overline{\Omega}_n$:
$$
\overline{\Omega}_n = \frac{1}{K} \sum_{k=1}^K \Omega_n^{(k)},
$$
where
$$
\Omega_n^{(k)} = C_p C_q \Omega_n(u_k^t X, v_k^t Y)
$$
and constants $C_p, C_q $ have been defined at the end of Section \ref{sec:intro}.
By Corollary \ref{coro:bar-omega-n}, we have  $\mathbb{E}\left[\Omega_n^{(k)}\right] = \Omega_n$, where $\mathbb{E}$ stands for the expectation with respect to the random projection.
Note that from the work in Section \ref{sec:Omega_n}, estimator $\Omega_n^{(k)}$ is a U-statistic.
The following equation reveals that estimator $\overline{\Omega}_n$ is also a U-statistic,
$$
\overline{\Omega}_n = \binom{n}{4}^{-1}\sum_{\psi \in \Psi_4} \frac{C_p C_q }{K} \sum_{k=1}^K h_4(u_k^t X_\psi, v_k^t Y_\psi) \triangleq \binom{n}{4}^{-1}\sum_{\psi \in \Psi_4} \bar{h}_4(X_\psi, Y_\psi),
$$
where
$$
\bar{h}_4(X_\psi, Y_\psi) = \frac{1}{K} \sum_{k=1}^K C_p C_q h_4(u_k^t X_\psi, v_k^t Y_\psi).
$$

We have seen that quantities $h_1$ and $h_2$ play significant roles in the asymptotic behavior of statistic $\Omega_n$.
Let us define the counterpart notations as follows:
\begin{eqnarray*}
\bar{h}_1((X_1,Y_1)) &=& \mathbb{E}_{2,3,4}[ \bar{h}_4((X_1,Y_1),(X_2,Y_2),(X_3,Y_3),(X_4,Y_4)) ] \\
&\triangleq& \frac{1}{K} \sum_{k=1}^K h_1^{(k)} \\
\bar{h}_2((X_1,Y_1),(X_2,Y_2)) &=& \mathbb{E}_{3,4}[ \bar{h}_4((X_1,Y_1),(X_2,Y_2),(X_3,Y_3),(X_4,Y_4)) ] \\
&\triangleq& \frac{1}{K} \sum_{k=1}^K h_2^{(k)},
\end{eqnarray*}
where $\mathbb{E}_{2,3,4}$ stands for taking expectation over $(X_2,Y_2),(X_3,Y_3)$ and $(X_4,Y_4)$;
$\mathbb{E}_{3,4}$ stands for taking expectation over $(X_3,Y_3)$ and $(X_4,Y_4)$; as well as the following:
\begin{eqnarray*}
h_1^{(k)} &=& \mathbb{E}_{2,3,4}[ C_p C_q h_4(u_k^t X_\psi, v_k^t Y_\psi) ], \\
h_2^{(k)} &=& \mathbb{E}_{3,4}[ C_p C_q h_4(u_k^t X_\psi, v_k^t Y_\psi) ]. \nonumber
\end{eqnarray*}
In the general case, we do not assume that $X$ and $Y$ are independent.
Let $U = (u_1,\ldots,u_K)$ and $V = (v_1,\ldots,v_K)$ denote the collection of random projections. We can write the variance of $\overline{\Omega}_n$ as follows.
The proof is an application of Lemma \ref{lem:var-U} and the law of total covariance. We relegate it to the appendix.
\begin{lemma}
Suppose $\mathbb{E}_{U,V}[ \mbox{Var}_{X,Y}(\bar{h}_1 | U,V) ] >0$ and $\mbox{Var}_{u,v}( \mathcal{V}^2(u^t X, v^t Y) ) >0$, then, the variance of $\overline{\Omega}_n$ is
\begin{eqnarray*}
\mbox{Var}(\overline{\Omega}_n) &=& \frac{1}{K} \mbox{Var}_{u,v}( \mathcal{V}^2(u^t X, v^t Y) ) +
\frac{16}{n} \mathbb{E}_{U,V}[ \mbox{Var}_{X,Y}(\bar{h}_1 | U,V) ] \\
&&+ \frac{72}{n^2} \mathbb{E}_{U,V}[ \mbox{Var}_{X,Y}(\bar{h}_2 | U,V) ]+ O\left(\frac{1}{n^3}\right).
\end{eqnarray*}
\label{lem:var-U-avg}
\end{lemma}

Equipped with above lemma, we can summarize the asymptotic properties in the following theorem. We state it without a proof as it is an immediate result from Lemma \ref{lem:var-U-avg} as well as the contents in \cite[Chapter 5.5.1 Theorem A]{serfling1980approximation}.
\begin{theorem}
Suppose $0< \mathbb{E}_{U,V}[ \mbox{Var}_{X,Y}(\bar{h}_1 | U,V) ] < \infty$, \\$\mathbb{E}_{U,V}[ \mbox{Var}_{X,Y}(\bar{h}_4 | U,V) ] < \infty$. Also, let us assume that $K \rightarrow \infty$, $n \rightarrow \infty$, then we have
$$
\overline{\Omega}_n \xrightarrow{P} \mathcal{V}^2(X,Y).
$$
And, the asymptotic distribution of $\overline{\Omega}_n$ could differ under different conditions.
\begin{enumerate}
	\item If $K \rightarrow \infty$ and $K/n \rightarrow 0$, then
	$$
	\sqrt{K} \left( \overline{\Omega}_n - \mathcal{V}^2(X,Y) \right) \xrightarrow{D} N\left(0,\mbox{Var}_{u,v}( \mathcal{V}^2(u^t X, v^t Y) ) \right).
	$$
	\item If $n \rightarrow \infty$ and $K/n \rightarrow \infty$, then
	$$
	\sqrt{n} \left( \overline{\Omega}_n- \mathcal{V}^2(X,Y) \right) \xrightarrow{D} N\left(0, 16\mathbb{E}_{U,V}[ \mbox{Var}_{X,Y}(\bar{h}_1 | U,V) ] \right).
	$$
	\item If $n \rightarrow \infty$ and $K/n \rightarrow C$, where $C$ is some constant, then
	\begin{multline*}
	\sqrt{n} \left( \overline{\Omega}_n- \mathcal{V}^2(X,Y) \right) \xrightarrow{D} \\ N\left(0, \frac{1}{C}\mbox{Var}_{u,v}( \mathcal{V}^2(u^t X, v^t Y) ) + 16\mathbb{E}_{U,V}[ \mbox{Var}_{X,Y}(\bar{h}_1 | U,V) ] \right).
	\end{multline*}
\end{enumerate}
\end{theorem}

Since our main idea is to utilize $\overline{\Omega}_n$ to approximate the quantity $\Omega_n$, it is of interests to compare the asymptotic variance of $\Omega_n$ in Theorem \ref{th:omega-n-asymp} with the asymptotic variances in the above theorem.
We present some discussions in the following remark.
\begin{remark}
Let us recall the asymptotic properties of $\Omega_n$ ,
$$
\sqrt{n} (\Omega_n- \mathcal{V}^2(X,Y)) \xrightarrow{D} N(0,16\mbox{Var}(h_1)).
$$
Then, we make the comparison in the following different scenarios.
\begin{enumerate}
	\item If $K \rightarrow \infty$ and $K/n \rightarrow 0$, then the convergence rate of $\overline{\Omega}_n$ is much slower than $\Omega_n$ as $K \ll n$.
	\item If $n \rightarrow \infty$ and $K/n \rightarrow \infty$, then the convergence rate of $\overline{\Omega}_n$ is the same with $\Omega_n$ and their variances is also the same
	\item If $n \rightarrow \infty$ and $K/n \rightarrow C$, where $C$ is some constant, then the convergence rate of $\overline{\Omega}_n$ is the same with $\Omega_n$ but the variance of $\overline{\Omega}_n$ is larger than that of $\Omega_n$.
\end{enumerate}
\end{remark}
Generally, when $X$ is not independent of Y,  $\overline{\Omega}_n$ is as as good as $\Omega_n$ in terms of convergence rate. However, in the independence test, the convergence rate of test statistics under the null hypotheses is of more interest. In the following context of this section, we will show that $\overline{\Omega}_n$ has the same convergence rate with $\Omega_n$ when $X$ is independent of $Y$.

Now, let us consider the case that $X$ and $Y$ are independent. Similarly, by Lemma \ref{lem:h1-h2-independent}, we have
$$
\bar{h}_1^{(k)} = 0, \bar{h}_1 = 0, \text{almost surely}, \text{ and }, \mbox{Var}(\bar{h}_1) = 0.
$$
And, by Lemma \ref{lemma:1}, we know that
$$
\mathcal{V}^2(u^t X, v^t Y) = 0, \forall u,v,
$$
which implies
$$
\mbox{Var}_{u,v} \left( \mathcal{V}^2(u^t X, v^t Y) \right) = 0.
$$
Therefore, we only need to consider $\mbox{Var}_{X,Y}(\bar{h}_2 | U,V)$. Suppose $(U,V)$ is given, a result in \cite[Chapter 5.5.2]{serfling1980approximation}, together with Lemma \ref{lem:var-U-avg}, indicates that $n \overline{\Omega}_n $ converges to a weighted sum of (possibly infinitely many) independent $\chi_1^2$ random variables.
The proof can be found in appendix.
\begin{theorem}
If $X$ and $Y$ are independent, given the value of $U = (u_1,\ldots,u_K)$ and $V = (v_1,\ldots,v_K)$, the asymptotic distribution of $\overline{\Omega}_n$ is
$$
n \overline{\Omega}_n \xrightarrow{D} \sum_{i=1}^\infty \bar{\lambda}_i (Z_i^2-1) = \sum_{i=1}^\infty \bar{\lambda}_i Z_i^2 - \sum_{i=1}^\infty \bar{\lambda}_i,
$$
where $Z_i^2 \sim \chi_1^2$ i.i.d, and
$$
\sum_{i=1}^\infty \bar{\lambda}_i = \frac{C_pC_q}{K} \sum_{k=1}^K \mathbb{E}[|u_k^t(X-X')|] \mathbb{E}[|v_k^t(Y-Y')|],
$$
$$
\sum_{i=1}^\infty \bar{\lambda}_i^2 = \frac{C^2_p C^2_q}{K^2} \sum_{k,k'=1}^K \mathcal{V}^2\left(u_k^t X, u_{k'}^t X\right) \mathcal{V}^2\left(v_k^t Y, v_{k'}^t Y\right).
$$
\label{th:independence_test}
\end{theorem}
\begin{remark}
Let us recall that if $X$ and $Y$ are independent, the asymptotic distribution of $\Omega_n$ is
$$
n \Omega_n \xrightarrow{D} \sum_{i=1}^\infty \lambda_i (Z_i^2-1).
$$
Theorem \ref{th:independence_test} shows that under the null hypotheses, $\overline{\Omega}_n$ enjoys the same convergence rate with $\Omega_n$.
\end{remark}
There usually does not exist a close-form expression for $\sum_{i=1}^\infty \bar{\lambda}_i Z_i^2$, but we can approximate it with the Gamma distribution whose first two moments matched.
Thus, we have that $\sum_{i=1}^\infty \bar{\lambda}_i Z_i^2$ could be approximated by
$\mbox{Gamma}(\alpha,\beta)$ with probability density function
$$
\frac{\beta^\alpha}{\Gamma(\alpha)} x^{\alpha-1} e^{-\beta x}, x>0,
$$
where
\begin{align}
\alpha = \frac{1}{2} \frac{(\sum_{i=1}^\infty \bar{\lambda}_i)^2}{\sum_{i=1}^\infty \bar{\lambda}_i^2}, \beta = \frac{1}{2} \frac{\sum_{i=1}^\infty \bar{\lambda}_i}{\sum_{i=1}^\infty \bar{\lambda}_i^2}. \label{eq:approx_dist}
\end{align}
See \cite[Section 3]{box1954some} for an empirical justification on this Gamma approximation. See \cite{bodenham2014comparison} for a survey on different approximation methods of weighted sum of chi-square distribution.

The following result shows that both $\sum_{i=1}^\infty \bar{\lambda}_i$ and $\sum_{i=1}^\infty \bar{\lambda}_i^2$ could be estimated from data, see appendix for the corresponding justification.
\begin{proposition}
\label{prop:approx}
One can approximate $\sum_{i=1}^\infty \bar{\lambda}_i$ and $\sum_{i=1}^\infty \bar{\lambda}_i$ as follows:
\begin{align*}
\sum_{i=1}^\infty \bar{\lambda}_i \approx & \frac{C_pC_q}{Kn^2(n-1)^2} \sum_{k=1}^K a_{\cdot \cdot}^{u_k} b_{\cdot \cdot}^{v_k}, \\
\sum_{i=1}^\infty \bar{\lambda}_i^2 \approx & \frac{K-1}{K} \Omega_n(X,X) \Omega_n(Y,Y) \\
& \hspace{0.1\textwidth} + \frac{C_p^2 C_q^2}{K} \sum_{k=1}^K  \Omega_n(u_k^t X, u_{k}^t X) \Omega_n(v_k^t Y, v_{k}^t Y).
\end{align*}
\end{proposition}

\section{Simulations}
\label{sec:simulations}

Our numerical studies follow the works of \cite{sejdinovic2013equivalence,gretton2005measuring,szekely2007measuring}.
In Section \ref{sec:sample-sizes-etc}, we study how the performance of the proposed estimator is influenced by some parameters, including the sample size, the dimensions of the data, as well as the number of random projections in our algorithm.
We also study and compare the computational efficiency of the direct method and the proposed method in Section \ref{sec:efficiency-compare}.
The comparison of the corresponding independence test with other existing methods will be included in Section \ref{sec:test-compare}.

\subsection{Impact of Sample Size, Data Dimensions and the Number of Monte Carlo Iterations}
\label{sec:sample-sizes-etc}

In this part, we will use some synthetic data to study impact of sample size $n$, data dimensions $(p,q)$ and the number of the Monte Carlo iterations $K$ on the convergence and test power of our proposed test statistic $\overline{\Omega}_n$.
The significance level is set to be $\alpha_s=0.05$. Each experiment is repeated for $N=400$ times to get reliable mean and variance of estimators.

In first two examples, we fix data dimensions $p=q=10$ and let the sample size $n$ vary in $100$, $500$, $1000$, $5000$, $10000$ and let the number of the Monte Carlo iterations $K$ vary in $10$, $50$, $100$, $500$, and $1000$.
The data generation mechanism is described as follows, and it generates independent variables.
\begin{example}
We generate random vectors $X \in \mathbb{R}^{10}$ and $Y \in \mathbb{R}^{10}$. Each entry $X_i$ follows $\mbox{Unif}(0,1)$, independently. Each entry $Y_i = Z_i^2$, where $Z_i$ follows $\mbox{Unif}(0,1)$, independently.
\label{ex:1}
\end{example}
See Figure \ref{fig:1} for the boxplots of the outcomes of Example \ref{ex:1}. In each subfigure, we fix the Monte Carlo iteration number $K$ and let the number of observations $n$ grow.
It is worth noting that the scale of each subfigure could be different in order to display the entire boxplots.
This experiment shows that the estimator converges to $0$ regardless of the number of the Monte Carlo iterations. It also suggests that $K=50$ Monte Carlo iterations should suffice in the independent cases.
\begin{figure}[ht!]
    \centering
    \begin{subfigure}[b]{0.3\textwidth}
        \includegraphics[width=\textwidth]{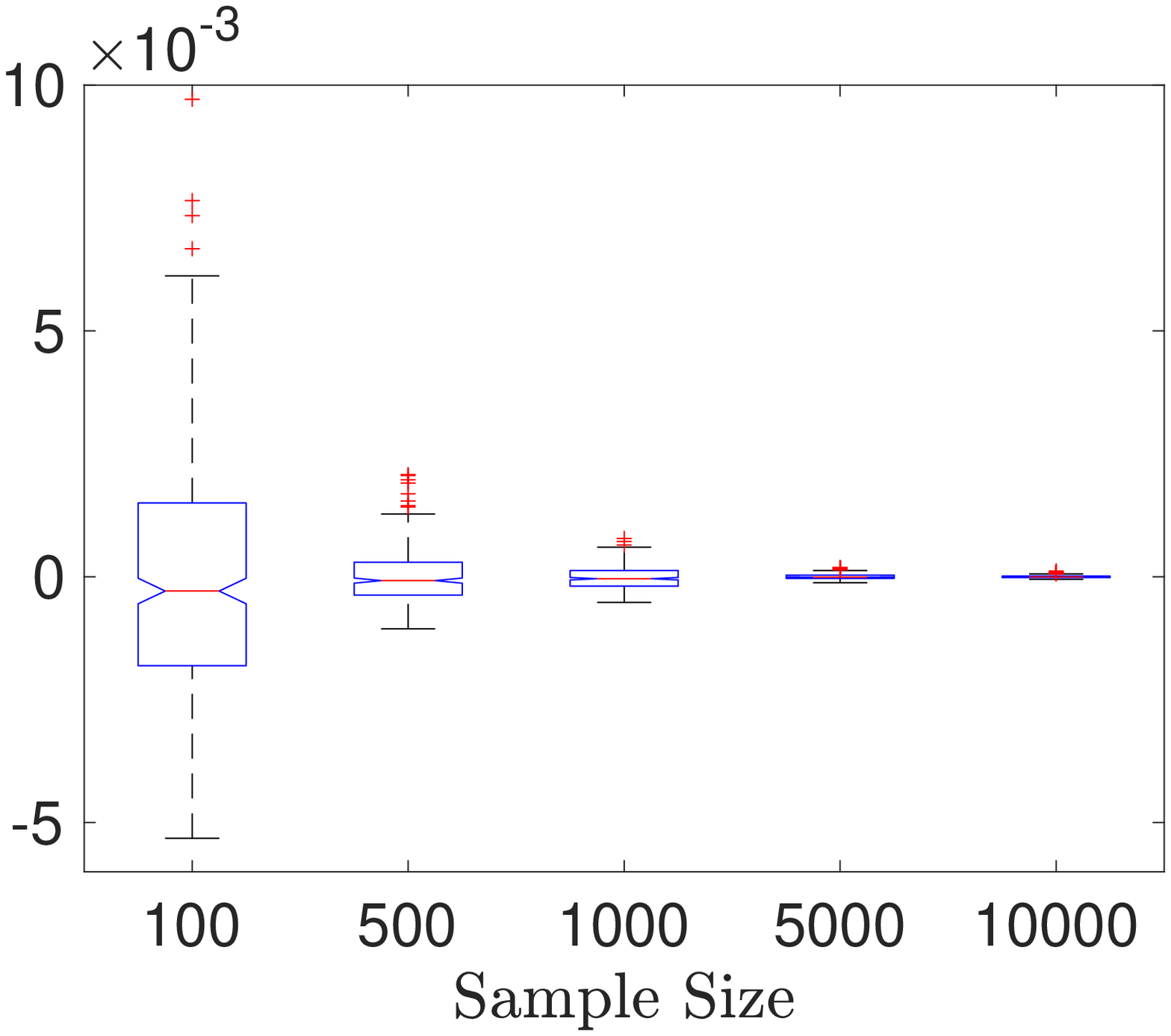}
        \caption{K=10}
    \end{subfigure}
    ~
    \begin{subfigure}[b]{0.3\textwidth}
        \includegraphics[width=\textwidth]{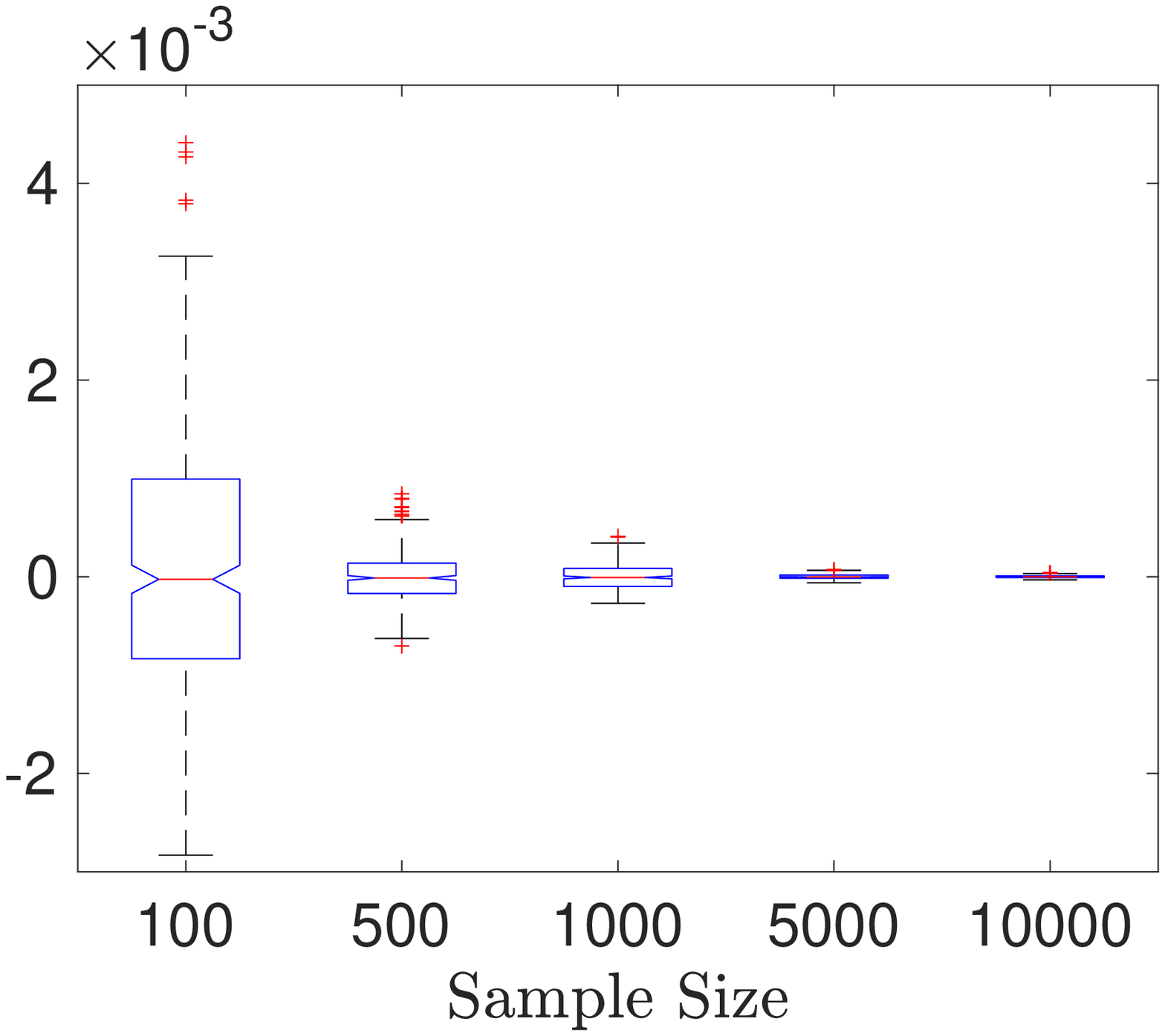}
        \caption{K=50}
    \end{subfigure}
    ~
    \begin{subfigure}[b]{0.3\textwidth}
        \includegraphics[width=\textwidth]{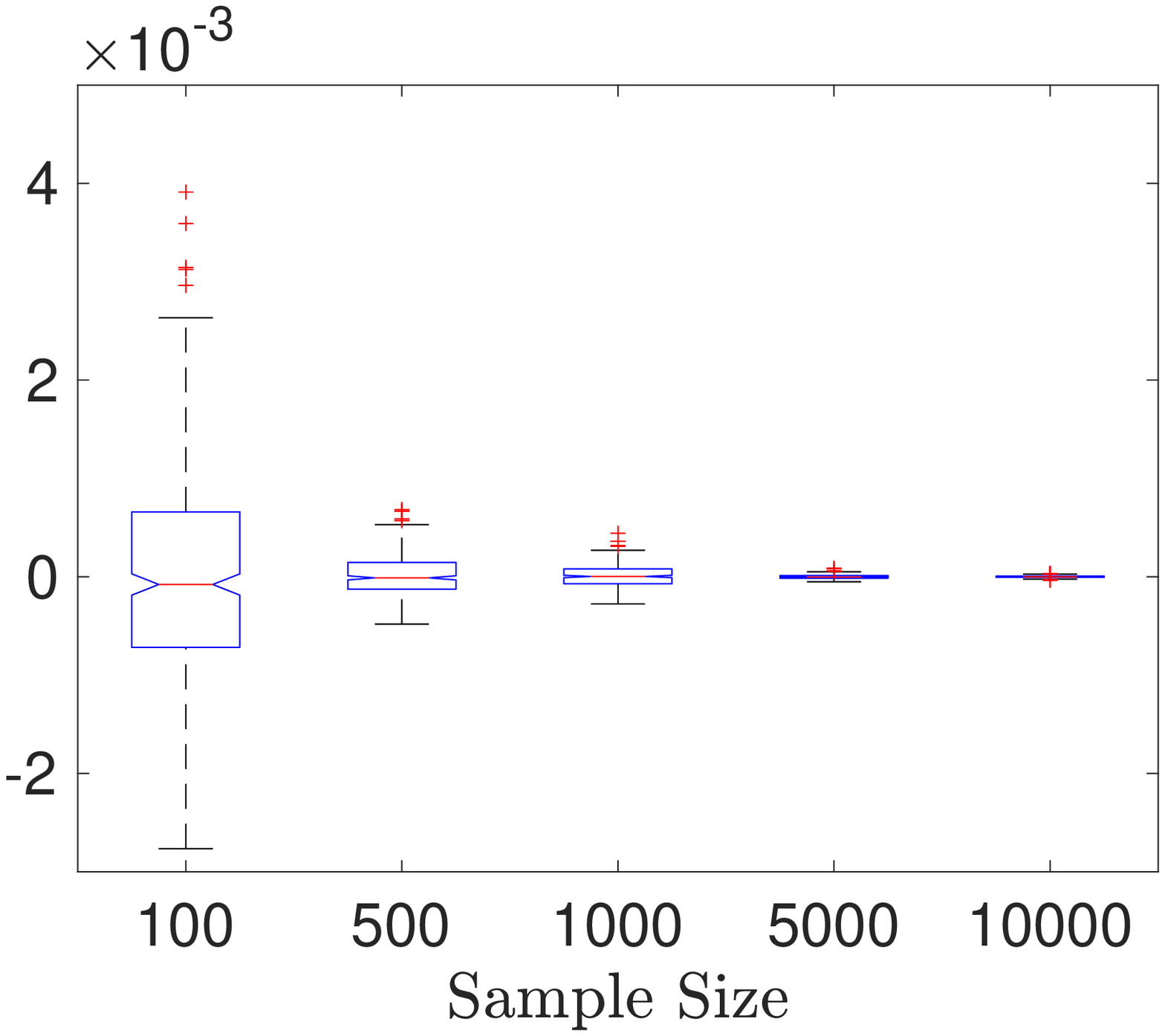}
        \caption{K=100}
    \end{subfigure}
    \\
    \begin{subfigure}[b]{0.3\textwidth}
        \includegraphics[width=\textwidth]{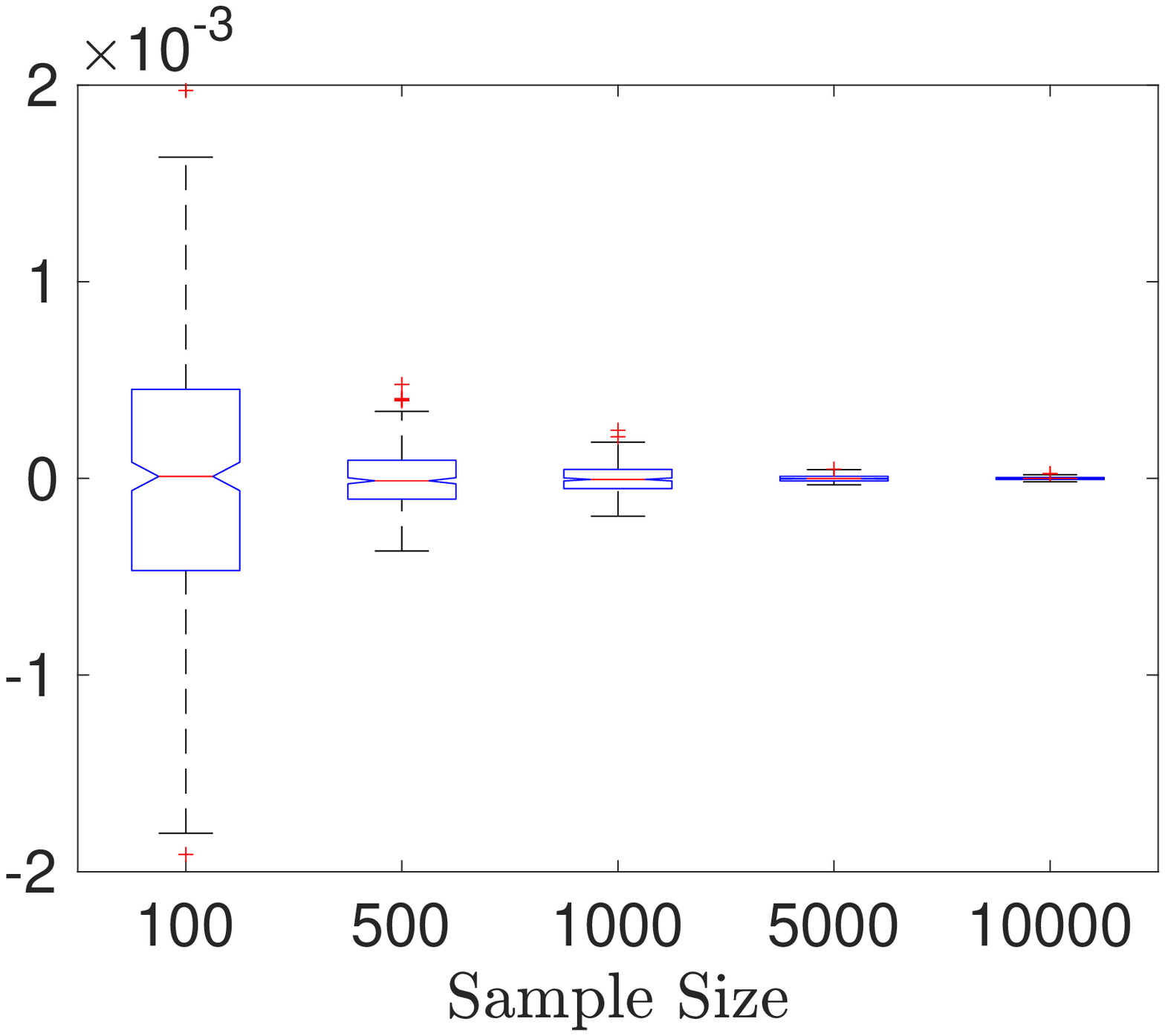}
        \caption{K=500}
    \end{subfigure}
    ~
    \begin{subfigure}[b]{0.3\textwidth}
        \includegraphics[width=\textwidth]{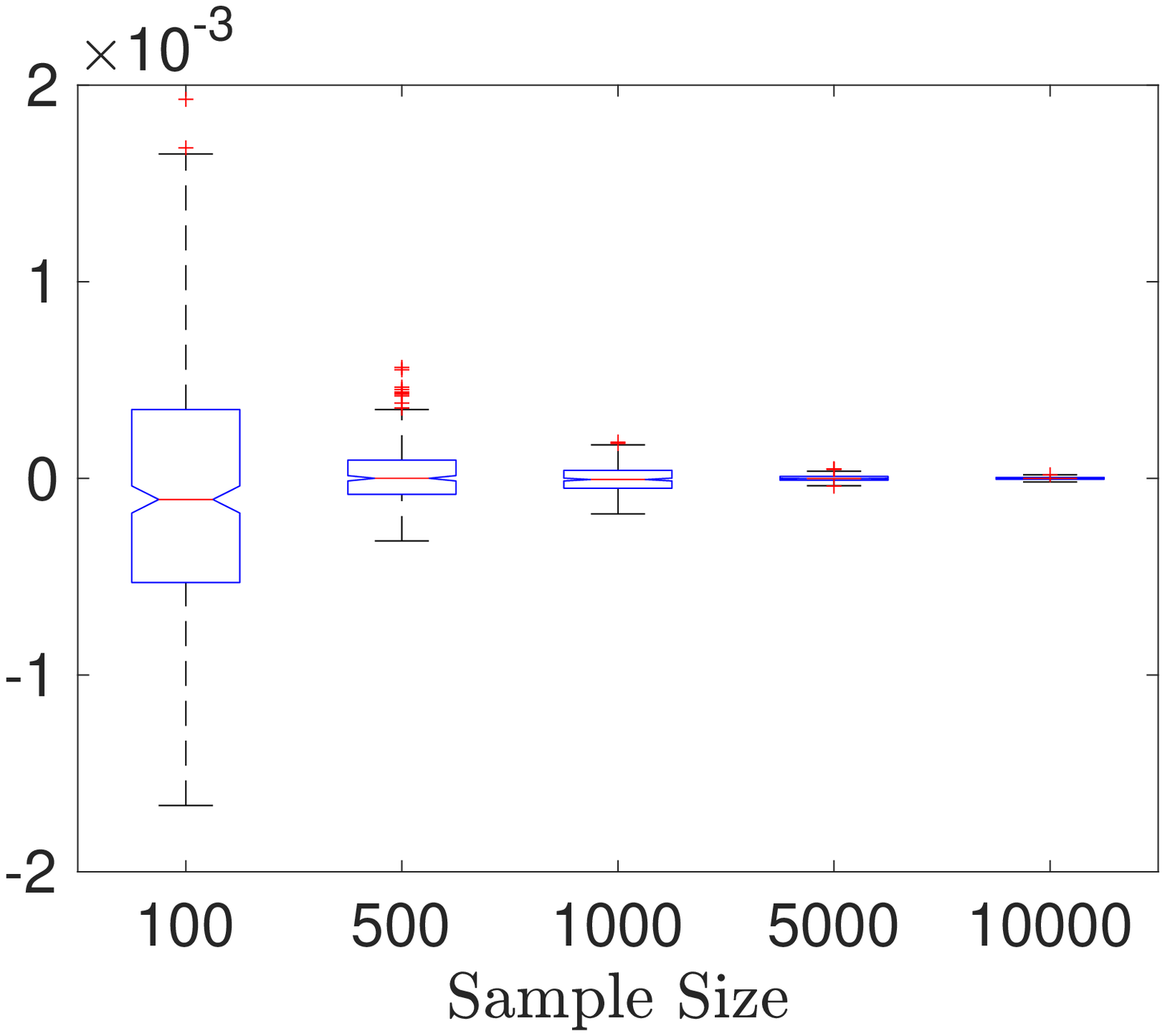}
        \caption{K=1000}
    \end{subfigure}
    \caption{Boxplots of estimators in Example \ref{ex:1}. Dimensions of $X$ and $Y$ are fixed to be $p=q=10$; the result is based on $400$ repeated experiments.}
    \label{fig:1}
\end{figure}

The following example is to study dependent random variables.
\begin{example}
We generate random vectors $X \in \mathbb{R}^{10}$ and $Y \in \mathbb{R}^{10}$.
Each entry $X_i$ follows $\mbox{Unif}(0,1)$, independently.
Let $Y_i$ denote the $i$-th entry of $Y$.
We let $Y_1 = X_1^2$ and $Y_2 = X_2^2$.
For the rest entry of $Y$, we have $Y_i = Z_i^2$, $i=3,\ldots,10$, where $Z_i$ follows $\mbox{Unif}(0,1)$, independently.
\label{ex:2}
\end{example}
See Figure \ref{fig:2} for the boxplots of the outcomes of Example \ref{ex:2}.
In each subfigure, we fix the number of the Monte Carlo iterations $K$ and let the number of observations $n$ grow.
This example shows that when $K$ is fixed, the variation of the estimator remains regardless of the sample size $n$. In the dependent cases, the number of the Monte Carlo iterations $K$ plays a more important role in estimator convergence than sample size $n$.
\begin{figure}[ht!]
    \centering
    \begin{subfigure}[b]{0.3\textwidth}
        \includegraphics[width=\textwidth]{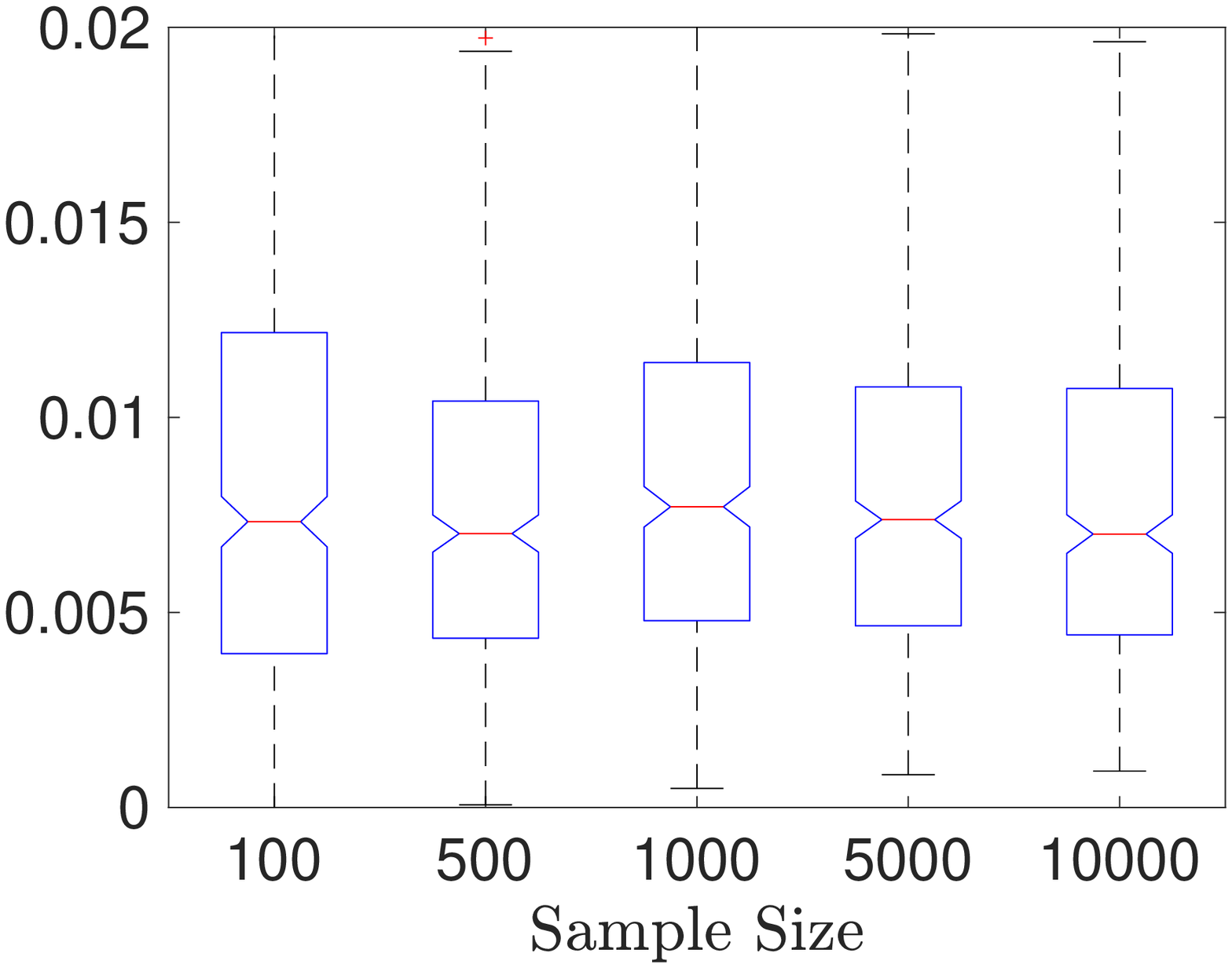}
        \caption{K=10}
    \end{subfigure}
    ~
    \begin{subfigure}[b]{0.3\textwidth}
        \includegraphics[width=\textwidth]{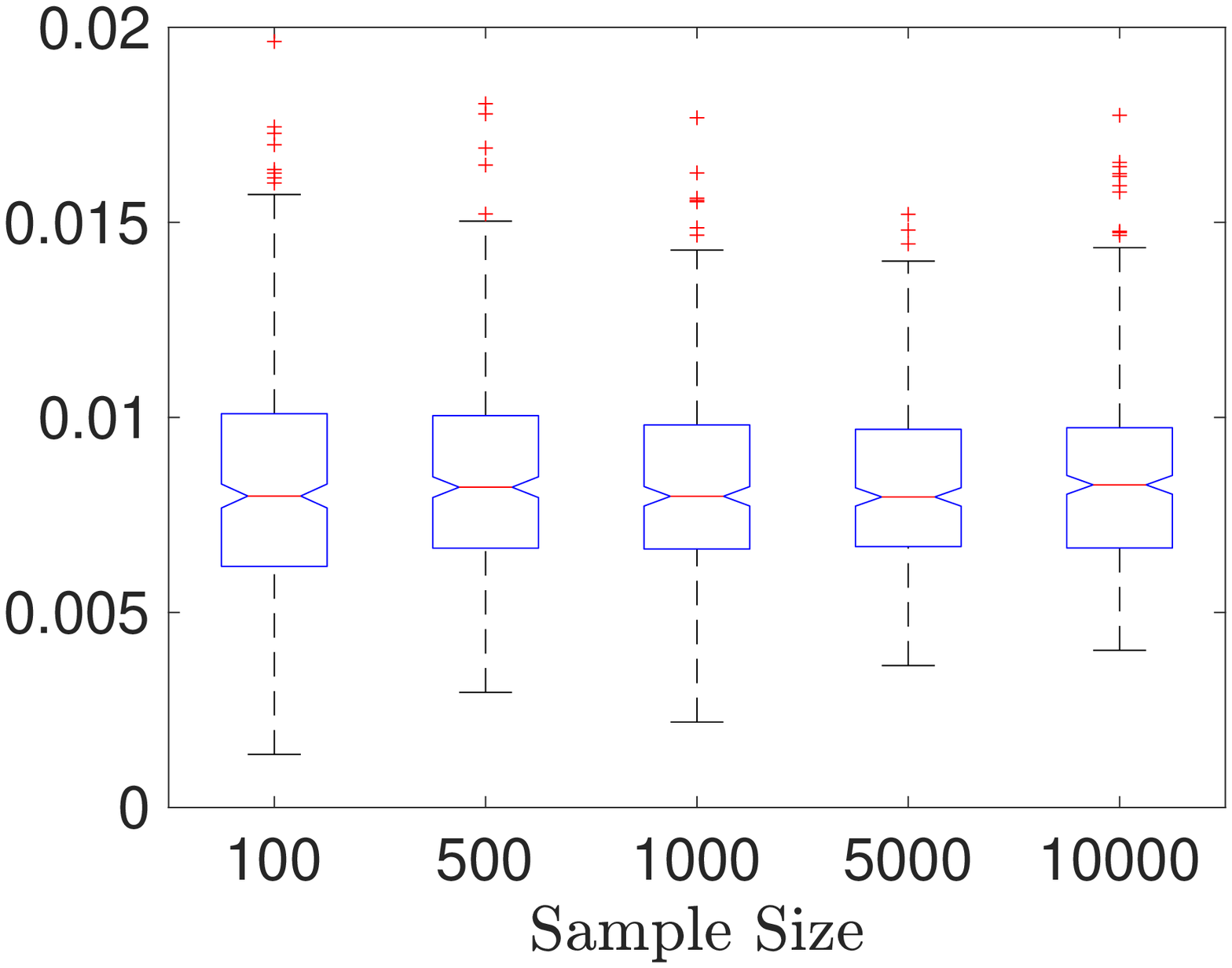}
        \caption{K=50}
    \end{subfigure}
    ~
    \begin{subfigure}[b]{0.3\textwidth}
        \includegraphics[width=\textwidth]{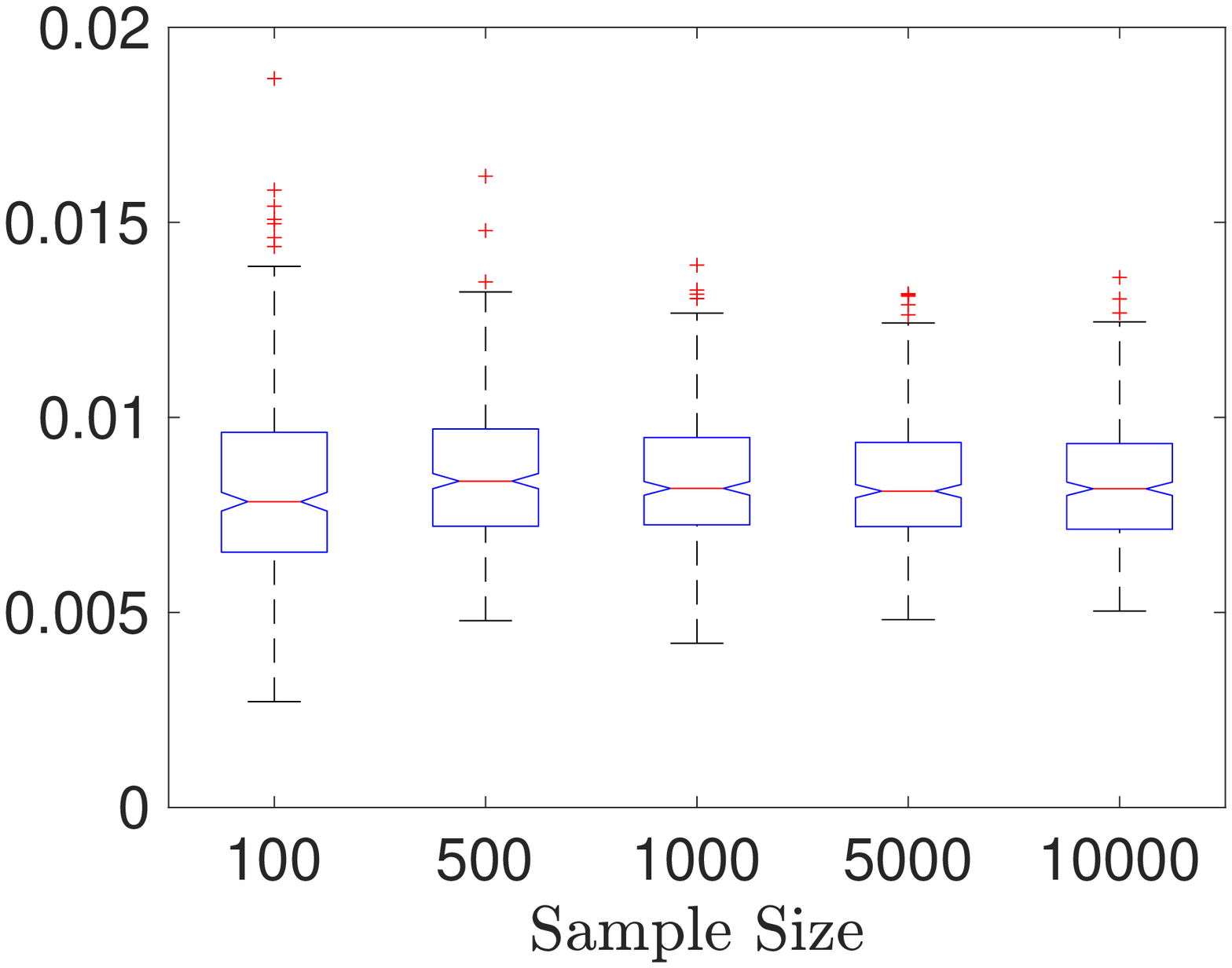}
        \caption{K=100}
    \end{subfigure}
    \\
    \begin{subfigure}[b]{0.3\textwidth}
        \includegraphics[width=\textwidth]{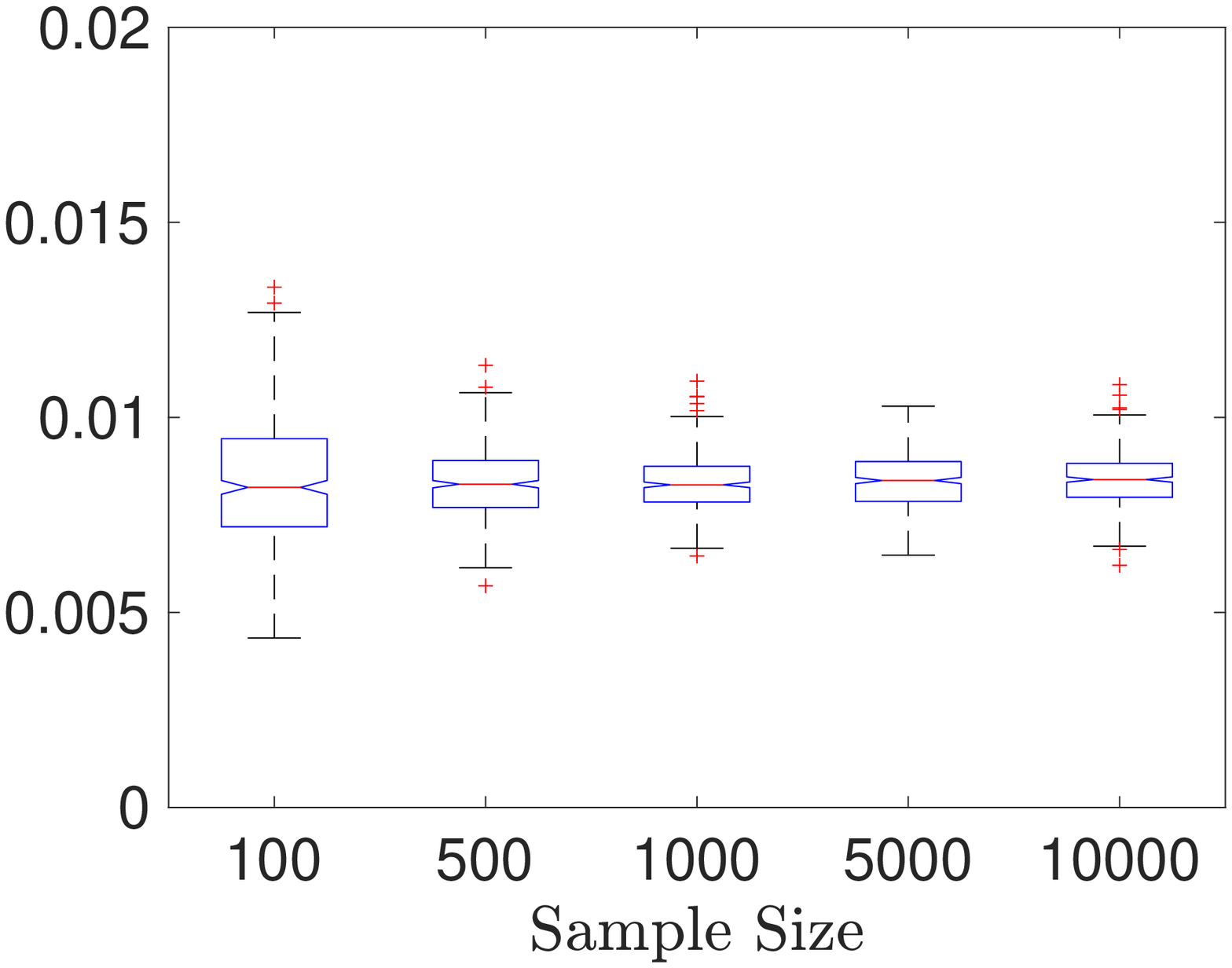}
        \caption{K=500}
    \end{subfigure}
    ~
    \begin{subfigure}[b]{0.3\textwidth}
        \includegraphics[width=\textwidth]{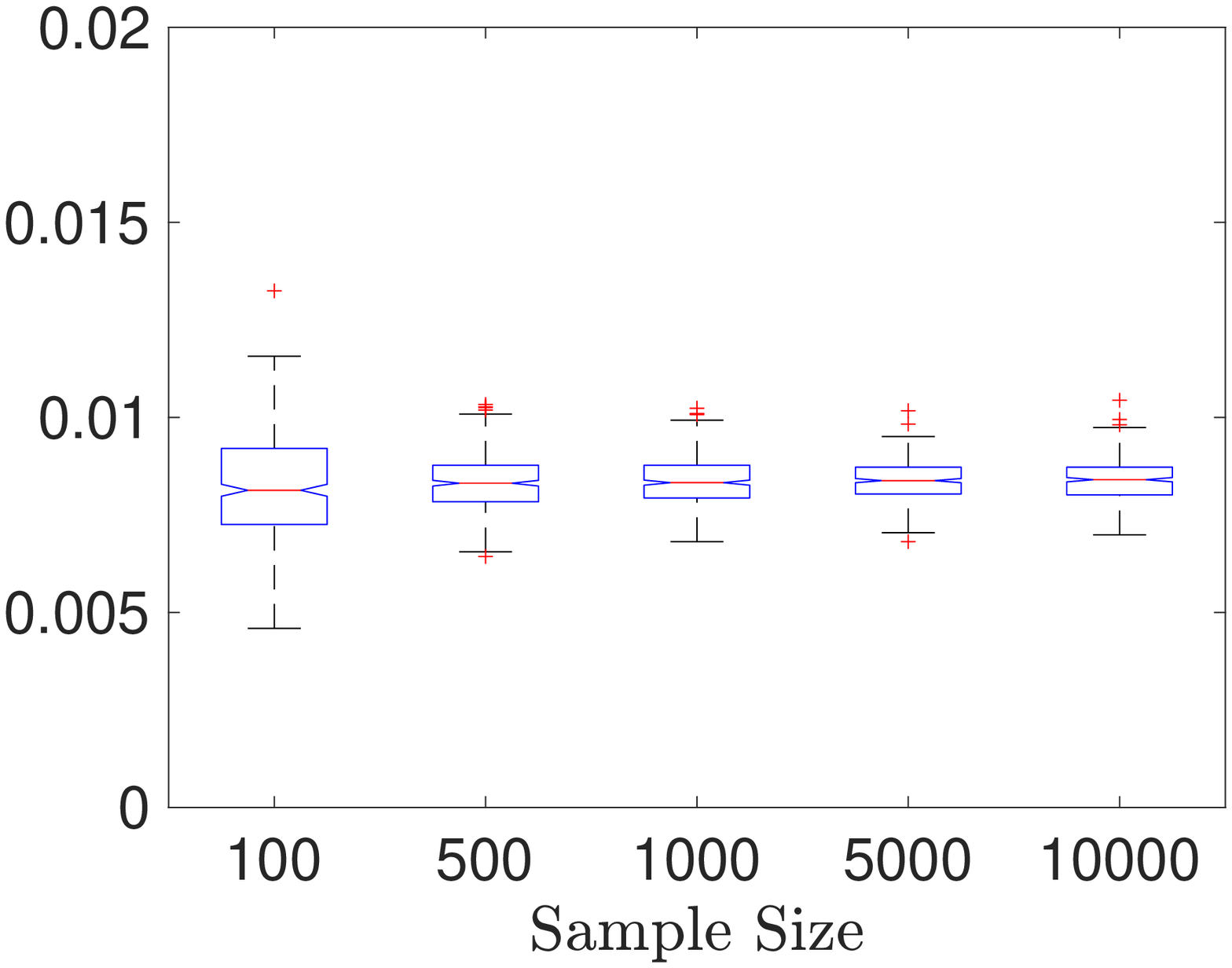}
        \caption{K=1000}
    \end{subfigure}
    \caption{Boxplots of our estimators in Example \ref{ex:2}. Dimension of $X$ and $Y$ are fixed to be $p=q=10$; the result is based on $400$ repeated experiments.}
    \label{fig:2}
\end{figure}

The outcomes of Example \ref{ex:1} and \ref{ex:2} confirm the theoretical results that the proposed estimator converges to $0$ as sample size $n$ grows in the independent case; and converges to some nonzero number as the number of the Monte Carlo iterations $K$ grows in the dependent case.

In the following two examples, we fix the sample size $n=2000$ as we noticed that our method is more efficient than direct method when $n$ is large.
We fix the number of the Monte Carlo iterations $K=50$ and relax the restriction on the data dimensions to allow $p \neq q$ and let $p$ and $q$ vary in $(10,50,100,500,1000)$.
We continue on with an independent case as follows.
\begin{example}
We generate random vectors $X \in \mathbb{R}^{p}$ and $Y \in \mathbb{R}^{q}$. Each entry of $X$ follows $\mbox{Unif}(0,1)$, independently. Each entry $Y_i = Z_i^2$, where $Z_i$ follows $\mbox{Unif}(0,1)$, independently. \label{ex:3}
\end{example}

See Figure \ref{fig:3} for the boxplots of the outcomes of Example \ref{ex:3}.
In each subfigure, we fix the dimension of $X$ and let the dimension of $Y$ grow.
It is worth noting that the scale of each subfigure could be different in order to display the entire boxplots.
It shows that the proposed estimator converges fairly fast in the independent case regardless of the dimension of the data.
\begin{figure}[ht!]
    \centering
    \begin{subfigure}[b]{0.3\textwidth}
        \includegraphics[width=\textwidth]{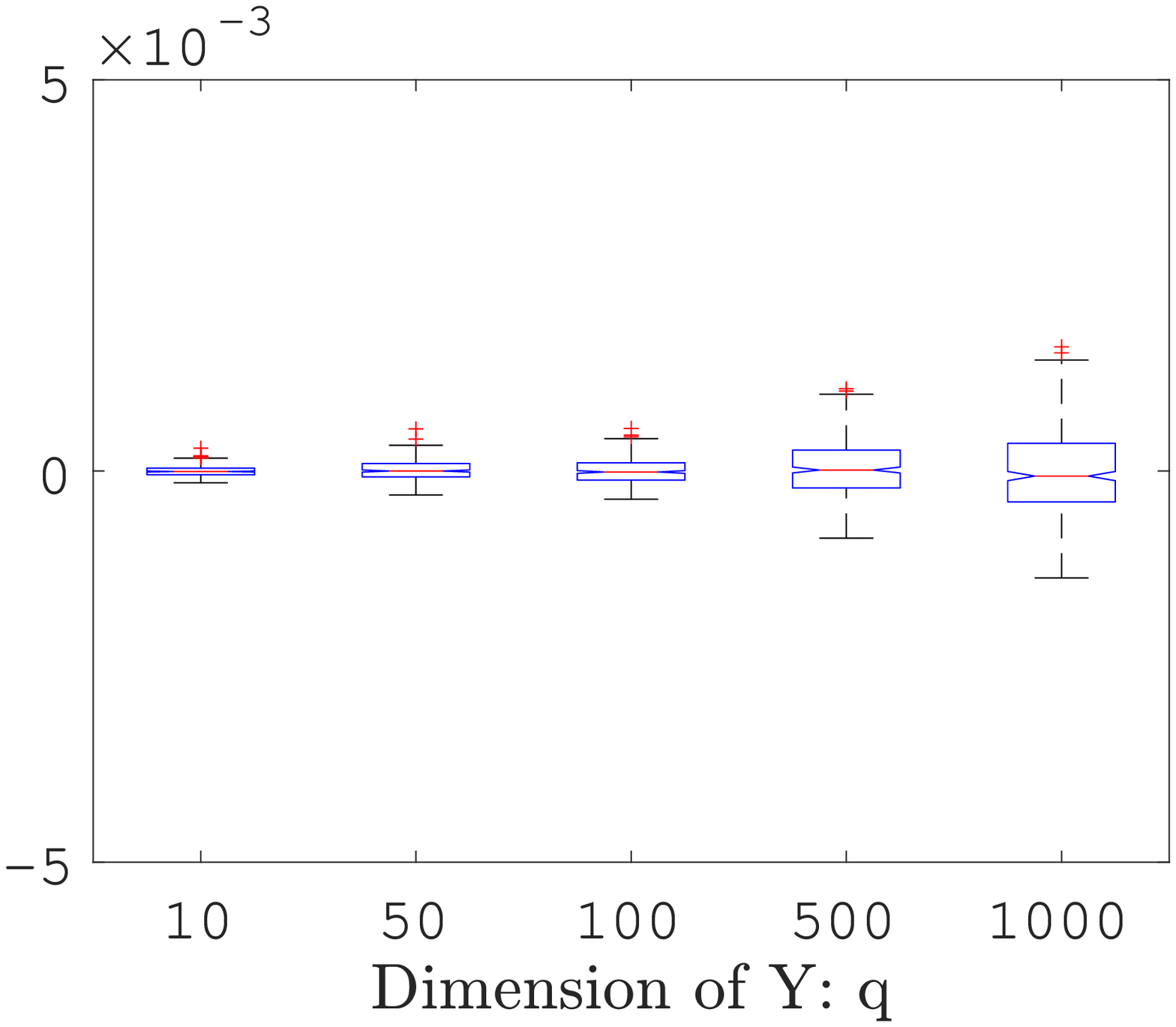}
        \caption{p=10}
    \end{subfigure}
    ~
    \begin{subfigure}[b]{0.3\textwidth}
        \includegraphics[width=\textwidth]{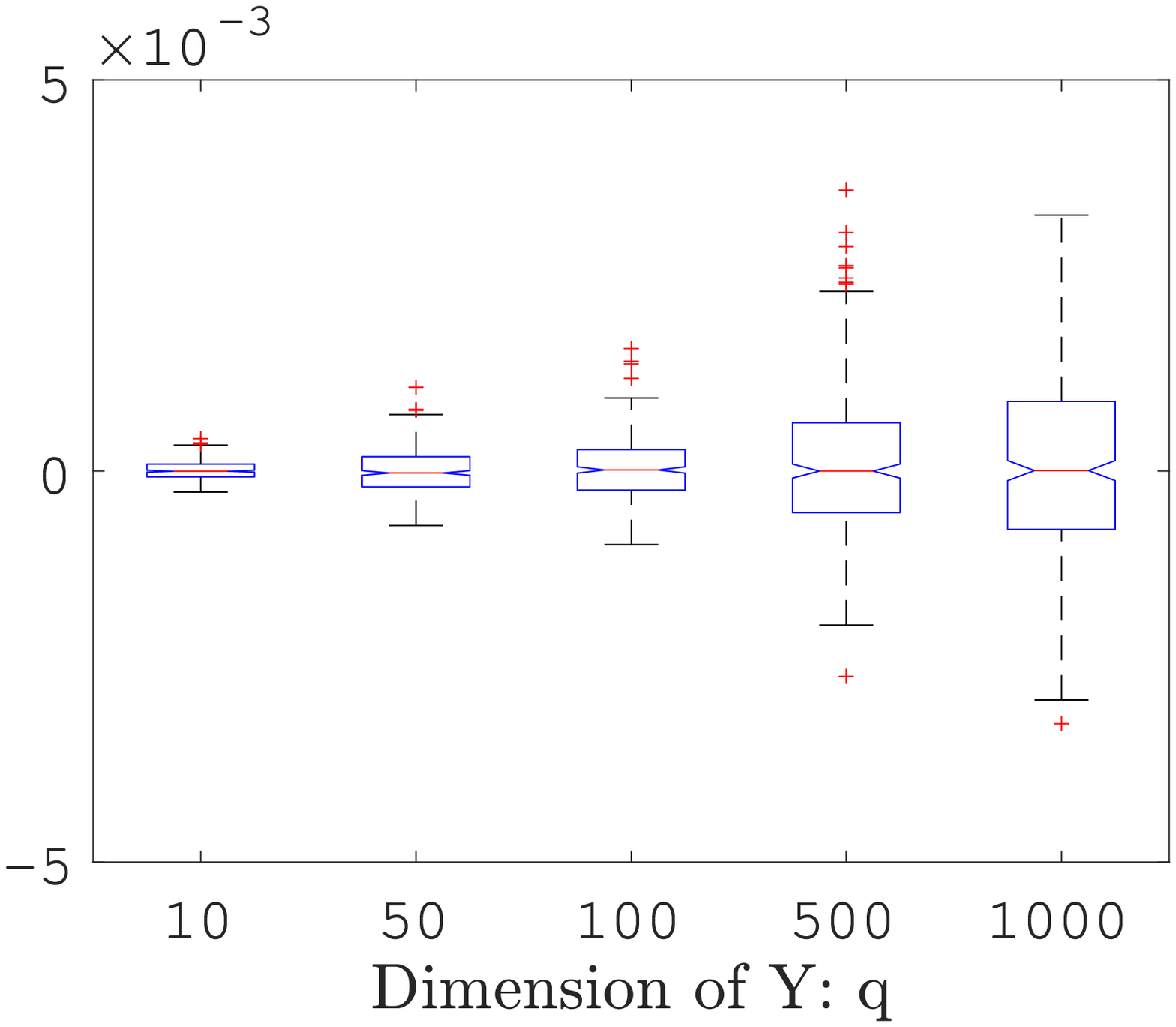}
        \caption{p=50}
    \end{subfigure}
    ~
    \begin{subfigure}[b]{0.3\textwidth}
        \includegraphics[width=\textwidth]{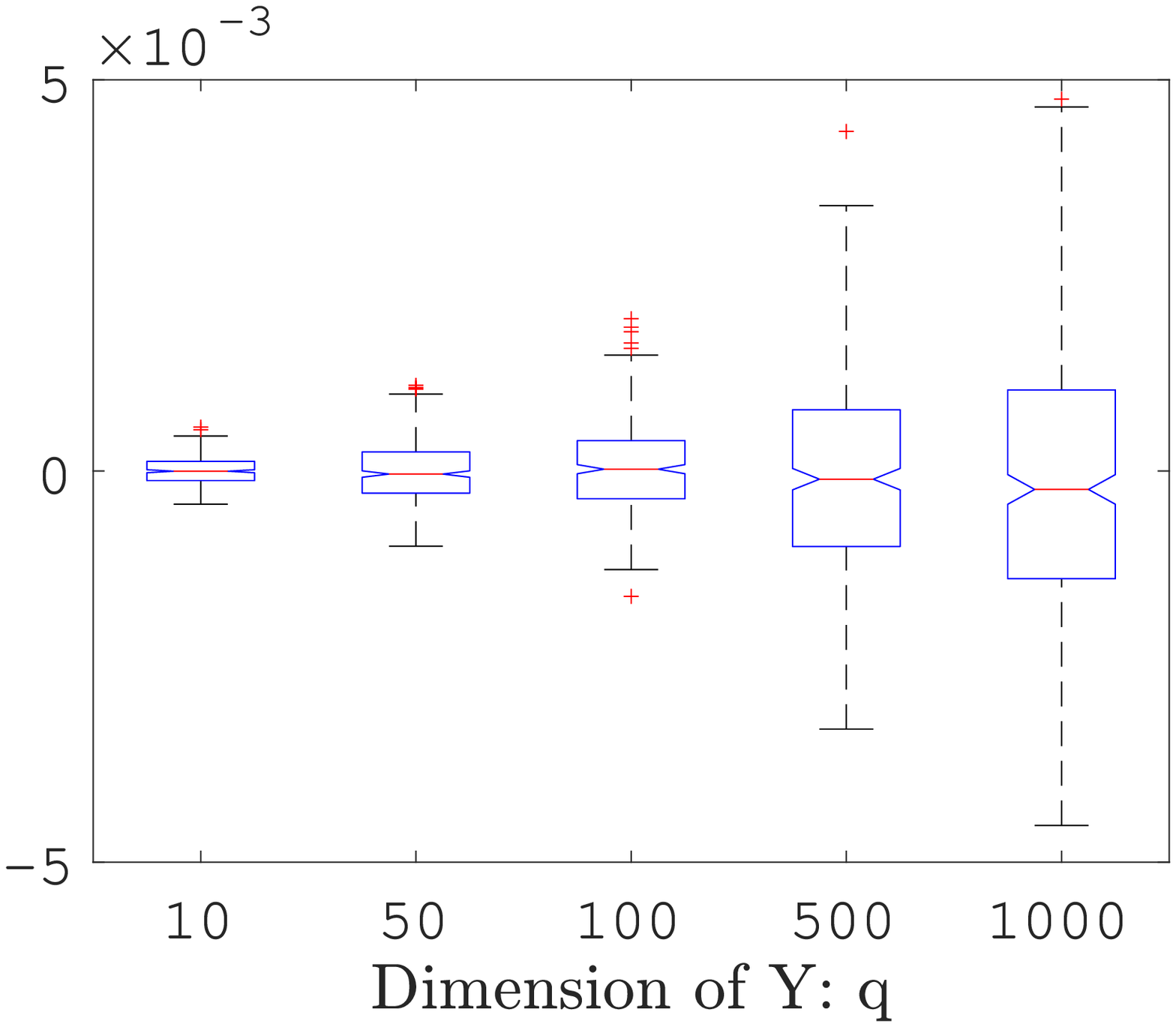}
        \caption{p=100}
    \end{subfigure}
    \\
    \begin{subfigure}[b]{0.3\textwidth}
        \includegraphics[width=\textwidth]{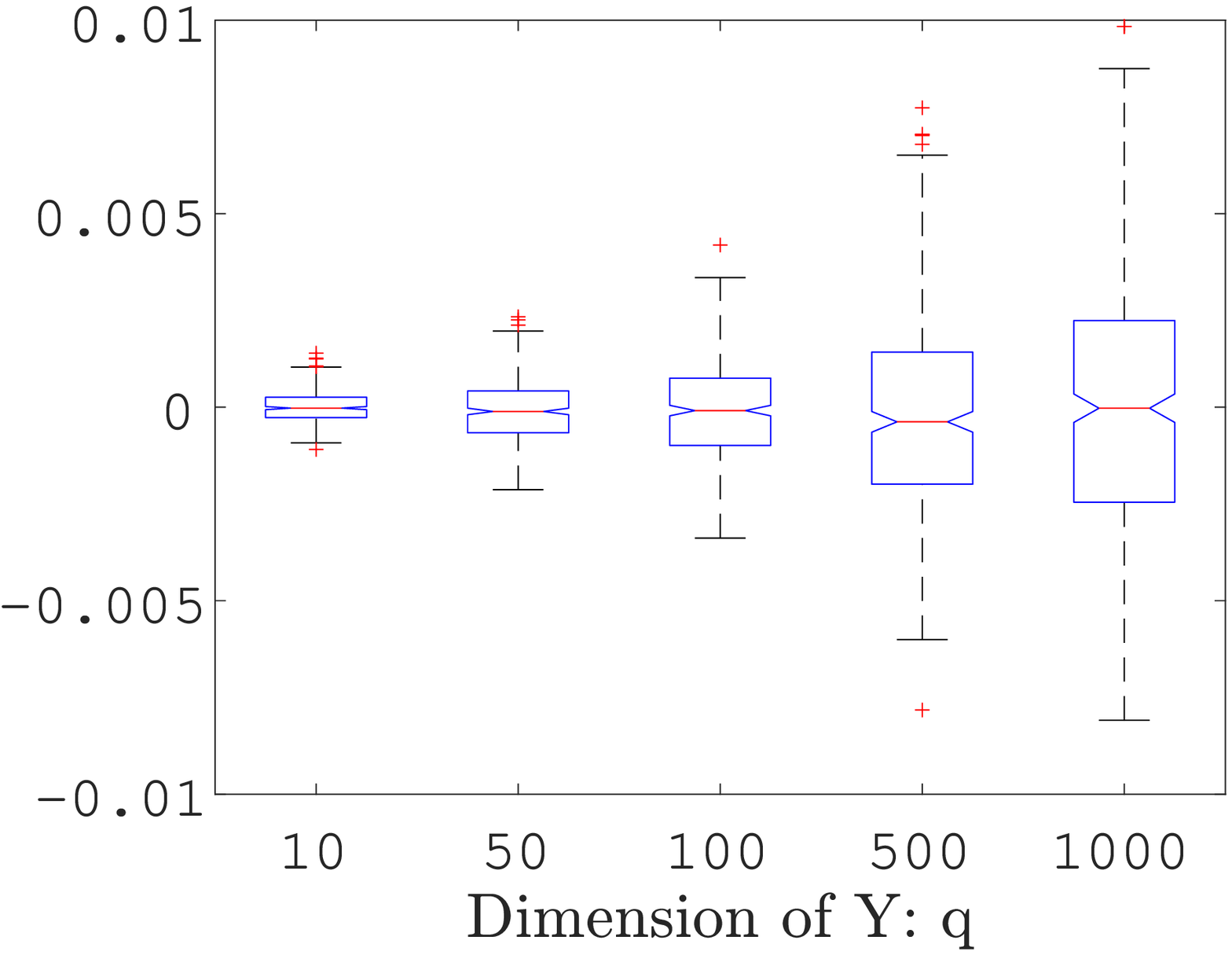}
        \caption{p=500}
    \end{subfigure}
    ~
    \begin{subfigure}[b]{0.3\textwidth}
        \includegraphics[width=\textwidth]{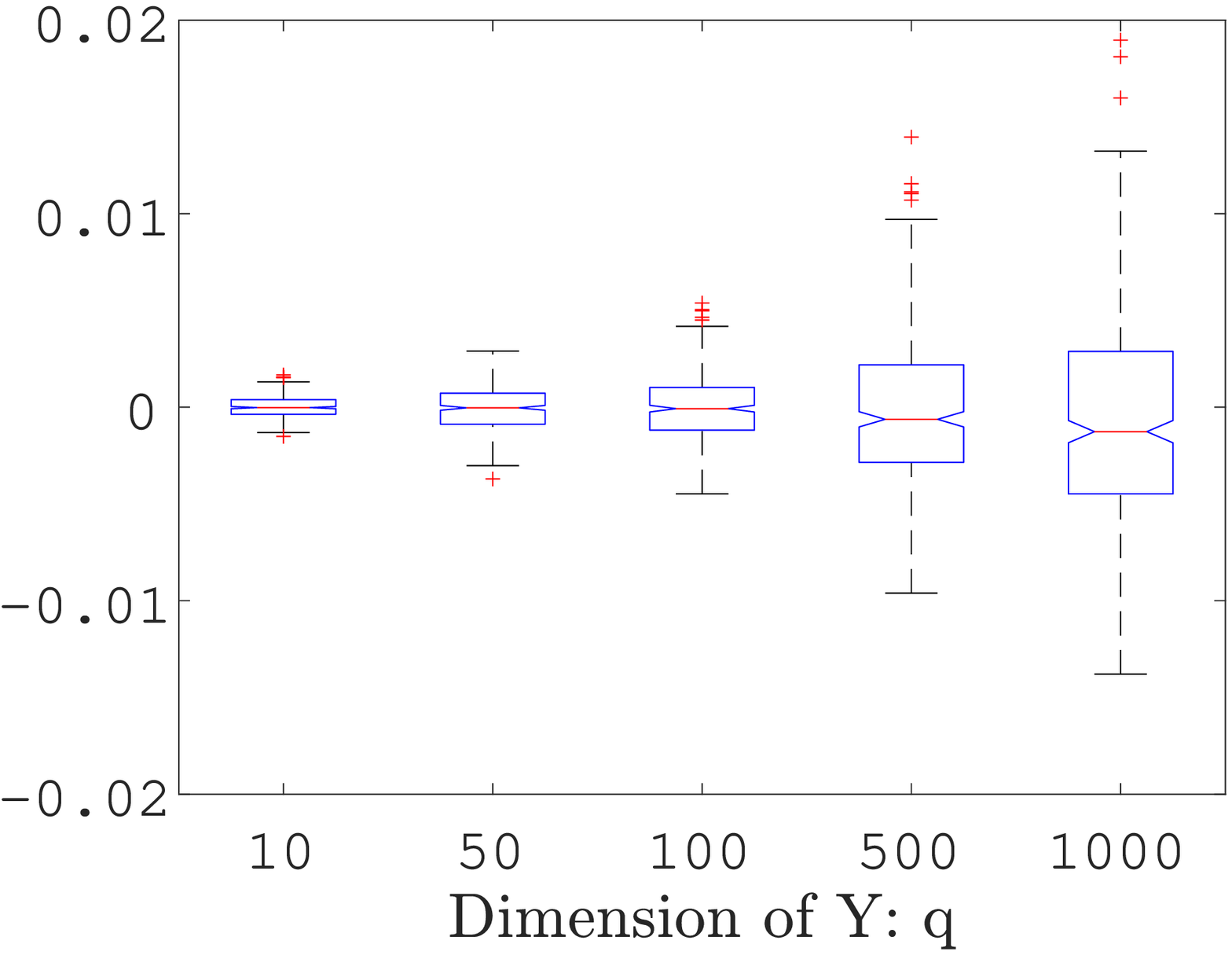}
        \caption{p=1000}
    \end{subfigure}
    \caption{Boxplot of Estimators in Example \ref{ex:3}: both sample size and the number of Monte Carlo iterations is fixed, $n=2000$, $K=50$; the result is based on $400$ repeated experiments.}
    \label{fig:3}
\end{figure}

The following presents a dependent case. In this case, only a small number of entries in $X$ and $Y$ are dependent, which means that the dependency structure between $X$ and $Y$ is low-dimensional though $X$ or $Y$ could be of high dimensions.
\begin{example}
We generate random vectors $X \in \mathbb{R}^{p}$ and $Y \in \mathbb{R}^{q}$.
Each entry of $X$ follows $\mbox{Unif}(0,1)$, independently.
We let the first 5 entries of $Y$ to be the square of first 5 entries of $X$ and let the rest entries of $Y$ to be the square of some independent $\mbox{Unif}(0,1)$ random variables.
Specifically, we let $Y_i = X_{i}^2, i=1,\ldots,5$, and, $Y_i = Z_i^2, i = 6,\ldots,q$, where $Z_i$'s are drawn independently from $\mbox{Unif}(0,1)$.
\label{ex:4}
\end{example}
See Figure \ref{fig:4} for the boxplots of the outcomes of Example \ref{ex:4}.
In each subfigure, we fix the dimension of $X$ and let the dimension of $Y$ grow.
The test power of proposed test against data dimensions can be seen in Table \ref{tb:ex4}.
It is worth noting that when sample size is fixed, the test power of our method decays as the dimension of $X$ and $Y$ increase.
We use the Direct Distance Covariance (DDC) defined in \eqref{eq:def_5} on the same data. As a contrast, the test power of DDC is 1.000 even $p=q=1000$.
This example raises a limitation of random projection: it may fail to detect the low dimensional dependency in high dimensional data.
A possible remedy for this issue is performing dimension reduction before applying the proposed method.
We do not research further along this direction since it is beyond the scope of this paper.
\begin{figure}[ht!]
    \centering
    \begin{subfigure}[b]{0.3\textwidth}
        \includegraphics[width=\textwidth]{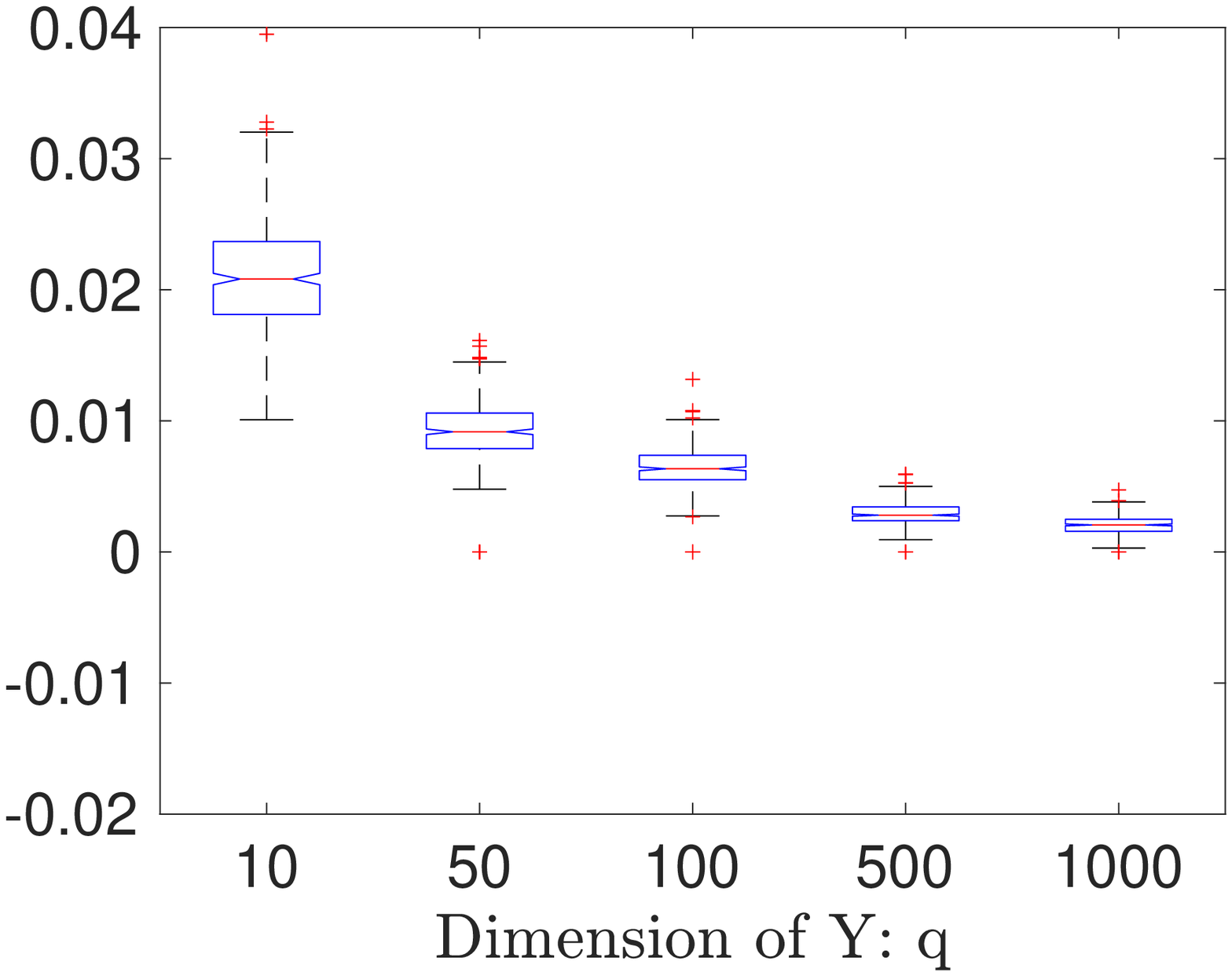}
        \caption{p=10}
    \end{subfigure}
    ~
    \begin{subfigure}[b]{0.3\textwidth}
        \includegraphics[width=\textwidth]{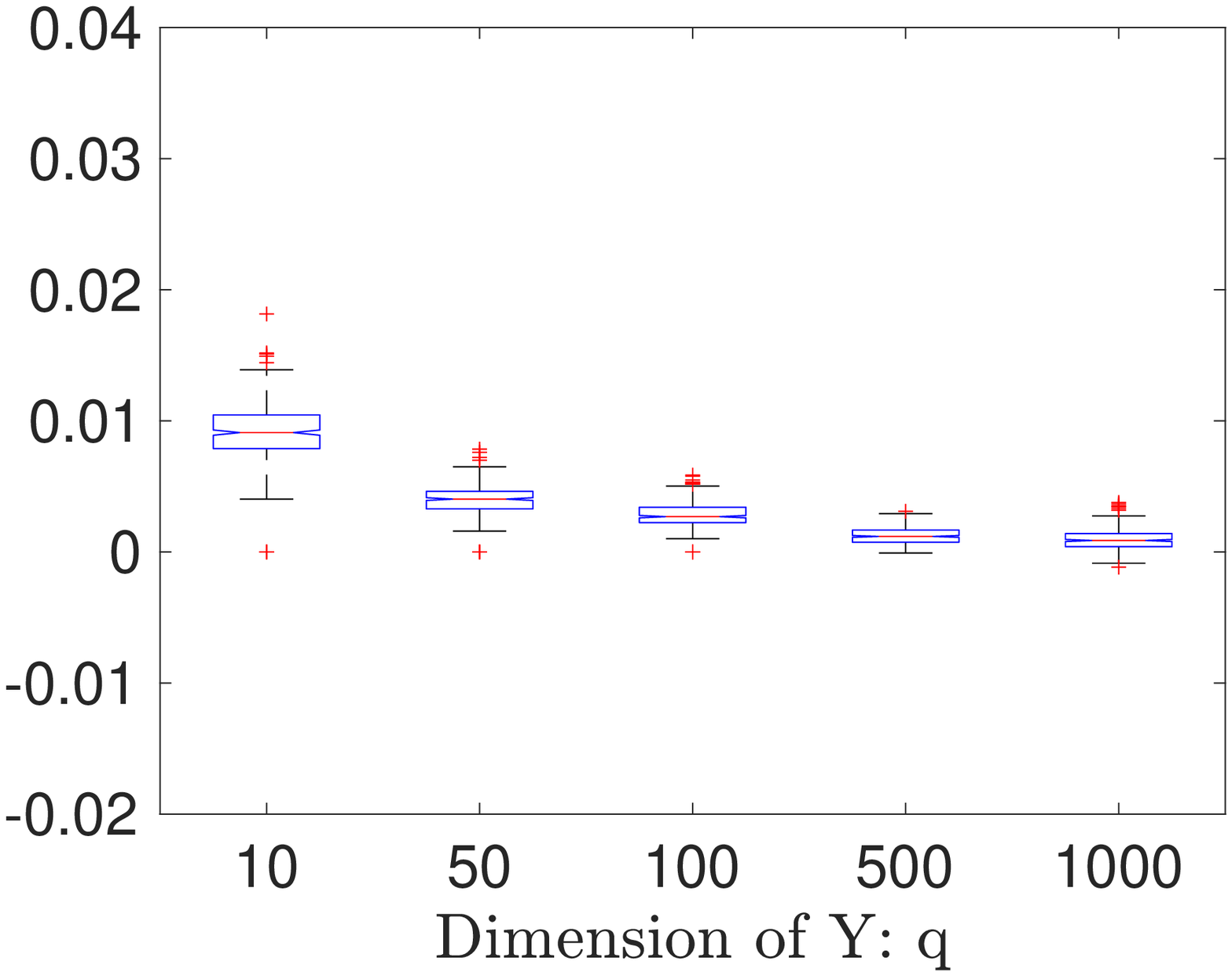}
        \caption{p=50}
    \end{subfigure}
    ~
    \begin{subfigure}[b]{0.3\textwidth}
        \includegraphics[width=\textwidth]{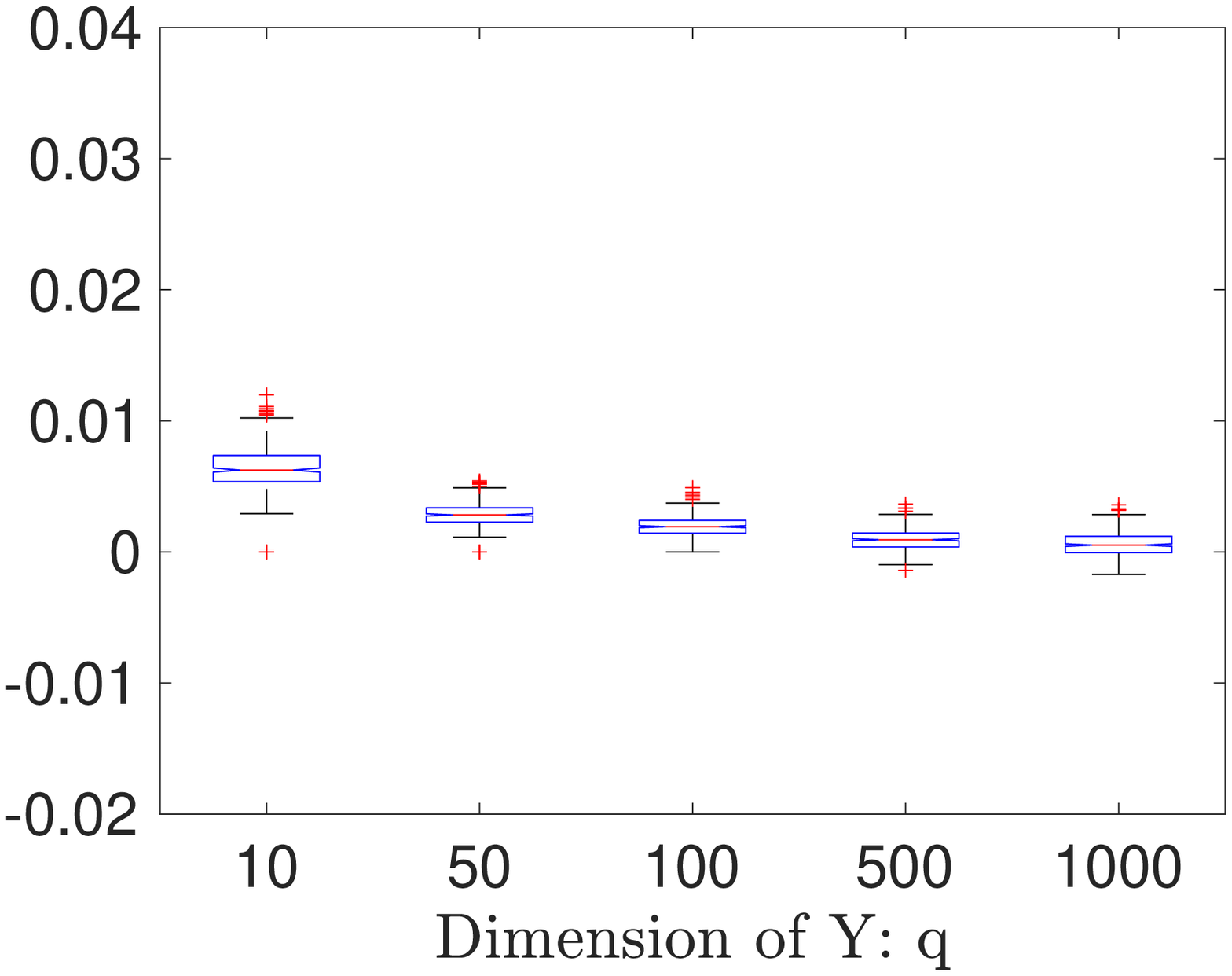}
        \caption{p=100}
    \end{subfigure}
    \\
    \begin{subfigure}[b]{0.3\textwidth}
        \includegraphics[width=\textwidth]{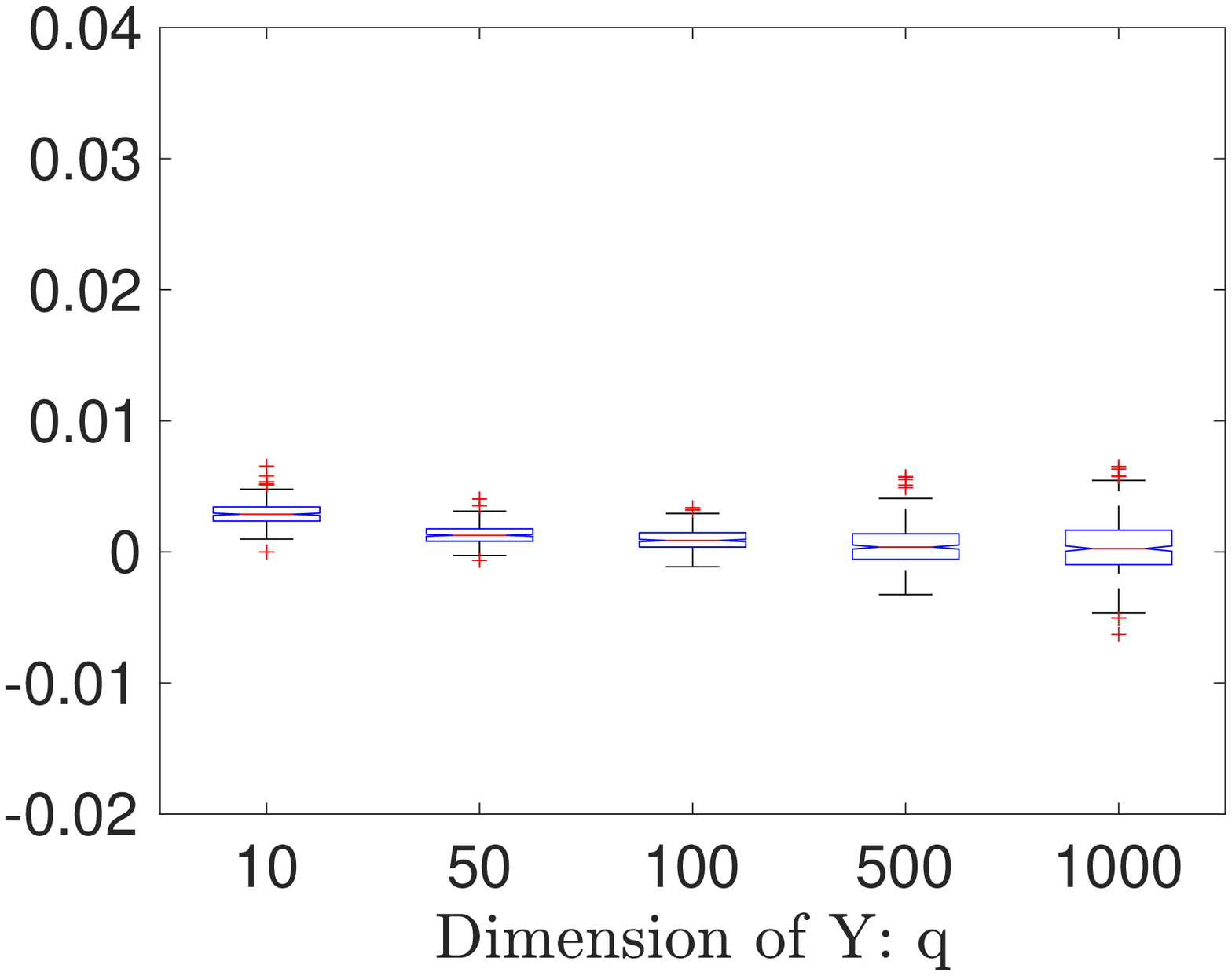}
        \caption{p=500}
    \end{subfigure}
    ~
    \begin{subfigure}[b]{0.3\textwidth}
        \includegraphics[width=\textwidth]{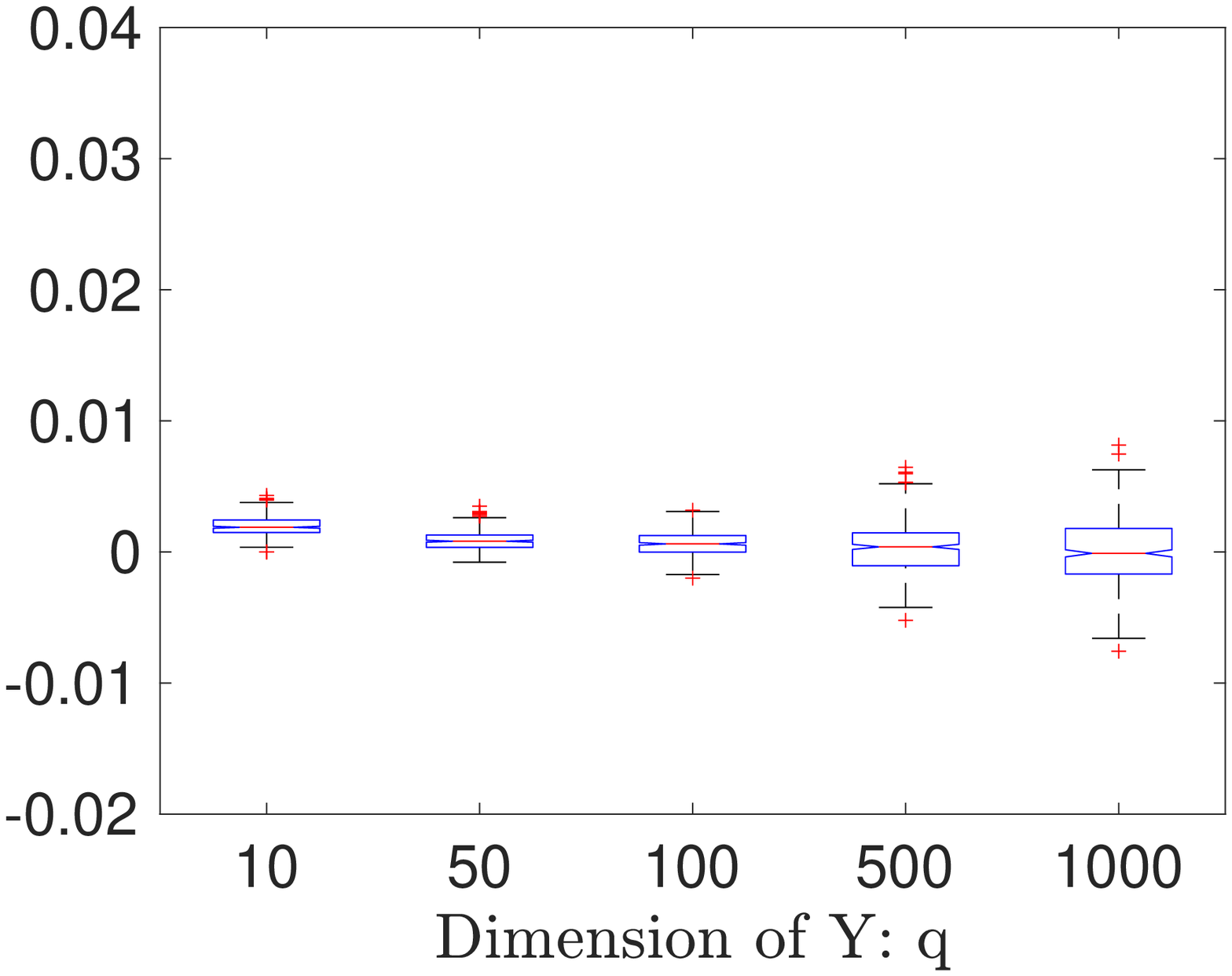}
        \caption{p=1000}
    \end{subfigure}
    \caption{Boxplots of the proposed estimators in Example \ref{ex:4}: both sample size and the number of the Monte Carlo iterations are fixed: $n=2000$ and $K=50$; the result is based on $400$ repeated experiments.}
    \label{fig:4}
\end{figure}

\begin{table}[ht!]
	\centering
	\begin{tabular}{|c|ccccc|}
	\hline
	\multirow{2}{*}{Dimension of $X$: p} & \multicolumn{5}{c|}{Dimension of $Y$: q} \\
	& 10 & 50 & 100 & 500 & 1000 \\
	\hline
	10 & 1.0000 & 1.0000 & 1.0000 & 1.0000 & 0.9975 \\
	50 & 1.0000 & 1.0000 & 1.0000 & 0.7775 & 0.4650 \\
	100 & 1.0000 & 1.0000 & 0.9925 & 0.4875 & 0.1800 \\
	500 & 0.9950 & 0.8150 & 0.4425 & 0.1225 & 0.0975 \\
	1000 & 0.9900 & 0.4000 & 0.2125 & 0.0900 & 0.0475 \\
	\hline
	\end{tabular}
	\caption{Test Power in Example \ref{ex:4}: this result is based $400$ repeated experiments; the significant level is $0.05$.}
	\label{tb:ex4}
\end{table}
Note this paper focuses on independence testing.
Therefore the independent case is of more relevance.

\subsection{Comparison with Direct Method}
\label{sec:efficiency-compare}
In this section, we would like to illustrate the computational and space efficiency of the proposed method (RPDC).
RPDC is much faster than the direct method (DDC, eq. \eqref{eq:def_5}) when the sample size is large.
It is worth noting that DDC is infeasible when the sample size is too large as its space complexity is $O(n^2)$.
See Table \ref{tb:sp_comp} for a comparison of computing time (unit: second) against the sample size $n$. This experiment is run on a laptop (MacBook Pro Retina, 13-inch, Early 2015, 2.7 GHz Intel Core i5, 8 GB 1867 MHz DDR3) with MATLAB R2016b (9.1.0.441655).
\begin{table}[ht!]
	\centering
	\begin{tabular}{|r|c|c|}
	\hline
	\multirow{2}{*}{Sample size} & \multirow{2}{*}{$\Omega_n$} & \multirow{2}{*}{$\overline{\Omega}_n$} \\
	& & \\
	\hline
	100 & 0.0043 (0.0047) & 0.0207 (0.0037) \\
	500 & 0.0210 (0.0066) & 0.0770 (0.0086) \\
	1000 & 0.0624 (0.0047) & 0.1685 (0.0141) \\
	2000 & 0.2349 (0.0133) & 0.3568 (0.0169) \\
	4000 & 0.9184 (0.0226) & 0.7885 (0.0114) \\
	8000 & 7.2067 (0.4669) & 1.7797 (0.0311) \\
	16000 & --- & 3.7539 (0.0289) \\
	\hline
	\end{tabular}
	\caption{Speed Comparison: the Direct Distance Covariance ($\Omega_n$) versus the Randomly Projected Distance Covariance ($\overline{\Omega}_n$).
This table is based on $100$ repeated experiments, the dimensions of $X$ and $Y$ are fixed to be $p=q=10$ and the number of Monte Carlo iterations in RPDC is $K=50$.
The numbers outside the parentheses are the average and the numbers inside the parentheses are the sample standard deviations.}
	\label{tb:sp_comp}
\end{table}

\subsection{Comparison with Other Independence Tests}
\label{sec:test-compare}

In this part, we compare the statistical test power of the proposed test (RPDC) with Hilbert-Schmidt Independence Criterion (HSIC) (\cite{gretton2005measuring}) as HSIC is gaining attention in machine learning and statistics communities. We also compare with Randomized Dependence Coefficient (RDC) (\cite{lopez2013randomized}), which utilizes the technique of random projection as we do.
Two classical tests for multivariate independence, which are described below, are included in the comparison, as well as the Direct Distance Covariance (DDC) defined in \eqref{eq:def_5}.
\begin{itemize}
	\item Wilks Lambda (WL): the likelihood ratio test of hypotheses $\Sigma_{12}=0$ with $\mu$ unknown is based on
	$$
	\frac{\mbox{det}(S)}{\mbox{det}(S_{11})\mbox{det}(S_{22})} = \frac{\mbox{det}(S_{22}-S_{21}S_{11}^{-1}S_{12})}{\mbox{det}(S_{22})},
	$$
	where $\mbox{det}(\cdot)$ is the determinant, $S$, $S_{11}$ and $S_{22}$ denote the sample covariances of $(X,Y)$, $X$ and $Y$, respectively, and $S_{12}$ is the sample covariance $\hat{\mbox{Cov}}(X,Y)$. Under multivariate normality, the test statistic
	$$
	W = - n \log \mbox{det}(I - S_{22}^{-1} S_{21} S_{11}^{-1} S_{12})
	$$
	has the Wilks Lambda distribution $\Lambda(q,n-1-p,p)$, see \cite{wilks1935independence}.
	\item Puri-Sen (PS) statistics: \cite{puri1971nonparametric}, Chapter 8, proposed similar tests based on more general sample dispersion matrices $T$. In that test $S,S_{11},S_{12} \text{ and } S_{22}$ are replaced by $T,T_{11},T_{12} \text{ and } T_{22}$, where $T$ could be a matrix of Spearman's rank correlation statistics. Then, the test statistic becomes
	$$
	W = - n \log \mbox{det}(I - T_{22}^{-1} T_{21} T_{11}^{-1} T_{12}).
	$$
\end{itemize}
The critical values of the Wilks Lambda (WL) and Puri-Sen (PS) statistics are given by Bartlett's approximation (\cite{mardia1982multivariate}, Section 5.3.2b): if $n$ is large and $p,q>2$, then
$$
-(n - \frac{1}{2}(p+q+3) ) \log \mbox{det}(I - S_{22}^{-1} S_{21} S_{11}^{-1} S_{12})
$$
has an approximate $\chi^2(pq)$ distribution.

The reference distributions of RDC and HSIC are approximated by $200$ permutations.
And the reference distributions of DDC and RPDC are approximated by the Gamma Distribution.
The significant level is set to be $\alpha_s=0.05$ and each experiment is repeated for $N=400$ times to get reliable type-I error / test power.

We start with an example that $(X,Y)$ is multivariate normal. In this case, WL and PS are expected to be optimal as the distributional assumptions of these two classical tests are satisfied.
Surprisingly, DDC has comparable performance with the aforementioned two methods. RPDC can achieve satisfactory performance when sample size is a reasonably large.
\begin{example}
We set the dimension of the data to be $p=q=10$.
We generate random vectors $X \in \mathbb{R}^{10}$ and $Y \in \mathbb{R}^{10}$ from the standard multivariate normal distribution $\mathcal{N}(0,\mathbf{I}_{10})$.
The joint distribution of $(X,Y)$ is also normal and we have $\mbox{Cor}(X_i, Y_i) = \rho, i = 1,\ldots, 10$, and the rest correlation are all $0$. We set the value of $\rho$ to be $0$ and $0.1$ to represent independent and correlated scenarios, respectively.
The sample size $n$ is set to be from $100$ to $1500$ with an increment of $100$.
\label{ex:5}
\end{example}
Figure \ref{fig:5} plots the type-I error in subfigure (a) and test power in subfigure (b) against sample size.
In the independence case ($\rho=0.0$), the type-I error of each test is always around the significance level $\alpha_s = 0.05$, which implies the Gamma approximation works well for the asymptotic distributions.
In the dependent case ($\rho=0.1$), the overall performance of RPDC is close to HSIC and RPDC outperforms when sample size is smaller and underperforms when sample size is larger.
Unfortunately, RDC's test power is unsatisfactory.
\begin{figure}[ht!]
    \centering
    \begin{subfigure}[b]{0.45\textwidth}
        \includegraphics[width=\textwidth]{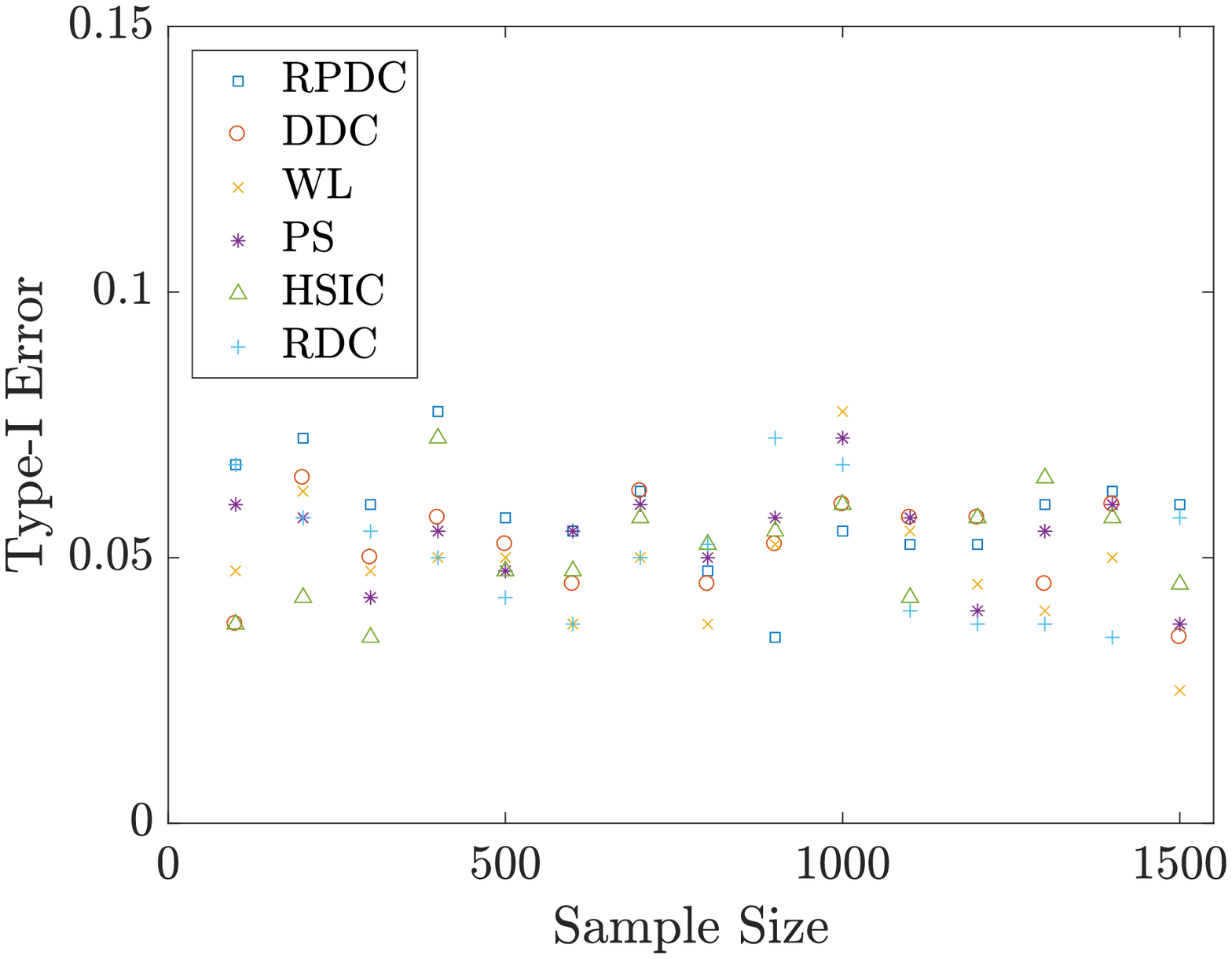}
        \caption{Independence: $\rho=0.0$}
    \end{subfigure}
    ~
    \begin{subfigure}[b]{0.45\textwidth}
        \includegraphics[width=\textwidth]{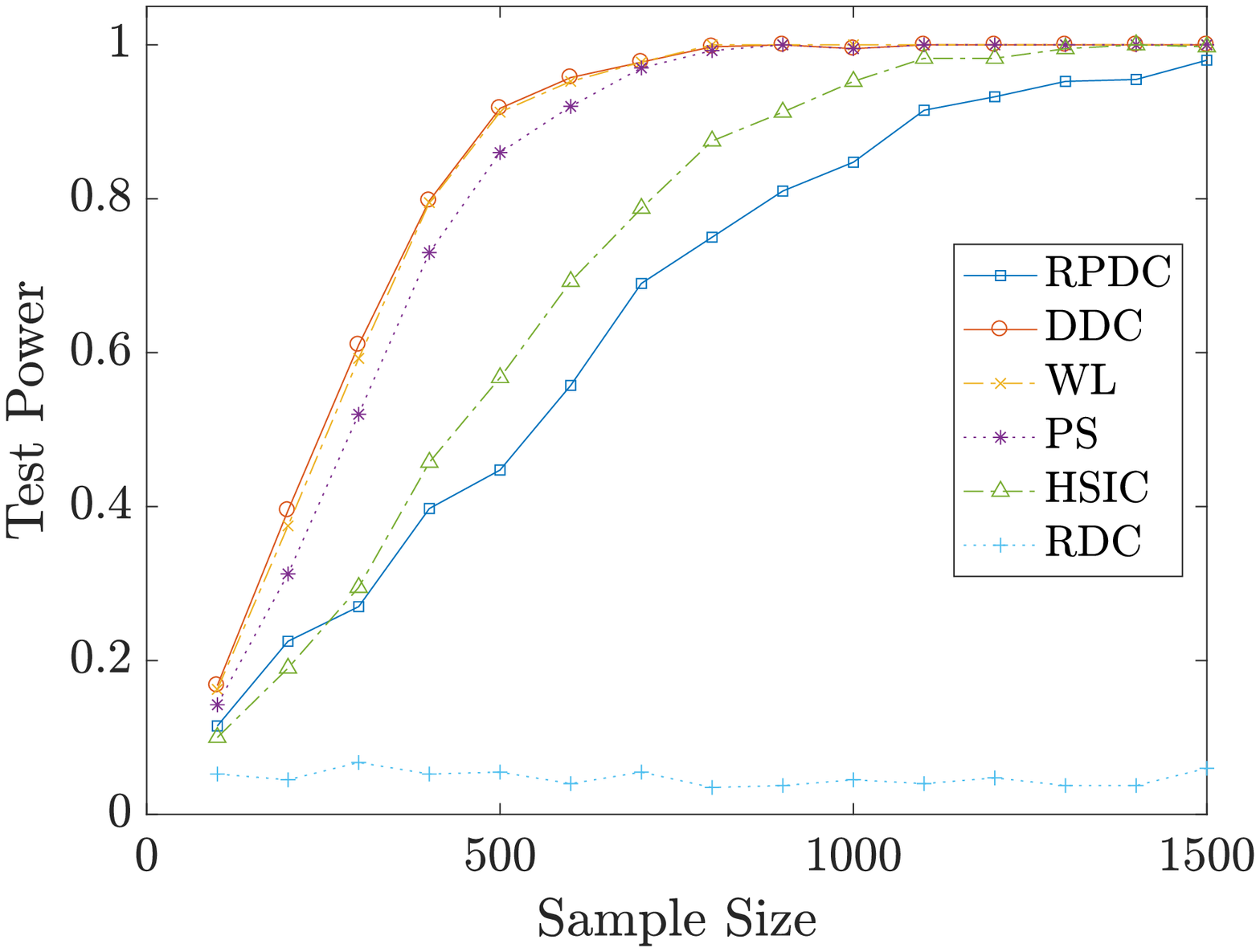}
        \caption{Dependence: $\rho=0.1$}
    \end{subfigure}
    \caption{Type-I Error/Test Power vs Sample Size $n$ in Example \ref{ex:5}. The result is based on $400$ repeated experiments.}
    \label{fig:5}
\end{figure}

Next, we compare those methods when $(X,Y)$ is no longer multivariate normal and the dependency between $X$ and $Y$ is non-linear.
We add a noise term to compare their performance in both the low and the high noise-to-signal ratio scenarios.
In this case, DDC and RPDC are much better than WL, PS and RDC.
The performance of HSIC is close to DDC and RPDC when the noise level is low but much worse than those two when the noise level is high.
\begin{example}
We set the dimension of data to be $p=q=10$.
We generate random vector $X \in \mathbb{R}^{10}$ from the standard multivariate normal distribution $\mathcal{N}(0,\mathbf{I}_{10})$.
Let the $i$-th entry of $Y$ be $Y_i = \log(X_i^2) + \epsilon_i, i = 1,\ldots,q$, where $\epsilon_i$'s are independent random errors, $\epsilon_i \sim \mathcal{N}(0,\sigma^2)$.
We set the value of $\sigma$ to be $1$ and $3$ to represent low and high noise ratios, respectively.
In the $\sigma=1$ case, the sample size $n$ is from $100$ to $1000$ with an increment $20$;
and in the $\sigma=3$ case, the sample size $n$ is from $100$ to $4000$ with an increment $100$.
\label{ex:6}
\end{example}
Figure \ref{fig:6} plots the test power of each test against sample size.
In both low and high noise cases, none of WL, PS and RDC has any test power.
In the low noise case, all of RPDC, DDC and HSIC have satisfactory test power ($>0.9$) when sample size is greater than $300$.
In the high noise case, RPDC and DDC could achieve more than $0.8$ in test power once sample size is greater than $500$ while the test power of HSIC reaches $0.8$ when the sample size is more than $2000$.
\begin{figure}[ht!]
    \centering
    \begin{subfigure}[b]{0.45\textwidth}
        \includegraphics[width=\textwidth]{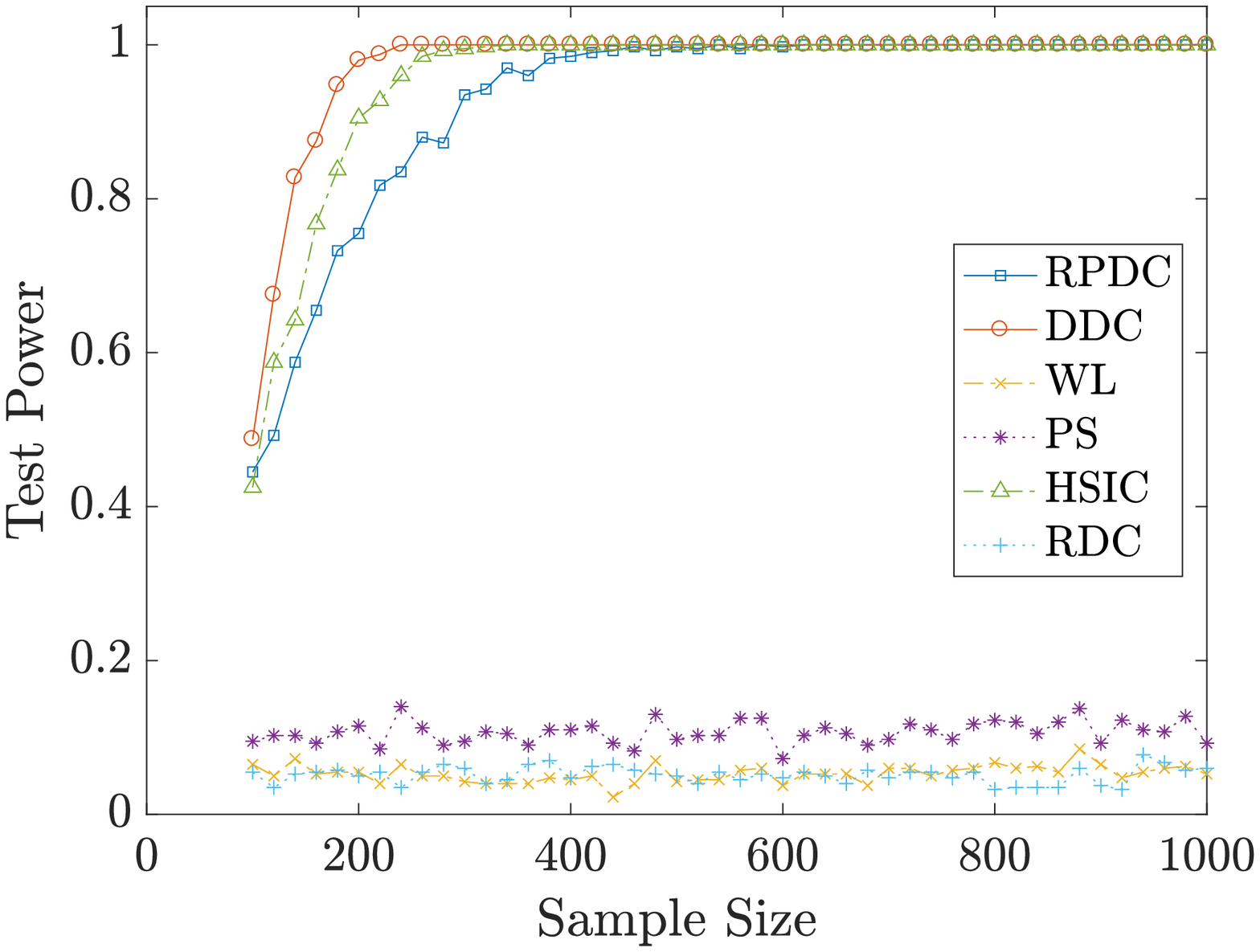}
        \caption{Low Noise: $\sigma=1$}
    \end{subfigure}
    ~
    \begin{subfigure}[b]{0.45\textwidth}
        \includegraphics[width=\textwidth]{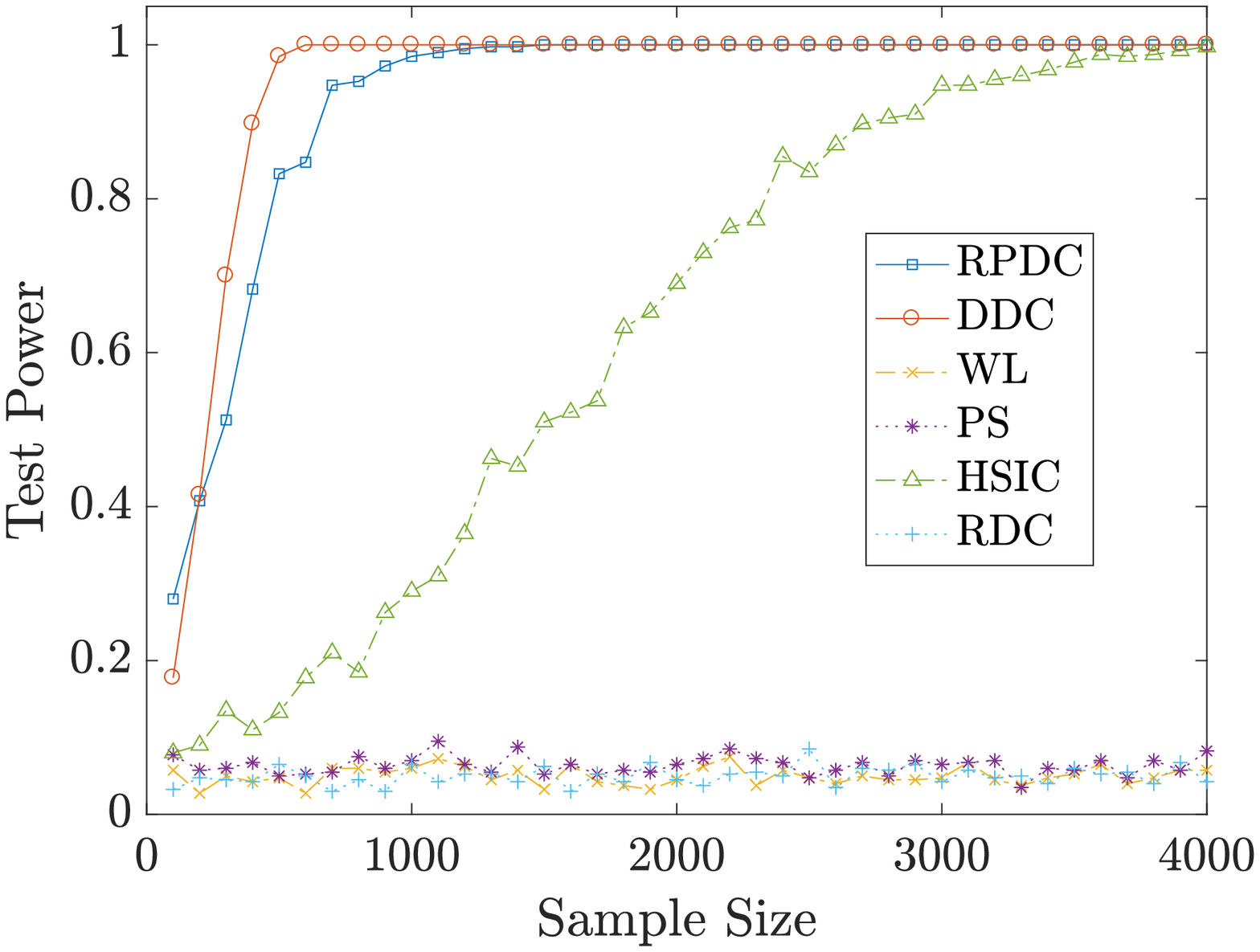}
        \caption{High Noise: $\sigma=3$}
    \end{subfigure}
    \caption{Test Power vs Sample Size $n$ in Example \ref{ex:6}. The significance level is $\alpha_s = 0.05$. The result is based on $N=400$ repeated experiments.}
    \label{fig:6}
\end{figure}

In the following example, we generate the data in the similar way with Example \ref{ex:6} but the difference is that the dependency is changing over time. Specifically, $X$ and $Y$ are independent at the beginning but they become dependent after some time point. Since all those tests are invariant with the order of the observations, this experiment simply means that only a proportion of observations are dependent while the rest are not.
\begin{example}
We set the dimension of data to be $p=q=10$. We generate random vector $X_t \in \mathbb{R}^{10}, t=1,\ldots,n$, from the standard multivariate normal distribution $\mathcal{N}(0,\mathbf{I}_{10})$.
Let the $i$-th entry of $Y_t$ be $Y_{t,i} = \log(Z_{t,i}^2) + \epsilon_{t,i}, t = 1,\ldots,T$ and $Y_{t,i} = \log(X_{t,i}^2) + \epsilon_{t,i}, t = T+1,\ldots,n$, where $Z_t \;i.i.d. \sim \mathcal{N}(0,\mathbf{I}_{10})$ and $\epsilon_{t,i}$'s are independent random errors, $\epsilon_{t,i} \sim \mathcal{N}(0,1)$.
We set the value of $T$ to be $0.5n$ and $0.8n$ to represent early and late dependency transition, respectively.
In the early change case, the sample size $n$ is from $500$ to $2000$ with an increment $100$;
and in the late change case, the sample size $n$ is from $500$ to $4000$ with an increment $100$.
\label{ex:7}
\end{example}
Figure \ref{fig:7} plots the test power of each test against sample size.
In both early and late change cases, none of WL, PS and RDC has any test power.
In the early change case, all of RPDC, DDC and HSIC have satisfactory test power ($>0.9$) when sample size is greater than $1500$.
In the late change case, DDC and HSIC could achieve more than $0.8$ in test power once sample size reaches $4000$ while the test power of RPDC is only $0.6$ when the sample size is $4000$.
As expected, the performance of DDC is better than RPDC in both cases and the performance of HSIC is between DDC and RPDC.
\begin{figure}[ht!]
    \centering
    \begin{subfigure}[b]{0.45\textwidth}
        \includegraphics[width=\textwidth]{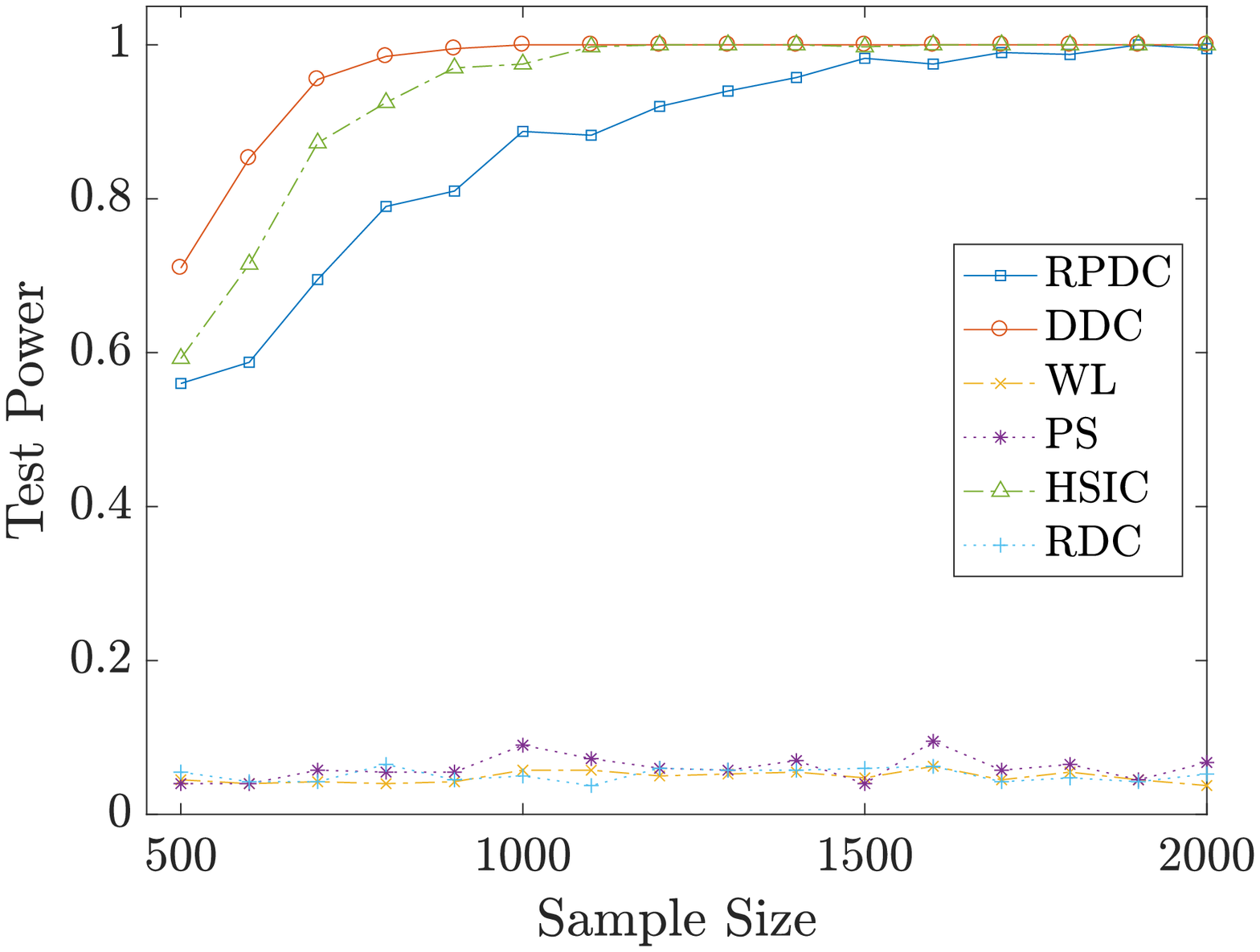}
        \caption{Early Change: $T=0.5n$}
    \end{subfigure}
    ~
    \begin{subfigure}[b]{0.45\textwidth}
        \includegraphics[width=\textwidth]{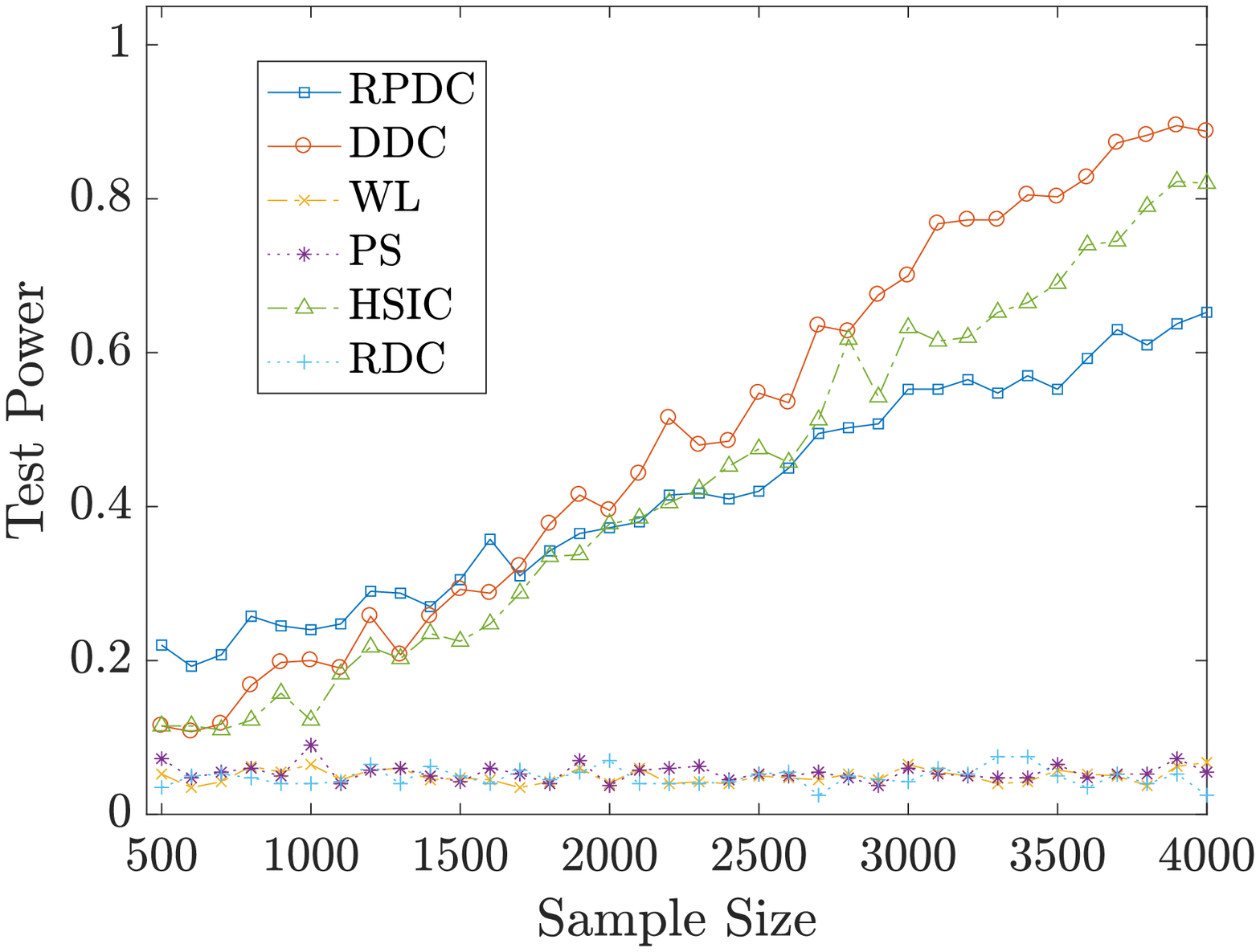}
        \caption{Late Change: $T=0.8n$}
    \end{subfigure}
    \caption{Test Power vs Sample Size $n$ in Example \ref{ex:7}. The significance level is $\alpha_s = 0.05$. The result is based on $N=400$ repeated experiments.}
    \label{fig:7}
\end{figure}

\begin{remark}
The experiments in this subsection show that though the RPDC under-performs the DDC when the sample size is relatively small, the RPDC could achieve the same test power with the DDC when the sample size is sufficiently large.
Considering the computational advantage of the RPDC (it has a lower order of computational complexity as indicated at the beginning of this paper), when the sample size is large enough, RPDC can be superior over the DDC.
\end{remark}


\section{Discussions}
\label{sec:discuss}

\subsection{A Discussion on the Computational Efficiency}
We compare the computational efficiency of proposed method (RPDC) and direct method (DDC) in Section \ref{sec:efficiency-compare}. We will discuss this issue here.

As $X \in \mathbb{R}^p$ and $Y \in \mathbb{R}^q$ are multivariate random variables, the effect of $p$ and $q$ on computing time could be significant when $p$ and $q$ are not negligible comparing to sample size $n$.
Now, we analyze the computational efficiency of DDC and RPDC by taking $p$ and $q$ into consideration. The computational complexity of DDC becomes $O(n^2(p+q))$ and that of RPDC becomes $O(nK (\log n + p + q))$. Let us denote the total number of operations in DDC by $O_1$ and that in RPDC by $O_2$. Then, by sacrificing the technical rigor, one may assume that there exist constants $L_1$ and $L_2$ such that
$$
O_1 \approx L_1 n^2(p+q) \text{, and } O_2 \approx L_2 nK (\log n + p + q).
$$
There is no doubt that $O_2$ will eventually much less than $O_1$ as the sample size $n$ grows.
Due to the complexity of the fast algorithm, we may expect $L_2 > L_1$, which means that the computational time of the RPDC can be even larger than the one for the DDC when the sample size is relatively small.
Then we need to study the problem: what is the break-even point in terms of sample size $n$ when the RPDC and the DDC has the same computational time?

Let $n_0 = n_0(p+q,K)$ denote the break-even point, which is a function of $p+q$ and number of Monte Carlo iterations $K$.
For simplicity, we fix $K=50$ since $50$ iterations could achieve satisfactory test power as we showed in Example \ref{ex:4}.
Consequently $n_0$ becomes a function solely depending on $p+q$.
Since it is hard to derive the close form of $n_0$, we derive it numerically instead.
For fixed $p+q$, we let the sample size vary and record the difference between the running time of two methods. We fit the difference of running time against sample size with smoothing spline.
The root of this spline is the numerical value of $n_0$ at $p+q$.

We plot the $n_0$ against $p+q$ in Figure \ref{fig:break-even}.
As the figure predicts, the break-even sample size decreases as the data dimension increases, which implies that our proposed method is more advantageous than the direct method when random variables are of high dimension.
However, as showed in Example \ref{ex:4}, the random projection based method does not perform well when high dimensional data have low dimensional dependency structure.
This indicates that one need to be cautious to use the proposed method when the dimension is high.
\begin{figure}[ht!]
    \centering
    \includegraphics[width=0.45\textwidth]{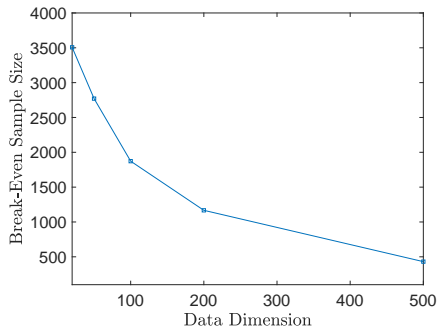}
    \caption{Break-Even Sample Size $n_0$ against Data Dimension $p+q$. This figure is based on $100$ repeated experiments.}
    \label{fig:break-even}
\end{figure}

\subsection{Connections with Existing Literature}
It turns out that distance-based methods are not the only choices in independence testing.
See \cite{lee2016variable} and the references therein to see alternatives.
On the other hand, in our numerical experiments, it is evident that the distance-correlated-based approaches compare favorably against many other popular contemporary alternatives.
Therefore it is meaningful to study the improvements of the distance-correlated-based approaches.

Our proposed method utilizes random projections, which bears similarity with the randomized feature mapping strategy \citep{rahimi2007random} that was developed in the machine learning community.
Such an approach has been proven to be effective in kernel-related methods \citep{achlioptas2001sampling,blum2006random,frieze2004fast,drineas2005nystrom}.
However, a closer examination will reveal the following difference: most of the aforementioned work are rooted on the Bochner's theorem \citep{rudin1990fourier} from harmonic analysis, which states that a continuous kernel in the Euclidean space is positive definite if
and only if the kernel function is the Fourier transform of a non-negative measure.
In this paper, we will deal with distance function which is not a positive definite kernel.
We managed to derive a counterpart to the randomized feature mapping, which was the influential idea that has been used in \cite{rahimi2007random}.

Random projections have been used in \cite{lopes2011more} to develop a powerful two-sample test in high dimensions.
They derived an asymptotic power function for their proposed test, and then provide sufficient conditions for their test to achieve greater power than other state-of-the-art tests.
They then used the receiver operating characteristic (ROC) curves (that are generated from their simulated data) to evaluate its performance against competing tests.
The derivation of the asymptotic relative efficiency (ARE) is of its own interests.
Despite the usage of random projection, the details of their methodology is very different from the one that is studied in the present paper.

Several distribution-free tests that are based on sample space partitions were suggested in \cite{heller2016consistent} for univariate random variables.
They proved that all suggested tests are consistent and showed the connection between their tests and the mutual information (MI).
Most importantly, they derived fast (polynomial-time) algorithms, which are essential for large sample size,
since the computational complexity of the naive algorithm is exponential in sample size.
Efficient implementations of all statistics and tests described in the aforementioned paper are available in the
R package HHG, which can be freely downloaded from the Comprehensive R Archive Network,
http://cran.r-project.org/.
Null tables can be downloaded from the first author's web site.

Distance-based independence/dependence measurements sometimes have been utilized in performing a greedy feature selection, often via dependence maximization \cite{huo2015fast}, \cite{zhu2012model} and \cite{li2012feature}, and it has been effective on some real-world datasets.
This paper simply mentions such a potential research line, without pursuing it.

Paper \cite{wang2015conditional} derives an efficient approach to compute for the conditional distance correlations.
We noted that there are strong resemblances between the distance covariances and its conditional counterpart.
The search for a potential extension of the work in this paper to conditional distance correlation can be a meaningful future topic of research.

\section{Conclusion}
\label{sec:conclude}
A significant contribution of this paper is we demonstrated that the multivariate variables in the independence tests need not imply the higher-order computational desideratum of the distance-based methods.

Distance-based methods are indispensable in statistics, particular in test of independence.
When the random variables are univariate, efficient numerical algorithms exist.
It is an open question when the random variables are multivariate.
This paper studies the random projection approach to tackle the above problem.
It first turn the multivariate calculation problem into univariate calculation one via random projections.
Then they study how the average of those statistics out of the projected (therefore univariate) samples can approximate the distance-based statistics that were intended to use.
Theoretical analysis was carried out, which shows that the loss of asymptotic efficiency (in the form of the asymptotic variance of the test statistics) is likely insignificant.
The new method can be numerically much more efficient, when the sample size is large; considering large sample sizes are well-expected under this information (or big-date) era.
Simulation studies validate the theoretical statements.
The theoretical analysis takes advantage of some newly available results, such as the equivalence of the distance-based methods with the reproducible kernel Hilbert spaces \cite{sejdinovic2013equivalence}.
The numerical methods utilizes a recently appeared fast algorithm in \cite{huo2015fast}.

\section*{Acknowledgement}
This material was based upon work partially supported by the National Science Foundation under Grant DMS-1127914 to the Statistical and Applied Mathematical Sciences Institute. Any opinions, findings, and conclusions or recommendations expressed in this material are those of the author(s) and do not necessarily reflect the views of the National Science Foundation.
This work has also been partially supported by NSF grant DMS-1613152.


\bibliographystyle{plain}

\newpage

\appendix
\pagenumbering{roman}
\numberwithin{equation}{section}
\setcounter{page}{1}

\section{Algorithms}
For readers' convenience, we present all the numerical algorithms here.
\begin{itemize}
\item
The Algorithm \ref{algo:approx_algo} summarizes how to compute the proposed distance covariance for multivariate inputs.

\item
The Algorithm \ref{algo:perm_based_test} describe an independence testing which applies permutation to generate a threshold.

\item
The Algorithm \ref{algo:dist_based_test} describes an independence test that is based on the approximate asymptotic distribution.

\end{itemize}
In the following algorithms, recall that $C_p$ and $C_q $ have been defined at the end of Section \ref{sec:intro}.

\begin{algorithm}
\KwData{Observations $X_1, \ldots, X_n \in \mathbb{R}^p$, $Y_1, \ldots, Y_n \in \mathbb{R}^q$; Number of Monte Carlo Iterations $K$}
\KwResult{Approximation of Sample Distance Covariance $\overline{\Omega}_n$}
\For{k = 1,\ldots, K}{
  Randomly generate $u_k$ from uniform($\mathcal{S}^{p-1}$); randomly generate $v_k$ from uniform($\mathcal{S}^{q-1}$)\;
Compute the projection of $X_i$'s on $u_k$: $u_k^t X = (u_k^t X_1, \ldots, u_k^t X_n)$\;
Compute the projection of $Y_i$'s on $v_k$: $v_k^t Y = (v_k^t Y_1, \ldots, v_k^t Y_n)$\;
Compute $\Omega_n^{(k)} = C_p C_q  \Omega_n(u_k^t X, v_k^t Y)$ with the Fast Algorithm in \cite{huo2015fast}\;
}
Return $\overline{\Omega}_n = \frac{1}{K} \sum_{k=1}^K \Omega_n^{(k)}$.
\caption{An Approximation of Sample Distance Covariance $\overline{\Omega}_n$}
\label{algo:approx_algo}
\end{algorithm}

\begin{algorithm}[H]
\KwData{Observations $X_1, \ldots, X_n \in \mathbb{R}^p$, $Y_1, \ldots, Y_n \in \mathbb{R}^q$; Number of Monte Carlo Iterations $K$; Significance Level $\alpha_s$; Number of Permutation: $L$}
\KwResult{Accept or Reject the Null Hypothesis $\mathcal{H}_0$: $X$ and $Y$ are independent}
	\For{l = 1,\ldots, L}{
		Generate a random permutation of $Y$: $Y^{\star,l} = (Y_1^\star, \ldots Y_n^\star)$\;
		Compute $V_l = \overline{\Omega}_n(X, Y^{\star,l})$, using the approach in Algorithm \ref{algo:approx_algo}\;
	}
	Reject $\mathcal{H}_0$ if $\frac{1+ \sum_{l=1}^L I(\overline{\Omega}_n>V_l)}{1+L} > \alpha_s$; otherwise, accept.
\caption{Independence Test Based on Permutations}
\label{algo:perm_based_test}
\end{algorithm}

\begin{algorithm}[H]
\KwData{Observations $X_1, \ldots, X_n \in \mathbb{R}^p$, $Y_1, \ldots, Y_n \in \mathbb{R}^q$; Number of Monte Carlo Iterations $K$; Significance Level $\alpha_s$}
\KwResult{Accept or Reject the Null Hypothesis $\mathcal{H}_0$: $X$ and $Y$ are independent}
\For{k = 1,\ldots, K}{
  		Randomly generate $u_k$ from uniform($\mathcal{S}^{p-1}$); randomly generate $v_k$ from uniform($\mathcal{S}^{q-1}$)\;
		Use the Fast Algorithm in \cite{huo2015fast} to compute: \\
		\hspace{0.5in} $\Omega_n^{(k)} = C_p C_q  \Omega_n(u_k^t X, v_k^t Y)$,\\
		\hspace{0.5in} $S_{n,1}^{(k)} = C_p^2 C_q^2 \Omega_n(u_k^t X, u_k^t X) \Omega_n(v_k^t Y, v_k^t Y)$,\\
		\hspace{0.5in} $S_{n,2}^{(k)} = \frac{ C_p a_{\cdot \cdot}^{u_k} }{n(n-1)},$\\
		\hspace{0.5in} $S_{n,3}^{(k)} = \frac{ C_q b_{\cdot \cdot}^{v_k} }{n(n-1)}$\;
		Randomly generate $u'_k$ from uniform($\mathcal{S}^{p-1}$); randomly generate $v'_k$ from uniform($\mathcal{S}^{q-1}$)\;
		Use the Fast Algorithm in \cite{huo2015fast} to compute: \\
		\hspace{0.5in} $\Omega_{n,X}^{(k)} = C_p^2 \Omega_n(u_k^t X, {u'}_k^t X)$,\\
		\hspace{0.5in} $\Omega_{n,Y}^{(k)} = C_q^2 \Omega_n(v_k^t Y, {v'}_k^t Y)$\;
	}
	$\overline{\Omega}_n = \frac{1}{K} \sum_{k=1}^K \Omega_n^{(k)}$;
	$\bar{S}_{n,1} = \frac{1}{K} \sum_{k=1}^K S_{n,1}^{(k)}$;
	$\bar{S}_{n,2} = \frac{1}{K} \sum_{k=1}^K S_{n,2}^{(k)}$;
	$\bar{S}_{n,3} = \frac{1}{K} \sum_{k=1}^K S_{n,2}^{(k)}$\;
	$\overline{\Omega}_{n,X} = \frac{1}{K} \sum_{k=1}^K \Omega_{n,X}^{(k)}$;
	$\overline{\Omega}_{n,Y} = \frac{1}{K} \sum_{k=1}^K \Omega_{n,Y}^{(k)}$\;
	$\alpha = \frac{1}{2}\frac{\bar{S}_{n,2}^2 \bar{S}_{n,3}^2}{\frac{K-1}{K} \overline{\Omega}_{n,X} \overline{\Omega}_{n,Y} + \frac{1}{K} \bar{S}_{n,1}}$;
	$\beta = \frac{1}{2}\frac{\bar{S}_{n,2} \bar{S}_{n,3}}{\frac{K-1}{K} \overline{\Omega}_{n,X} \overline{\Omega}_{n,Y} + \frac{1}{K} \bar{S}_{n,1}}$\;
	Reject $\mathcal{H}_0$ if $n \overline{\Omega}_n + \bar{S}_{n,2} \bar{S}_{n,3} > \mbox{Gamma}(\alpha,\beta; 1-\alpha_s)$; otherwise, accept it.
Here $\mbox{Gamma}(\alpha,\beta; 1-\alpha_s)$ is the $1-\alpha_s$ quantile of the distribution Gamma$(\alpha,\beta)$.
\caption{Independence Test Based on Asymptotic Distribution}
\label{algo:dist_based_test}
\end{algorithm}

\section{Proofs}

\subsection{Proof of Lemma \ref{lemma:1}}

\begin{proof}
The proof is straightforward as follows.
It is known that $X$ and $Y$ are independent if and only if $\phi_{X,Y}(t,s) = \phi_X(t) \phi_Y(s), \forall t \in \mathbb{R}^p, s \in \mathbb{R}^q$, which by definition of the characteristic functions is equivalent to
$$
\mathbb{E}[e^{iX^t t + iY^t s}] = \mathbb{E}[e^{iX^t t}] \mathbb{E}[e^{iY^t s}], \forall t \in \mathbb{R}^p, s \in \mathbb{R}^q.
$$
Changing of variables $t=ut'$ and $s=vs'$ in the above expression results in the following:
$$
\mathbb{E}[e^{iX^t ut' + iY^t vs'}] = \mathbb{E}[e^{iX^t ut'}] \mathbb{E}[e^{iY^t vs'}], \forall u \in \mathcal{S}^{p-1}, v \in \mathcal{S}^{q-1}, t',s' \in \mathbb{R},
$$
or equivalently, the following
$$
\mathbb{E}[e^{iu^tXt' + iv^tYs'}] = \mathbb{E}[e^{iu^tX t'}] \mathbb{E}[e^{iv^tY s'}], \forall u \in \mathcal{S}^{p-1}, v \in \mathcal{S}^{q-1}, t',s' \in \mathbb{R}.
$$
Note the above, again by the definitions of the characteristic functions, is equivalent to
$$
\phi_{u^tX,v^tY}(t',s') = \phi_{u^tX}(t') \phi_{v^tY}(s'), \forall u \in \mathcal{S}^{p-1}, v \in \mathcal{S}^{q-1}, t',s' \in \mathbb{R}.
$$
From the definition and the properties of the distance covariance $\mathcal{V}^2$ (Theorem \ref{th:dist-cov}), we know that the previous is equivalent to
$$
\mathcal{V}^2(u^tX,v^tY) = 0, \forall u \in \mathcal{S}^{p-1}, v \in \mathcal{S}^{q-1}.
$$
From all the above, we have proved Lemma \ref{lemma:1}.
\end{proof}

\subsection{Proof of Lemma \ref{lemma:2}}

We prove Lemma \ref{lemma:2}.
\begin{proof}
We will use the following change of variables:
$t = r_1 \cdot u, s = r_2 \cdot v,$
where $r_1,r_2 \in (-\infty, +\infty)$ and $u \in \mathcal{S}^{p-1}, v \in \mathcal{S}^{q-1}$. As the surface area of $\mathcal{S}^{p-1}$ is equal to $\frac{2\pi^{p/2}}{\Gamma(p/2)} = 2c_{p-1}$, we have
\begin{eqnarray*}
&& \mathcal{V}^2(X,Y) \\
&=& \int_{R^{p+q}} \frac{| \mathbb{E}[e^{iX^t t + iY^t s}] - \mathbb{E}[e^{iX^t t}] \mathbb{E}[e^{iY^t s}] |^2}{c_p c_q |t|^{p+1}|s|^{q+1}} dt ds \\
&=& c_{p-1}c_{q-1} \int_{\mathcal{S}^{p-1}_+} \int_{-\infty}^{+\infty} \int_{\mathcal{S}^{q-1}_+} \int_{-\infty}^{+\infty} \frac{|\mathbb{E}[e^{ir_1u^tX + ir_2 v^t Y}] - \mathbb{E}[e^{ir_1 u^tX}] \mathbb{E}[e^{ir_2 v^t Y}]|^2}{c_p c_q |r_1|^{p+1}|r_2|^{q+1}} \\
& & \hspace{3.5in} |r_1|^{p-1} |r_2|^{q-1} d\mu(u) dr_1 d\nu(v) dr_2 \\
&=& c_{p-1}c_{q-1} \int_{\mathcal{S}^{p-1}_+} \int_{\mathcal{S}^{q-1}_+} \int_{-\infty}^{+\infty}  \int_{-\infty}^{+\infty}  \frac{|\mathbb{E}[e^{ir_1u^tX + ir_2 v^t Y}] - \mathbb{E}[e^{ir_1 u^tX}] \mathbb{E}[e^{ir_2 v^t Y}]|^2}{c_p c_q |r_1|^2|r_2|^2} d\mu(u) d\nu(v) dr_1 dr_2 \\
&=& \frac{c_1^2c_{p-1}c_{q-1}}{c_p c_q} \int_{\mathcal{S}^{p-1}_+} \int_{\mathcal{S}^{q-1}_+} \mathcal{V}^2(u^tX,v^tY) d\mu(u) d\nu(v) \\
&=& C_p C_q \int_{\mathcal{S}^{p-1}} \int_{\mathcal{S}^{q-1}} \mathcal{V}^2(u^tX,v^tY) d\mu(u) d\nu(v).
\end{eqnarray*}
In the above, the first and fourth equations are due to the definition of $\mathcal{V}^2(\cdot,\cdot)$;
the second equation reflects the aforementioned change of variables;
the third equation is a reorganization;
the last equation is from the definition of constants $C_p$ and $C_q$.
From all the above, we establish the first part of Lemma \ref{lemma:2}.

For the sample distance covariance part, we just need to replace the population characteristic function $\phi_X(t) = \mathbb{E}[e^{iX^t t}]$ with the sample characteristic function $\hat{\phi}_X(t) = \frac{1}{n} \sum_{j=1}^n e^{iX_j^t t}$, the rest reasoning part is nearly identical.
We omit the details here.
\end{proof}

\subsection{Proof of Lemma \ref{lem:4.3}}

We will need the following lemma.
\begin{lemma}
\label{lemma:app1}
Suppose $v$ is a fixed unit vector in $\mathbb{R}^{p-1}$ and $u \in \mathcal{S}^{p-1}$.
Let $\mu$ be the uniform probability measure on $\mathcal{S}^{p-1}$.
We have
$$
C_p \int_{\mathcal{S}^{p-1}} |u^tv| d\mu(u) = 1,
$$
where constant $C_p$ has been mentioned at the end of Section \ref{sec:intro}.
\end{lemma}
\begin{proof}
Since both $u$ and $v$ are unit vector, we have
$$
|u^tv| = \left|\frac{\langle u, v \rangle}{\sqrt{|u||v|}}\right| = |\cos \theta|,
$$
where $\theta$ is the angle between vectors $u$ and $v$.
As we know, the angle between two random vectors on $\mathcal{S}^{p-1}$ follows distribution with density, (see \cite{cai2013distributions}) for $\theta \in [0,\pi]$,
\begin{equation}
\label{eq:app-h-theta}
h(\theta) = \frac{1}{\sqrt{\pi}} \frac{\Gamma(p/2)}{\Gamma((p-1)/2)} (\sin \theta)^{p-2}.
\end{equation}
Therefore, we have
\begin{eqnarray*}
\int_{\mathcal{S}^{p-1}} |u^tv| d\mu(u)
&=& \int_{0}^\pi h(\theta) |\cos \theta| d\theta \\
&=& 2\int_{0}^{\pi/2} h(\theta) \cos \theta d\theta \\
&\stackrel{\eqref{eq:app-h-theta}}{=}& 2\int_{0}^{\pi/2} \frac{1}{\sqrt{\pi}} \frac{\Gamma(p/2)}{\Gamma((p-1)/2)} (\sin \theta)^{p-2} \cos \theta d\theta \\
&=& 2\int_{0}^{1} \frac{1}{\sqrt{\pi}} \frac{\Gamma(p/2)}{\Gamma((p-1)/2)} x^{p-2} dx \\
&=& \frac{2}{\sqrt{\pi}} \frac{\Gamma(p/2)}{\Gamma((p-1)/2)} \int_{0}^{1} x^{p-2} dx \\
&=& \frac{\Gamma(p/2)}{\sqrt{\pi} \Gamma((p+1)/2)} = \frac{1}{C_p}.
\end{eqnarray*}
The second equation is due to the symmetry of the function on $[0,\pi]$; the third equation is a change of random variable; the sixth equation is from the fact that $\Gamma((p+1)/2) = \frac{p-1}{2} \Gamma((p-1)/2)$.
\end{proof}

We now prove Lemma \ref{lem:4.3}
\begin{proof}
We will need the following notations:
\begin{align}
&a_{ij}^u = |u^t(X_i - X_j)|, \hspace{0.2in} b_{ij}^v = |v^t(Y_i - Y_j)|, \nonumber \\
&a_{i\cdot}^u = \sum_{l=1}^n a_{il}^u, \hspace{0.2in} b_{i\cdot }^v = \sum_{l=1}^n b_{il}^v, \label{notation:2} \\
&a_{\cdot \cdot}^u = \sum_{k,l=1}^n a_{kl}^u, \hspace{0.1in} \text{and} \hspace{0.1in} b_{\cdot \cdot}^v = \sum_{k,l=1}^n b_{kl}^v. \nonumber
\end{align}
Recall the definition of $\Omega_n(\cdot,\cdot)$ in \eqref{eq:def_5}, we have
\begin{multline}
\label{eq:B3.1}
\Omega_n(u^tX,v^tY) = \frac{1}{n(n-3)} \sum_{i \neq j} a_{ij}^u b_{ij}^v \\- \frac{2}{n(n-2)(n-3)} \sum_{i=1}^n a_{i \cdot}^u b_{i \cdot}^v + \frac{a_{\cdot \cdot}^u b_{\cdot \cdot}^v}{n(n-1)(n-2)(n-3)}.
\end{multline}
By Lemma \ref{lemma:app1}, we have the following: $\forall 1\le i,j\le n$,
\begin{eqnarray}
C_p \int_{\mathcal{S}^{p-1}} |u^t(X_i - X_j)| d\mu(u) &=& |X_i - X_j| \quad \text{ and } \label{eq:B3.2} \\
C_q \int_{\mathcal{S}^{q-1}} |v^t(Y_i - Y_j)| d\nu(v) &=& |Y_i - Y_j|. \label{eq:B3.3}
\end{eqnarray}
By integrating $\Omega_n(u^tX,v^tY)$ on $u$ and $v$, we have
\begin{align*}
&C_p C_q \int_{\mathcal{S}^{p-1} \times \mathcal{S}^{q-1}} \Omega_n(u^tX,v^tY) d\mu(u) d\nu(v) \\
\stackrel{\mbox{\eqref{eq:B3.1}}}{=}& \frac{1}{n(n-3)} \sum_{i \neq j} C_p \int_{\mathcal{S}^{p-1}} a_{ij}^u d\mu(u) C_q\int_{\mathcal{S}^{q-1}} b_{ij}^v d\nu(v) \\
& - \frac{2}{n(n-2)(n-3)} \sum_{i=1}^n C_p \int_{\mathcal{S}^{p-1}} a_{i \cdot}^u d\mu(u) C_q \int_{\mathcal{S}^{q-1}} b_{i \cdot}^v d\nu(v) \\
& + \frac{C_p \int_{\mathcal{S}^{p-1}} a_{\cdot \cdot}^u d\mu(u) C_q\int_{\mathcal{S}^{q-1}} b_{\cdot \cdot}^v d\nu(v) }{n(n-1)(n-2)(n-3)} \\
\stackrel{\mbox{\eqref{eq:B3.2}\eqref{eq:B3.3}}}{=}& \frac{1}{n(n-3)} \sum_{i \neq j} a_{ij} b_{ij} - \frac{2}{n(n-2)(n-3)} \sum_{i=1}^n a_{i \cdot} b_{i \cdot} \\
& + \frac{a_{\cdot \cdot} b_{\cdot \cdot}}{n(n-1)(n-2)(n-3)} = \Omega_n(X,Y).
\end{align*}
From all the above, the equation in the lemma is established.
\end{proof}

\subsection{Proof of Lemma \ref{lem:concentration}}

\begin{proof}
We can regard $\Omega_n(u^tX,v^tY)$ as a real-valued function on $\mathbb{R}^p \times \mathbb{R}^q$. It is easy to find that $\Omega_n(u^tX,v^tY)$ is a continuous differentiable function by its definition. Since $\mathcal{B}^p \times \mathcal{B}^q$ is a convex compact set, $\Omega_n(u^tX,v^tY)$ must be bounded on this set. Let $L_{X,Y} = \sup_{u \in \mathcal{B}^p, v \in \mathcal{B}^q} \Omega_n(u^tX,v^tY)$ denote this upper bound, which is constant depending on the distribution of $X$ and $Y$ only. Since $a_{ij}^u = |u^t(X_i - X_j)| \le |u| |X_i - X_j| = |X_i - X_j| = a_{ij}$, then we have
\begin{align*}
L_{X,Y} &\le \frac{1}{n(n-3)} \sum_{i \neq j} a_{ij} b_{ij} + \frac{a_{\cdot \cdot} b_{\cdot \cdot}}{n(n-1)(n-2)(n-3)} \\
&\le \mathbb{E}[|X-X'||Y-Y'|] + \mathbb{E}[|X-X'|] \mathbb{E}[|Y-Y'|] + o_P(1) \\
&\le 2 \sqrt{\mathbb{E}[|X-X'|^2] \mathbb{E}[|Y-Y'|^2]} + o_P(1) \\
&\le 2\sqrt{2\mbox{Tr}[\Sigma_X] 2\mbox{Tr}[\Sigma_Y]} + o_P(1) \\
&\le 5 \sqrt{\mbox{Tr}[\Sigma_X] \mbox{Tr}[\Sigma_Y]} \mbox{ for sufficiently large } n.
\end{align*}
We can get the first inequality from the definition in (2.5) by removing the negative term.
It is worth noting that $\frac{1}{n(n-3)} \sum_{i \neq j} a_{ij} b_{ij}$ and $\frac{a_{\cdot \cdot} b_{\cdot \cdot}}{n(n-1)(n-2)(n-3)}$ are the U-statistics for $\mathbb{E}[|X-X'||Y-Y'|]$ and $\mathbb{E}[|X-X'|] \mathbb{E}[|Y-Y'|]$, respectively.
So, the second inequality is due to almost sure convergence of U-statistics, see \cite[Chapter 5.4 Theorem A]{serfling1980approximation}, where $o_P(1)$ represents a small error that converges to 0 as $n \rightarrow \infty$.
The third inequality is an immediate result from H\"{o}lder's inequality.
The fourth inequality holds as
\begin{align*}
\mathbb{E}[|X-X'|^2] &= \sum_{i=1}^p \mathbb{E}[(X_{(i)} - X_{(i)}')^2] = \sum_{i=1}^p (\mathbb{E}[X_{(i)}^2] + \mathbb{E}[X_{(i)}^2] - 2\mathbb{E}[X_{(i)} X_{(i)}']) \\
&= 2 \sum_{i=1}^p (\mathbb{E}[X_{(i)}^2] - \mathbb{E}^2[X_{(i)}]) = 2 \sum_{i=1}^p \mbox{Var}(X_{(i)}) = 2 \mbox{Tr}[\Sigma_X],
\end{align*}
where $X_{(i)}$ and $X_{(i)}'$ are the $i$-th component of $X$ and $X'$, respectively.

Since $(u_1,v_1), \ldots, (u_K,v_K)$ are draw i.i.d. from uniform distribution on $\mathcal{S}^{p-1} \times \mathcal{S}^{q-1}$. $h_1, \ldots, \Omega_K$ are i.i.d. random variables with $\mathbb{E}[\Omega^{(k)}] = \Omega_n, \forall k$. And, we know that $\Omega^{(k)} \le C_p C_q L_{X,Y}$. By Chernoff-Hoeffding's inequality \citep{hoeffding1963probability}, we have
\begin{eqnarray*}
\mathbf{P}\left(\left| \overline{\Omega}_n - \Omega_n \right| > \epsilon \right)
&=& \mathbf{P}\left(\left| \sum_{k=1}^K \Omega^{(k)} - K\Omega_n \right| > K\epsilon \right) \\
&\le& 2\exp\left\{ \frac{-2K^2\epsilon^2}{KC_p^2 C_q^2 L_{X,Y}^2} \right\}  \\
&\le& 2\exp\left\{ -\frac{2K\epsilon^2}{25C_p^2 C_q^2 \mbox{Tr}[\Sigma_X] \mbox{Tr}[\Sigma_Y]} \right\}.
\end{eqnarray*}
\end{proof}

\subsection{Proof of Lemma \ref{lem:bound-vars}}

\begin{proof}
Recall that $\Omega_n$ is an unbiased estimator of $\mathcal{V}^2(X,Y)$ and $\Omega_4 = h_4$, we have  $\mathbb{E}[h_4] = \mathcal{V}^2(X,Y) \ge 0$, consequently, we have the following:
\begin{align*}
& \mbox{Var}(h_4) \le \mathbb{E}[h_4^2] \\
=& \mathbb{E} \left[ \frac{1}{4} \sum_{1\le i,j \le 4, i \neq j} |X_i-X_j| |Y_i - Y_j| \right. \\
& \hspace{0.05\textwidth} \left. - \frac{1}{4} \sum_{i=1}^4 \left( \sum_{1 \le j \le 4, j \neq i} |X_i-X_j| \sum_{1 \le j \le 4, j \neq i} |Y_i-Y_j| \right) \right.\\
& \hspace{0.05\textwidth} \left. + \frac{1}{24} \sum_{1\le i,j \le 4, i \neq j} |X_i-X_j| \sum_{1\le i,j \le 4, i \neq j} |Y_i-Y_j| \right]^2 \\
\le& C_1 \mathbb{E}[ |X_1 - X_2|^2 |Y_1-Y_2|^2 ] + C_2 \mathbb{E}[ |X_1 - X_2|^2 |Y_1-Y_2| |Y_1-Y_3| ] \\
& \hspace{0.05\textwidth} + C_3 \mathbb{E}[ |X_1 - X_2|^2 |Y_1-Y_2| |Y_3-Y_4| ] \\
& \hspace{0.05\textwidth} + C_4 \mathbb{E}[ |X_1-X_2| |X_1-X_3| |Y_1-Y_2|^2 ] \\
& \hspace{0.05\textwidth} + C_5 \mathbb{E}[ |X_1-X_2| |X_1-X_3| |Y_1-Y_2| |Y_1-Y_3| ] \\
& \hspace{0.05\textwidth} + C_6 \mathbb{E}[ |X_1-X_2| |X_1-X_3| |Y_1-Y_2| |Y_3-Y_4| ] \\
& \hspace{0.05\textwidth} + C_7 \mathbb{E}[ |X_1-X_2| |X_3-X_4| |Y_1-Y_2|^2 ] \\
& \hspace{0.05\textwidth} + C_8 \mathbb{E}[ |X_1-X_2| |X_3-X_4| |Y_1-Y_2| |Y_1-Y_3| ] \\
& \hspace{0.05\textwidth} + C_9 \mathbb{E}[ |X_1-X_2| |X_3-X_4| |Y_1-Y_2| |Y_3-Y_4| ] \\
\le& C_1' \mathbb{E}[ |X_1 - X_2|^2 |Y_1-Y_2|^2 ] + C_2' \mathbb{E}[ |X_1 - X_2|^2 |Y_1-Y_3|^2 ] \\
& \hspace{0.05\textwidth} + C_3' \mathbb{E}[ |X_1 - X_2|^2 |Y_3-Y_4|^2 ] \\
\le& C_4' \mathbb{E}[|X|^2 |Y|^2] \le \infty,
\end{align*}
where $C_1,\ldots,C_9, C_1', \ldots, C_4' \ge 0$ are some constants.
The second inequality is due to computing the squared term and set all coefficients to their absolution value,
the third inequality is by Cauchy's inequality $ab \le \frac{1}{2}a^2 + b^2$,
and the fourth inequality is because of $|X_1 - X_2|^2 \le 2|X_1|^2 + 2|X_2|^2$.

By the law of total variance, both $h_1$ and $h_2$ must have variances no more than the variance of $h_4$. We can have $\mbox{Var}(h_1) \le \mbox{Var}(h_4) < \infty \text{ and } \mbox{Var}(h_2) \le \mbox{Var}(h_4) < \infty.$
\end{proof}

\subsection{Proof of Lemma \ref{lem:generic-h1-h2}}

\begin{proof}
Under the general case, we derive the formulas of $h_1((X_1,Y_1))$ and $h_2((X_1,Y_1),(X_2,Y_2))$.
Recall that
\begin{align*}
h_1((X_1,Y_1)) &= \mathbb{E}_{2,3,4}[ h_4((X_1,Y_1),(X_2,Y_2),(X_3,Y_3),(X_4,Y_4)) ], \\
h_2((X_1,Y_1),(X_2,Y_2)) &= \mathbb{E}_{3,4}[ h_4((X_1,Y_1),(X_2,Y_2),(X_3,Y_3),(X_4,Y_4)) ],
\end{align*}
where
\begin{align*}
&h_4((X_1,Y_1),(X_2,Y_2),(X_3,Y_3),(X_4,Y_4)) \\
=& \frac{1}{4} \sum_{1\le i,j \le 4, i \neq j} |X_i-X_j| |Y_i - Y_j| - \frac{1}{4} \sum_{i=1}^4 \left( \sum_{j=1, j \neq i}^4 |X_i-X_j| \sum_{j=1, j \neq i}^4 |Y_i-Y_j| \right)\\
& \hspace{1in} + \frac{1}{24} \sum_{1\le i,j \le 4, i \neq j} |X_i-X_j| \sum_{1\le i,j \le 4, i \neq j} |Y_i-Y_j|. \\
\end{align*}
To facilitate the calculation, we introduce the notations $a_{ij} = |X_i - X_j|$ and $b_{ij} = |Y_i - Y_j|$, and then utilize them to expand quantity $h_4((X_1,Y_1),(X_2,Y_2),(X_3,Y_3),(X_4,Y_4))$ as follows:
\begin{align*}
&h_4((X_1,Y_1),(X_2,Y_2),(X_3,Y_3),(X_4,Y_4)) \\
=& \hspace{.2in} \frac{1}{6}a_{12}b_{12} - \frac{1}{12} a_{12}b_{13} - \frac{1}{12} a_{12}b_{14} - \frac{1}{12} a_{12}b_{23} - \frac{1}{12} a_{12}b_{24} + \frac{1}{6} a_{12}b_{34} \\
& \hspace{.0in} - \frac{1}{12} a_{13}b_{12} + \frac{1}{6}a_{13}b_{13} - \frac{1}{12} a_{13}b_{14} - \frac{1}{12} a_{13}b_{23} + \frac{1}{6} a_{13}b_{24} - \frac{1}{12} a_{13}b_{34} \\
& \hspace{.0in} - \frac{1}{12} a_{14}b_{12} - \frac{1}{12} a_{14}b_{13} + \frac{1}{6}a_{14}b_{14} + \frac{1}{6} a_{14}b_{23} - \frac{1}{12} a_{14}b_{24} - \frac{1}{12} a_{14}b_{34} \\
& \hspace{.0in} - \frac{1}{12} a_{23}b_{12} - \frac{1}{12} a_{23}b_{13} + \frac{1}{6} a_{23}b_{14} + \frac{1}{6}a_{23}b_{23} - \frac{1}{12} a_{23}b_{24} - \frac{1}{12} a_{23}b_{34} \\
& \hspace{.0in} - \frac{1}{12} a_{24}b_{12} + \frac{1}{6} a_{24}b_{13} - \frac{1}{12} a_{24}b_{14} - \frac{1}{12} a_{24}b_{23} + \frac{1}{6}a_{24}b_{24} - \frac{1}{12}a_{24}b_{34} \\
& \hspace{.0in} + \frac{1}{6} a_{34}b_{12} - \frac{1}{12} a_{34}b_{13} - \frac{1}{12} a_{34}b_{14} - \frac{1}{12} a_{34}b_{23} - \frac{1}{12} a_{34}b_{24} + \frac{1}{6}a_{34}b_{34}.
\end{align*}
One may verify the correctness of the above by brute force.
The following is a matrix that consists of the terms of $h_4((X_1,Y_1),(X_2,Y_2),(X_3,Y_3),(X_4,Y_4))$.
In the same matrix, we highlighted the terms, which will become equal after taking the expectation with respect to random variables $(X_2, Y_2)$, $(X_3, Y_3)$ and $(X_4, Y_4)$.
\begin{center}
\begin{tikzpicture}[
	strip/.style = {
	draw=#1,
	line width=1.2em, opacity=0.2,
	line cap=round ,
		},
	]
\matrix (mtrx)  [matrix of math nodes, left delimiter={(},right delimiter={)},
                 column sep=.1em,
                 nodes={text height=3ex, text width=10ex}]
{
+\frac{1}{6}a_{12}b_{12} &- \frac{1}{12} a_{12}b_{13} &- \frac{1}{12} a_{12}b_{14} &- \frac{1}{12} a_{12}b_{23} &- \frac{1}{12} a_{12}b_{24} &+ \frac{1}{6} a_{12}b_{34} \\
- \frac{1}{12} a_{13}b_{12} &+ \frac{1}{6}a_{13}b_{13} &- \frac{1}{12} a_{13}b_{14} &- \frac{1}{12} a_{13}b_{23} &+ \frac{1}{6} a_{13}b_{24} &- \frac{1}{12} a_{13}b_{34} \\
 - \frac{1}{12} a_{14}b_{12} &- \frac{1}{12} a_{14}b_{13} &+ \frac{1}{6}a_{14}b_{14} &+ \frac{1}{6} a_{14}b_{23} &- \frac{1}{12} a_{14}b_{24} &- \frac{1}{12} a_{14}b_{34} \\
- \frac{1}{12} a_{23}b_{12} &- \frac{1}{12} a_{23}b_{13} &+ \frac{1}{6} a_{23}b_{14} &+ \frac{1}{6}a_{23}b_{23} &- \frac{1}{12} a_{23}b_{24} &- \frac{1}{12} a_{23}b_{34} \\
- \frac{1}{12} a_{24}b_{12} &+ \frac{1}{6} a_{24}b_{13} &- \frac{1}{12} a_{24}b_{14} &- \frac{1}{12} a_{24}b_{23} &+ \frac{1}{6}a_{24}b_{24} &- \frac{1}{12}a_{24}b_{34} \\
+ \frac{1}{6} a_{34}b_{12} &- \frac{1}{12} a_{34}b_{13} &- \frac{1}{12} a_{34}b_{14} &- \frac{1}{12} a_{34}b_{23} &- \frac{1}{12} a_{34}b_{24} &+ \frac{1}{6}a_{34}b_{34} \\
};

\path[draw,strip=red,shorten <=-3mm] (mtrx-1-1.center) -- (mtrx-3-3.center);
\path[draw,strip=red,shorten <=-3mm] (mtrx-4-4.center) -- (mtrx-6-6.center);
\path[draw,strip=red,shorten <=-3mm] (mtrx-6-1.center) -- (mtrx-4-3.center);
\path[draw,strip=red,shorten <=-3mm] (mtrx-3-4.center) -- (mtrx-1-6.center);
\path[draw,strip=blue,shorten <=-3mm] (mtrx-1-2.center) -- (mtrx-1-3.center) -- (mtrx-2-3.center) ;
\path[draw,strip=blue,shorten <=-3mm] (mtrx-2-1.center) -- (mtrx-3-1.center) -- (mtrx-3-2.center);
\path[draw,strip=blue,shorten <=-3mm] (mtrx-4-2.center) -- (mtrx-4-1.center) --  (mtrx-5-1.center);
\path[draw,strip=blue,shorten <=-3mm] (mtrx-5-3.center) -- (mtrx-6-3.center) -- (mtrx-6-2.center);
\path[draw,strip=blue,shorten <=-3mm] (mtrx-1-5.center) -- (mtrx-1-4.center) -- (mtrx-2-4.center);
\path[draw,strip=blue,shorten <=-3mm] (mtrx-2-6.center) -- (mtrx-3-6.center) -- (mtrx-3-5.center);
\path[draw,strip=blue,shorten <=-3mm] (mtrx-4-5.center) -- (mtrx-4-6.center) -- (mtrx-5-6.center);
\path[draw,strip=blue,shorten <=-3mm] (mtrx-5-4.center) -- (mtrx-6-4.center) -- (mtrx-6-5.center);
\end{tikzpicture}
\end{center}

Thus, $h_1((X_1,Y_1))$ could be expressed as follows.
\begin{align}
& h_1((X_1,Y_1)) = \mathbb{E}_{2,3,4}[ h_4((X_1,Y_1),(X_2,Y_2),(X_3,Y_3),(X_4,Y_4)) ] \nonumber \\
=& \frac{1}{2} \mathbb{E}[|X_1 - X'| |Y_1 - Y'|] - \frac{1}{2} \mathbb{E}[|X_1 - X'| |Y_1 - Y''|] \nonumber \\
& + \frac{1}{2} \mathbb{E}[|X_1 - X'| |Y - Y''|] -\frac{1}{2} \mathbb{E}[|X_1 - X'| |Y' - Y''|] \label{def:h1_1}\\
& + \frac{1}{2} \mathbb{E}[|X - X''| |Y_1 - Y'|] - \frac{1}{2} \mathbb{E}[|X' - X''| |Y_1 - Y'|] \nonumber \\
& + \frac{1}{2} \mathbb{E}[|X - X'| |Y - Y'|]    - \frac{1}{2} \mathbb{E}[|X - X'| |Y - Y''|]. \nonumber
\end{align}
We may notice that the four above lines are equal to the expectations of sums of terms in the upper left, upper right, bottom left, and bottom right quadrants of the aforementioned matrix, respectively.

Similarly, we can highlight the entries, which will be the same after taking expectation with respect to $(X_3, Y_3)$ and $(X_4, Y_4)$.
We do it in the following:
\begin{center}
\begin{tikzpicture}[
	strip/.style = {
	draw=#1,
	line width=1.2em, opacity=0.2,
	line cap=round ,
		},
	]
\matrix (mtrx)  [matrix of math nodes, left delimiter={(},right delimiter={)},
                 column sep=.1em,
                 nodes={text height=3ex, text width=10ex}]
{
+\frac{1}{6}a_{12}b_{12} &- \frac{1}{12} a_{12}b_{13} &- \frac{1}{12} a_{12}b_{14} &- \frac{1}{12} a_{12}b_{23} &- \frac{1}{12} a_{12}b_{24} &+ \frac{1}{6} a_{12}b_{34} \\
- \frac{1}{12} a_{13}b_{12} &+ \frac{1}{6}a_{13}b_{13} &- \frac{1}{12} a_{13}b_{14} &- \frac{1}{12} a_{13}b_{23} &+ \frac{1}{6} a_{13}b_{24} &- \frac{1}{12} a_{13}b_{34} \\
 - \frac{1}{12} a_{14}b_{12} &- \frac{1}{12} a_{14}b_{13} &+ \frac{1}{6}a_{14}b_{14} &+ \frac{1}{6} a_{14}b_{23} &- \frac{1}{12} a_{14}b_{24} &- \frac{1}{12} a_{14}b_{34} \\
- \frac{1}{12} a_{23}b_{12} &- \frac{1}{12} a_{23}b_{13} &+ \frac{1}{6} a_{23}b_{14} &+ \frac{1}{6}a_{23}b_{23} &- \frac{1}{12} a_{23}b_{24} &- \frac{1}{12} a_{23}b_{34} \\
- \frac{1}{12} a_{24}b_{12} &+ \frac{1}{6} a_{24}b_{13} &- \frac{1}{12} a_{24}b_{14} &- \frac{1}{12} a_{24}b_{23} &+ \frac{1}{6}a_{24}b_{24} &- \frac{1}{12}a_{24}b_{34} \\
+ \frac{1}{6} a_{34}b_{12} &- \frac{1}{12} a_{34}b_{13} &- \frac{1}{12} a_{34}b_{14} &- \frac{1}{12} a_{34}b_{23} &- \frac{1}{12} a_{34}b_{24} &+ \frac{1}{6}a_{34}b_{34} \\
};

\path[draw,strip=red,shorten <=-3mm] (mtrx-2-2.center) -- (mtrx-3-3.center);
\path[draw,strip=red,shorten <=-3mm] (mtrx-4-4.center) -- (mtrx-5-5.center);
\path[draw,strip=red,shorten <=-3mm] (mtrx-5-2.center) -- (mtrx-4-3.center);
\path[draw,strip=red,shorten <=-3mm] (mtrx-3-4.center) -- (mtrx-2-5.center);

\path[draw,strip=blue,shorten <=-3mm] (mtrx-1-2.center) -- (mtrx-1-3.center);\path[draw,strip=blue,shorten <=-3mm] (mtrx-2-1.center) -- (mtrx-3-1.center);
\path[draw,strip=blue,shorten <=-3mm] (mtrx-4-1.center) --  (mtrx-5-1.center);
\path[draw,strip=blue,shorten <=-3mm] (mtrx-6-2.center) -- (mtrx-6-3.center);
\path[draw,strip=blue,shorten <=-3mm] (mtrx-1-5.center) -- (mtrx-1-4.center);
\path[draw,strip=blue,shorten <=-3mm] (mtrx-2-6.center) -- (mtrx-3-6.center);
\path[draw,strip=blue,shorten <=-3mm] (mtrx-4-6.center) -- (mtrx-5-6.center);
\path[draw,strip=blue,shorten <=-3mm] (mtrx-6-4.center) -- (mtrx-6-5.center);

\path[draw,strip=blue,shorten <=-3mm] (mtrx-2-3.center) -- (mtrx-3-2.center);
\path[draw,strip=blue,shorten <=-3mm] (mtrx-4-2.center) -- (mtrx-5-3.center);
\path[draw,strip=blue,shorten <=-3mm] (mtrx-2-4.center) -- (mtrx-3-5.center);
\path[draw,strip=blue,shorten <=-3mm] (mtrx-5-4.center) -- (mtrx-4-5.center);
\end{tikzpicture}
\end{center}

Therefore, the expression of $h_2((X_1,Y_1),(X_2,Y_2))$ can be written as follows.
\begin{align}
\label{def:h2_1}
& h_2((X_1,Y_1),(X_2,Y_2)) =\mathbb{E}_{3,4}[ h_4((X_1,Y_1),(X_2,Y_2),(X_3,Y_3),(X_4,Y_4)) ]  \\
=& \frac{1}{6} |X_1 - X_2| |Y_1 - Y_2| + \frac{1}{3} \mathbb{E}[|X_1 - X'| |Y_1 - Y'|] + \frac{1}{3} \mathbb{E}[|X_2 - X'| |Y_2 - Y'|] \nonumber \\
& + \frac{1}{6} \mathbb{E}[|X - X'| |Y - Y'|] + \frac{1}{6} |X_1 - X_2| \mathbb{E}[|Y - Y'|] + \frac{1}{3} \mathbb{E}[|X_1 - X| |Y_2 - Y'|] \nonumber \\
& + \frac{1}{3} \mathbb{E}[|X_2 - X| |Y_1 - Y'|] + \frac{1}{6} |Y_1 - Y_2| \mathbb{E}[|X - X'|] - \frac{1}{6} |X_1 - X_2| \mathbb{E}[|Y_1 - Y'|] \nonumber \\
& - \frac{1}{6} |X_1 - X_2| \mathbb{E}[|Y_2 - Y'|]  - \frac{1}{6} |Y_1- Y_2| \mathbb{E}[|X_1 - X|] - \frac{1}{6} \mathbb{E}[|X_1 - X| |Y_1- Y'|] \nonumber \\
& - \frac{1}{6} \mathbb{E}[|X_1 - X| |Y_2 - Y|] - \frac{1}{6} \mathbb{E}[|X_1 - X| |Y - Y'|] - \frac{1}{6} |Y_1- Y_2| \mathbb{E}[|X_2 - X|] \nonumber \\
& - \frac{1}{6} \mathbb{E}[|X_2 - X| |Y_1- Y'|] - \frac{1}{6} \mathbb{E}[|X_2 - X| |Y_2 - Y'|] - \frac{1}{6} \mathbb{E}[|X_2 - X| |Y - Y'|] \nonumber \\
& - \frac{1}{6} \mathbb{E}[|X - X'| |Y_1 - Y|] - \frac{1}{6} \mathbb{E}[|X - X'| |Y_2 - Y|]. \nonumber
\end{align}
\end{proof}

\subsection{Proof of Lemma \ref{lem:h1-h2-independent}}

\begin{proof}
In the rest of this section, let us assume that $X$'s are independent of $Y$'s.
The following notations will be utilized to simplify our calculations.
\begin{align*}
&a_{12} = |X_1 - X_2|,                   &&b_{12} = |Y_1 - Y_2|, \\
&a_{1} = \mathbb{E}[ |X_1 - X| ], 	 &&b_{1} = \mathbb{E}[ |Y_1 - Y| ], \\
&a_{2} = \mathbb{E}[ |X_2 - X| ], 	 &&b_{2} = \mathbb{E}[ |Y_2 - Y| ], \\
&a = \mathbb{E}[ |X - X'| ], \text{ and }&&b = \mathbb{E}[ |Y - Y'| ],
\end{align*}
where the expectation operator $\mathbb{E}$ is taken with respect to $X$, $X'$, $Y$, $Y'$, or any combination of them, whenever it is appropriate.
Then, when $X$'s are independent of $Y$'s, one can easily verify the following:
\begin{equation*}
h_1((X_1,Y_1)) = \frac{1}{2} a_1 b_1 + \frac{1}{2} ab + \frac{1}{2} a_1b + \frac{1}{2} a b_1 - \frac{1}{2} a_1 b_1 -\frac{1}{2} a_1 b - \frac{1}{2} a b_1 - \frac{1}{2} ab = 0, %
\end{equation*}
as well as the following:
\begin{align*}
& h_2((X_1,Y_1),(X_2,Y_2)) \nonumber \\
=& \frac{1}{6} a_{12} b_{12} + \frac{1}{3} a_1 b_1 + \frac{1}{3} a_2 b_2 + \frac{1}{6} ab + \frac{1}{6} a_{12} b + \frac{1}{3} a_1 b_2  + \frac{1}{3} a_2 b_1 + \frac{1}{6} a b_{12} \nonumber \\
& - \frac{1}{6} a_{12} b_1 - \frac{1}{6} a_{12} b_2  - \frac{1}{6} a_1 b_{12} - \frac{1}{6} a_1 b_1 - \frac{1}{6} a_1 b_2 - \frac{1}{6} a_1 b \nonumber \\
& - \frac{1}{6} a_2 b_{12} - \frac{1}{6} a_2 b_1 - \frac{1}{6} a_2 b_2 - \frac{1}{6} a_2 b - \frac{1}{6} a b_1 - \frac{1}{6} a b_2 \nonumber \\
=& \frac{1}{6} ( a_{12} b_{12} + a_1 b_1 + a_2 b_2 + ab + a_{12} b + a_1 b_2  + a_2 b_1 + a b_{12} \nonumber \\
& - a_{12} b_1 - a_{12} b_2  - a_1 b_{12} - a_1 b - a_2 b_{12}  - a_2 b - a b_1 - a b_2) \nonumber \\
=& \frac{1}{6} ( a_{12} - a_1 - a_2 + a) ( b_{12} - b_1 - b_2 + b).
\end{align*}
Note that the above two are essentially \eqref{def: h1_2} and \eqref{def:h2_2} in Lemma \ref{lem:h1-h2-independent}.
As we have had $\mathbb{E}[h_2] = \mathbb{E}[h_4] = 0$ when $X$ and $Y$ are independent, we have $\mbox{Var}(h_2) = \mathbb{E}[h_2^2]$.
Let us compute $\mathbb{E}[( a_{12} - a_1 - a_2 + a)^2]$ first.
It is worth noting that
\begin{align*}
&\mathbb{E}[a_{12}^2] = \mathbb{E}[|X-X'|^2], \\
&\mathbb{E}[a^2] = \mathbb{E}[a_1a] = \mathbb{E}[a_2a] = \mathbb{E}[a_{12}a] = \mathbb{E}^2 [|X-X'|], \text{ and } \\
&\mathbb{E}[a_1^2] = \mathbb{E}[a_2^2] = \mathbb{E}[a_{12}a_1] = \mathbb{E}[a_{12}a_1] = \mathbb{E}[|X-X'| |X-X''|].
\end{align*}
As a result, we have
\begin{align*}
&\mathbb{E}[( a_{12} - a_1 - a_2 + a)^2] \\
=& \mathbb{E}[ a_{12}^2 + a_1^2 + a_2^2 + a^2 - 2 a_{12}a_1 - 2a_{12}a_2 + 2a_{12}a + 2a_1 a_2 - 2a_1 a - 2a_2 a ] \\
=& \mathbb{E}[|X-X'|^2] + 2 \mathbb{E}[|X-X'| |X-X''|] + \mathbb{E}^2 [|X-X'|] \\
& - 2 \mathbb{E}[|X-X'| |X-X''|] - 2 \mathbb{E}[|X-X'| |X-X''|] \\
& + 2 \mathbb{E}^2 [|X-X'|] + 2 \mathbb{E}^2 [|X-X'|] - 2 \mathbb{E}^2 [|X-X'|] - 2 \mathbb{E}^2 [|X-X'|] \\
=& \mathbb{E}[|X-X'|^2] - 2 \mathbb{E}[|X-X'| |X-X''|] + \mathbb{E}^2 [|X-X'|] = \mathcal{V}^2(X,X).
\end{align*}
Similarly, we have $\mathbb{E}[( b_{12} - b_1 - b_2 + b)^2]= \mathcal{V}^2(Y,Y)$.
In summary, we have
\begin{equation*}
\mbox{Var}(h_2) = \mathbb{E}[h_2^2] = \frac{1}{36} \mathcal{V}^2(X,X) \mathcal{V}^2(Y,Y),
\end{equation*}
which is \eqref{result:B} in Lemma \ref{lem:h1-h2-independent}.
\end{proof}

\subsection{Proof of Lemma \ref{lem:hXhY}}

\begin{proof}
By \cite[Lemma 12]{sejdinovic2013equivalence}, it is known that
$$
\tilde{k}(x,x') = |x-x_0| + |x'-x_0| - |x-x'|
$$
is a positive definite kernel.
Due to \cite[equation (4.4)]{sejdinovic2013equivalence}, we have the following:
\begin{align*}
\tilde{k}_P(x,x') &= \tilde{k}(x,x') + \mathbb{E}_{W,W'} \tilde{k}(W,W') - \mathbb{E}_{W'} \tilde{k}(x,W') - \mathbb{E}_{W} \tilde{k}(W,x') \\
&= |x-x_0| + |x'-x_0| - |x-x'| + \mathbb{E}_x|x-x_0| + \mathbb{E}_{x'}|x'-x_0| \\
& \hspace{0.2in} - \mathbb{E}_{x,x'}|x-x'| - |x-x_0| - \mathbb{E}_{x'}|x'-x_0| \\
& \hspace{0.2in} + \mathbb{E}_{x'}|x-x'| - \mathbb{E}_{x}|x-x_0| - |x'-x_0| + \mathbb{E}_{x}|x-x'| \\
&= - |x-x'|  - \mathbb{E}_{x,x'}|x-x'| + \mathbb{E}_{x'}|x-x'|  + \mathbb{E}_{x}|x-x'| \\
&= h_X(x,x')
\end{align*}
is also a positive definite kernel.
Similarly, $h_Y(Y_1,Y_2)$ is also a positive definite kernel.
\end{proof}

\subsection{Proof of Lemma \ref{lem:tensor-prod}}

\begin{proof}
Since $h_X$ is a positive definite kernel, by Mercer's Theorem, there exists a function sequence $\psi^X_1, \psi^X_1, \ldots$ and eigenvalues $\lambda^X_1 \ge \lambda^X_2 \ge \ldots \ge 0$ such that
$$
h_X(x,x') = \sum_{l=1}^\infty \lambda^X_l \psi^X_l(x) \psi^X_l(x'),
$$
where $\mathbb{E}[\psi^X_l(x)] = 0$, $\mathbb{E}[\psi^X_l(x)^2] = 1$ and $\mathbb{E}[\psi^X_l(x) \psi^X_{l'}(x)] = 0$ for $l \neq l'$. Similarly, we have
$$
h_Y(y,y') = \sum_{l=1}^\infty \lambda^Y_l \psi^Y_l(y) \psi^Y_l(y').
$$
By \cite{sejdinovic2013equivalence} equation (3.5), we that know
$$
h_2((X_1,Y_1),(X_2,Y_2)) = \frac{1}{6} h_X(X_1,X_2) h_Y(Y_1,Y_2)
$$
is a kernel with Reproducing Kernel Hilbert Space (RKHS) $\mathcal{H}$ isometrically isomorphic to the tensor product $\mathcal{H}_X \otimes \mathcal{H}_Y$.
Thus,
$$
6 h_2((X_1,Y_1),(X_2,Y_2)) = \sum_{l,l'=1}^\infty \lambda^X_l \lambda^Y_{l'} [\psi^X_l(X_1) \psi^Y_{l'}(Y_1)][\psi^X_l(X_2) \psi^Y_{l'}(Y_2)],
$$
which implies
$$
\{\lambda_1, \lambda_2, \ldots\} = \{\lambda^X_1, \lambda^X_2, \ldots\}
\otimes
\{\lambda^Y_1, \lambda^Y_2, \ldots\}.
$$
\end{proof}
\subsection{Proof of Corollary \ref{coro: eigen_sum}}
\begin{proof}
In this proof, we follow the notations in the proof of Lemma \ref{lem:tensor-prod}. It is worth noting that
$$
\sum_{l=1}^\infty \lambda^X_l = \mathbb{E}[h_X(x,x)] = \mathbb{E}_{x}[ - \mathbb{E}_{x,x'}|x-x'| + \mathbb{E}_{x'}|x-x'|  + \mathbb{E}_{x'}|x-x'| ] = \mathbb{E}[|X-X'|].
$$
As an immediate result of Lemma \ref{lem:tensor-prod}, we have
$$
\sum_{i=1}^\infty \lambda_i = \sum_{i=1}^\infty \lambda_i^X \sum_{i=1}^\infty \lambda_i^Y =  \mathbb{E}[|X-X'|]  \mathbb{E}[|Y-Y'|].
$$
Similarly, we verify that
$$
\sum_{l=1}^\infty (\lambda^X_l)^2 = \mathbb{E}[h_X(x,x')^2] = \mathcal{V}^2(X,X).
$$
Then, we have
$$
\sum_{i=1}^\infty \lambda_i^2 = \sum_{i=1}^\infty (\lambda^X_i)^2 \sum_{i=1}^\infty (\lambda^Y_i)^2 = \mathcal{V}^2(X,X) \mathcal{V}^2(Y,Y).
$$
\end{proof}

\subsection{Proof of Lemma \ref{lem:var-U-avg}}
\begin{proof}
By the law of total variance, we have
$$
\mbox{Var}(\overline{\Omega}_n) = \mathbb{E}_{U,V}[ \mbox{Var}_{X,Y}(\overline{\Omega}_n | U,V) ] + \mbox{Var}_{U,V}[ \mathbb{E}_{X,Y}(\overline{\Omega}_n | U,V) ].
$$
For the first term, when the random projections $U = (u_1,\ldots,u_K)$ and $V = (v_1,\ldots,v_K)$ are given, then by Lemma \ref{lem:var-U}, we have
$$
\mbox{Var}_{X,Y}(\overline{\Omega}_n | U,V) = \frac{16}{n} \mbox{Var}_{X,Y}(\bar{h}_1 | U,V) + \frac{72}{n^2} \mbox{Var}_{X,Y}(\bar{h}_2 | U,V)+ O\left(\frac{1}{n^3}\right),
$$
thus,
\begin{multline*}
\mathbb{E}_{U,V}[ \mbox{Var}_{X,Y}(\overline{\Omega}_n | U,V) ] = \frac{16}{n} \mathbb{E}_{U,V}[ \mbox{Var}_{X,Y}(\bar{h}_1 | U,V) ] \\+ \frac{72}{n^2} \mathbb{E}_{U,V}[ \mbox{Var}_{X,Y}(\bar{h}_2 | U,V) ]+ O\left(\frac{1}{n^3}\right).
\end{multline*}
For the second term, we have
$$
\mathbb{E}_{X,Y}(\overline{\Omega}_n | U,V) = \frac{1}{K} \sum_{k=1}^K \mathcal{V}^2(u_k^t X, v_k^t Y)
$$
thus, since $(u_k,v_k), k=1,\ldots,K$ are independent,
\begin{multline*}
 \mbox{Var}_{U,V}[ \mathbb{E}_{X,Y}(\overline{\Omega}_n | U,V) ] = \mbox{Var}_{U,V} \left( \frac{1}{K} \sum_{k=1}^K \mathcal{V}^2(u_k^t X, v_k^t Y) \right) \\= \frac{1}{K} \mbox{Var}_{u,v}( \mathcal{V}^2(u^t X, v^t Y) ),
\end{multline*}
where $(u,v)$ stands for random projection vectors from $\mbox{Unif}(\mathcal{S}^{p-1})$ and $\mbox{Unif}(\mathcal{S}^{q-1})$, respectively. In summary, the variance of $\overline{\Omega}_n$ is
\begin{eqnarray*}
\mbox{Var}(\overline{\Omega}_n) &=& \frac{1}{K} \mbox{Var}_{u,v}( \mathcal{V}^2(u^t X, v^t Y) ) +
\frac{16}{n} \mathbb{E}_{U,V}[ \mbox{Var}_{X,Y}(\bar{h}_1 | U,V) ] \\
&&+ \frac{72}{n^2} \mathbb{E}_{U,V}[ \mbox{Var}_{X,Y}(\bar{h}_2 | U,V) ]+ O\left(\frac{1}{n^3}\right).
\end{eqnarray*}
\end{proof}

\subsection{Proof of Theorem \ref{th:independence_test}}
\begin{proof}
For simplicity of notation, in this proof, without explicit statement, $\mbox{Var}(\cdot)$ and $\mbox{Cov}(\cdot)$ are with respect to $(X,Y)$. By the definition of $\bar{h}_2$, we have
$$
\mbox{Var}(\bar{h}_2 | U,V) = \frac{1}{K^2} \sum_{k,k'=1}^K \mbox{Cov}(h_2^{(k)},h_2^{(k')} | U,V).
$$
To simplify the notation, we define the following:
\begin{align*}
&a_{12}^u = |u^t(X_1 - X_2)|,                   &&b_{12}^v = |v^t(Y_1 - Y_2)|, \\
&a_{1}^u = \mathbb{E}[ |u^t(X_1 - X)| ], 	 &&b_{1}^v = \mathbb{E}[ |v^t(Y_1 - Y)| ], \\
&a_{2}^u = \mathbb{E}[ |u^t(X_2 - X)| ], 	 &&b_{2}^v = \mathbb{E}[ |v^t(Y_2 - Y)| ], \\
&a^u = \mathbb{E}[ |u^t(X - X')| ], \text{ and }&&b^v = \mathbb{E}[ |v^t(Y - Y')| ].
\end{align*}
Thus, by \eqref{def:h2_2}, we have
\begin{align*}
& \mbox{Cov}(h_2^{(k)},h_2^{(k')} | U,V) \\
=& \frac{C^2_p C^2_q }{36} \mathbb{E}_{X,Y} [(a_{12}^{u_k} - a_1^{u_k} - a_2^{u_k} + a^{u_k})(b_{12}^{v_k} - b_1^{v_k} - b_2^{v_k} + b^{v_k}) \\
& \hspace{0.2in} (a_{12}^{u_{k'}} - a_1^{u_{k'}} - a_2^{u_{k'}} + a^{u_{k'}})(b_{12}^{v_{k'}} - b_1^{v_{k'}} - b_2^{v_{k'}} + b^{v_{k'}})] \\
=& \frac{C^2_p C^2_q }{36} \mathbb{E}_{X,Y} [(a_{12}^{u_k} - a_1^{u_k} - a_2^{u_k} + a^{u_k})(a_{12}^{u_{k'}} - a_1^{u_{k'}} - a_2^{u_{k'}} + a^{u_{k'}})] \\
& \hspace{0.2in} \mathbb{E}_{X,Y}[ (b_{12}^{v_k} - b_1^{v_k} - b_2^{v_k} + b^{v_k})(b_{12}^{v_{k'}} - b_1^{v_{k'}} - b_2^{v_{k'}} + b^{v_{k'}})] \\
=& \frac{C^2_p C^2_q }{36} \mathcal{V}^2(u_k^t X, u_{k'}^t X) \mathcal{V}^2(v_k^t Y, v_{k'}^t Y),
\end{align*}
where the second equation holds by the assumption that $X$ and $Y$ are independent and the last equation holds by the definition of distance covariance in \eqref{eq:def_4}.

To summarize, the variance of $\overline{\Omega}_n$ with respect to $(X,Y)$ is
$$
\mbox{Var}(\overline{\Omega}_n | U,V) = \frac{2C^2_p C^2_q}{n^2} \frac{1}{K^2} \sum_{k,k'=1}^K \mathcal{V}^2(u_k^t X, u_{k'}^t X) \mathcal{V}^2(v_k^t Y, v_{k'}^t Y) + O(\frac{1}{n^3}),
$$
which implies
$$
\sum_{i=1}^\infty \bar{\lambda}_i^2 = 36\mbox{Var}(\bar{h}_2 | U,V) = \frac{C^2_p C^2_q}{K^2} \sum_{k,k'=1}^K \mathcal{V}^2(u_k^t X, u_{k'}^t X) \mathcal{V}^2(v_k^t Y, v_{k'}^t Y).
$$
By Corollary \ref{coro: eigen_sum}, we know that
$$
\sum_{i=1}^\infty \bar{\lambda}_i = \mathbb{E}[6\bar{h}_4(x,x)] = \frac{C_pC_q}{K} \sum_{k=1}^K \mathbb{E}[|u_k^t(X-X')|] \mathbb{E}[|v_k^t(Y-Y')|].
$$
\end{proof}

\subsection{Proof of Proposition \ref{prop:approx}}
\begin{proof}
Let us recall the definition,
$$
\sum_{i=1}^\infty \bar{\lambda}_i = \mathbb{E}[6\bar{h}_4(x,x)] = \frac{C_pC_q}{K} \sum_{k=1}^K \mathbb{E}[|u_k^t(X-X')|] \mathbb{E}[|v_k^t(Y-Y')|],
$$
$$
\sum_{i=1}^\infty \bar{\lambda}_i^2 = \frac{C^2_p C^2_q}{K^2} \sum_{k,k'=1}^K \mathcal{V}^2(u_k^t X, u_{k'}^t X) \mathcal{V}^2(v_k^t Y, v_{k'}^t Y).
$$
To estimate $\sum_{i=1}^\infty \bar{\lambda}_i^2$, we can use
$$
\frac{C^2_p C^2_q}{K^2} \sum_{k,k'=1}^K \Omega_n(u_k^t X, u_{k'}^t X) \Omega_n(v_k^t Y, v_{k'}^t Y),
$$
which takes $O(K^2 n \log n)$ time and is costly when $K$ is large.
It is worth noting that if $k \neq k'$ and $(u_k,v_k)$ is independent of $(u_{k'},v_{k'})$, by Lemma \ref{lemma:2}, we know that
$$
C^2_p C^2_q \mathbb{E}_{U,V}[ \mathcal{V}^2(u_k^t X, u_{k'}^t X) \mathcal{V}^2(v_k^t Y, v_{k'}^t Y) ] = \mathcal{V}^2(X,X) \mathcal{V}^2(Y,Y).
$$
Thus, $\sum_{i=1}^\infty \bar{\lambda}_i^2$ could be estimated by
$$
\frac{K-1}{K} \Omega_n(X,X) \Omega_n(Y,Y) + \frac{C_p^2 C_q^2}{K} \sum_{k=1}^K  \Omega_n(u_k^t X, u_{k}^t X) \Omega_n(v_k^t Y, v_{k}^t Y),
$$
which takes only $O(K n \log n)$ time.

And, $\sum_{i=1}^\infty \bar{\lambda}_i$ could be estimated by:
\begin{align*}
\frac{C_pC_q}{Kn^2(n-1)^2} \sum_{k=1}^K a_{\cdot \cdot}^{u_k} b_{\cdot \cdot}^{v_k},
\end{align*}
where
$$
a_{\cdot \cdot}^{u_k} = \sum_{i,j=1}^n |u_k^t(X_i - X_j)| \text{ and } b_{\cdot \cdot}^{v_k} = \sum_{i,j=1}^n |v_k^t(Y_i - Y_j)|.
$$
So, in summary, we have
\begin{align*}
\sum_{i=1}^\infty \bar{\lambda}_i & \approx \frac{C_pC_q}{Kn^2(n-1)^2} \sum_{k=1}^K a_{\cdot \cdot}^{u_k} b_{\cdot \cdot}^{v_k}, \\
\sum_{i=1}^\infty \bar{\lambda}_i^2 & \approx \frac{K-1}{K} \Omega_n(X,X) \Omega_n(Y,Y) + \frac{C_p^2 C_q^2}{K} \sum_{k=1}^K  \Omega_n(u_k^t X, u_{k}^t X) \Omega_n(v_k^t Y, v_{k}^t Y).
\end{align*}
\end{proof}



\end{document}